\documentclass[a4paper,10pt]{article}
 
\usepackage{amsmath,amssymb,amsthm}
\usepackage{xspace}
\usepackage[english]{babel}
\usepackage{hyperref}
\usepackage{color}
\usepackage[T1]{fontenc}
\usepackage[utf8]{inputenc}
\usepackage[margin=1in]{geometry}
\usepackage{bm}
\usepackage{float}
\floatstyle{boxed}
\newfloat{pseudocode}{thb}{pseudo}
\usepackage{enumitem}

\usepackage{amsmath,amsfonts,amssymb,amsthm}
\usepackage{bbold}
\usepackage{mathtools}

\usepackage{imakeidx}
\usepackage{multirow}
\usepackage{booktabs}

\usepackage{thmtools,thm-restate}

\usepackage{mathtools}
\mathtoolsset{showonlyrefs}

\usepackage
    {todonotes}
    
\newcommand{\isam}[1]{\todo[color=yellow]{}}
\newcommand{\isa}[1]{\todo[inline,color=yellow]{}}

\newcommand{\andym}[1]{\todo[color=purple]{}}
\newcommand{\andy}[1]{\todo[inline,color=purple]{}}

\newcommand{\fram}[1]{\todo[color=green]{}}
\newcommand{\fra}[1]{\todo[inline,color=green]{}}

\newcommand{\lucab}[1]{\todo[inline,color=red]{}}

\newcommand{\R}{\mathbb{R}}

\newtheorem{definition}{Definition}[section]
\newtheorem{lemma}[definition]{Lemma}
\newtheorem{claim}[definition]{Claim}
\newtheorem{fact}[definition]{Fact}
\newtheorem{theorem}[definition]{Theorem}

\newtheorem*{theorem*}{Theorem}
\newtheorem*{lemma*}{Lemma}

\renewcommand{\le}{\leqslant}
\renewcommand{\leq}{\leqslant}
\renewcommand{\ge}{\geqslant}
\renewcommand{\geq}{\geqslant}

\newcommand{\calG}{\mathcal{G}}

\newcommand{\bigO}{\mathcal{O}}
\newcommand{\polylog}{\mathrm{polylog}}

\newcommand{\SDG}{\cal{SDG}}
\newcommand{\SDGE}{\cal{SDGR}}

\newcommand{\PDGE}{\cal{PDGR}}
\newcommand{\PDG}{\cal{PDG}}

\newcommand{\Prc}[1]{\mathbf{Pr} \left( #1 \right)}
\newcommand{\Expcc}[1]{\mathbf{E} \left[ #1 \right]}

\title{Expansion and Flooding in Dynamic Random Networks with Node Churn\thanks{LT's work  on this project has received funding from the European Research Council
(ERC) under the European Union's Horizon 2020 research and innovation programme (grant agreement No.
834861). LB's work on this project was partially supported by the ERC Advanced Grant 788893 AMDROMA,
the EC H2020RIA project ``SoBigData++'' (871042),
the MIUR PRIN project ALGADIMAR.}}

\author{Luca Becchetti\\
    {\footnotesize{}Sapienza Università di Roma}\\
    {\footnotesize{} Rome, Italy}\\
    {\footnotesize{}\texttt{becchetti@dis.uniroma1.it}} 
    \and Andrea Clementi \\ 
    {\footnotesize{}Università di Roma Tor Vergata}\\
    {\footnotesize{} Rome, Italy}\\
    {\footnotesize{}\texttt{clementi@mat.uniroma2.it}} 
    \and Francesco Pasquale \\
    {\footnotesize{}Università di Roma Tor Vergata}\\
    {\footnotesize{} Rome, Italy}\\
    {\footnotesize{}\texttt{pasquale@mat.uniroma2.it}} 
    \and Luca Trevisan \\
    {\footnotesize{}Università Bocconi}\\
    {\footnotesize{} Milan, Italy}\\
    {\footnotesize{}\texttt{l.trevisan@unibocconi.it}} 
    \and Isabella Ziccardi\\
    {\footnotesize{}Università dell'Aquila}\\
    {\footnotesize{}L'Aquila, Italy}\\
    {\footnotesize{}\texttt{isabella.ziccardi@graduate.univaq.it}} 
}
\date{}

\begin{document}

\maketitle

\begin{abstract}
We study expansion and information diffusion in dynamic networks, that is in networks in
which nodes and edges are continuously created and destroyed. We consider information diffusion by {\em flooding}, the process by which, once a node is informed, it broadcasts its information to all its neighbors.

We study models in which the network is {\em sparse}, meaning that it has $\bigO(n)$ edges, where $n$ is the number of nodes, and in which edges are created randomly, rather than according to a carefully designed distributed algorithm.  In our models, when a node is ``born'', it connects to $d=\bigO(1)$ random other nodes. An edge remains alive as long as both its endpoints do.

If no further edge creation takes place, we show that, although the network will have $\Omega_d(n)$ isolated nodes, it is possible, with large constant probability, to inform a $1-exp(-\Omega(d))$ fraction of nodes in $\bigO(\log n)$ time. Furthermore, the graph exhibits, at any given time, a ``large-set expansion'' property.

We also consider  models with {\em edge regeneration}, in which if an edge $(v,w)$ chosen by $v$ at birth goes down because of the death of $w$, the edge is replaced by a fresh random edge $(v,z)$. In models with edge regeneration, we prove that the network is, with high probability, a vertex expander at any given time, and flooding takes $\bigO(\log n)$ time.

The above results hold both for a simple but artificial streaming model of node churn, in which at each time step one node is born and the oldest node dies, and in a more realistic continuous-time  model in which the time between births is Poisson and the lifetime of each node follows an exponential distribution.

Previous work on expansion and flooding    studied models in which either the vertex set is fixed and only edges change with time or models in which edge generation occurs according to an algorithm. Our motivation for studying models with random edge generation is to go in the direction of models that may eventually capture the formation of social networks or peer-to-peer networks.

\end{abstract}

\newpage
\setcounter{page}{1}

\section{Introduction}

We study {\em information diffusion} in {\em dynamic networks}. 
We focus on {\em flooding}, the information diffusion process whereby each node, once informed, spreads the information to all its neighbors.

By {\em dynamic networks} we mean communication networks that change 
over time, in which nodes enter and leave the networks, and links 
between nodes are created and destroyed. Several networks in which 
information diffusion is of interest, such as social networks and 
peer-to-peer network, exhibit change over time. 

Information diffusion in dynamic networks has been the focus of extensive previous work, 
surveyed in Section \ref{sec:related}. We are interested in models that 
exhibit {\em node churn} (that is, in which nodes enter and exit the 
network over time) and in which edge creation occurs  randomly, rather 
than being controlled by a sophisticated distributed algorithm. 
Our motivation is that a satisfactory modeling of network formation 
in social networks and peer-to-peer networks will have to satisfy both 
characteristics. As far as we are aware, information diffusion in 
dynamic networks with node churn and with uniformly random edge 
generation has not been studied before.

We made all other modeling choices as simple as possible, and we 
defined models with as few parameters as possible, in order to 
highlight qualitative features that we believe to be robust to 
different modelling choices. While our models are too simplified to predict all 
properties of realistic networks, one of our models (the Poisson model 
with edge regeneration that will be defined below) bears a certain 
resemblance of the way peer-to-peer networks such as bitcoin are formed. 

We study models in which the network is {\em sparse}, meaning that it 
has $\bigO(n)$ edges, where $n$ is the number of nodes. Specifically, 
when a node is ``born,'' it connects to $d=\bigO(1)$ random other 
nodes. We show that these dynamic random graphs maintain interesting 
expansion properties and that flooding informs all or most nodes 
(depending on details of the model) in $\bigO(\log n)$ time.

\subsection{Modeling networks that change with time}

To specify a dynamic network model we have  to specify how nodes enter 
and exit the network, and how edges are generated and destroyed.

\paragraph{Modeling node churn.}
We initially study an unrealistic but very simple model of node churn: 
at each discrete time unit, one node enters the network, and each node 
is alive for precisely $n$ time units. We refer to this as a { \em 
streaming} model of node churn. After the first $n$ time units, the 
network has always exactly $n$ nodes, and precisely one node is born 
and one node dies in each time unit. We then study a more realistic 
continous-time model, in which the number of births within each time 
unit follows a Poisson distribution with mean $\lambda$, and the lifetime of each node is 
independently distributed as an exponential distribution with parameter 
$\mu$, so that the average lifetime of a node is $1/\mu$ and the 
average number of nodes in the network at any given time is 
$\lambda/\mu$.  In order to reduce the number of parameters, we  
assume that the time that it takes to send a message along an edge is 
the same, or is of the same order as, the typical time between node 
births, which is 1 in the streaming model and $\lambda$ in the Poisson 
model. In order to have a consistent notation in the two models, we 
choose  time units in the continuous model such that $\lambda=1$, and 
we call $n=1/\mu$. With these conventions, the results that we prove in 
the streaming model are also true in this continuous time model, 
suggesting that the results have a certain robustness and that the 
streaming model, despite its simplicity, has some predictive power on 
the behavior of more realistic models.

\paragraph{Modeling edge creation and destruction.} When a node enters the
network, we assume that it connects to $d=\bigO(1)$ nodes chosen uniformly at
random among those currently in the network. Once an edge $(u,v)$ is created,
it remains active as long as both $u$ and $v$ are alive. We study two models:
one {\em without edge regeneration} and one {\em with edge regeneration}. In
the former edges are created only when a new node joins the network, in the
latter a node creates its outgoing edges not only when it joins the network,
but also every time it looses an outgoing edge due to one of its neighbors
leaving the network, in order to keep its out-degree always equal to $d$.

Although the assumption that a node can pick its neighbors uniformly at random
among all nodes of the network is unrealistic in many scenarios, the edge
creation and regeneration processes in our models resembles the way in which
some unstructured peer-to-peer networks maintain a ``random'' topology. For
example, each \textit{full-node} of the Bitcoin network running the Bitcoin
Core implementation has a ``target out-degree value'' and a ``maximum in-degree
value'' (respectively $8$ and $125$, in the default configuration) and it
locally stores a large list of (ip addresses of) ``active'' nodes. Such list is
initially started with nodes received in response to queries to some DNS seeds.
Every time the number of current neighbors of a full-node is below the
configured target value it tries to create new connections with nodes sampled
from its list.  The list stored by a full-node is periodically advertised to
its neighbors and updated with the lists advertised by the neighbors. Hence, in
the long run each full-node samples its out-neighbors from a list formed by a
``sufficiently random'' subset of all the nodes of the network.

\bigskip

\subsection{Results and techniques}

\subsubsection{Informing most nodes in the models without edge regeneration}\label{sssec:intro-oursnoedge}

In the models without edge regeneration, we prove that, with high 
probability, at any given time, there are $\Omega_d(n)$ isolated 
vertices in the network. A vertex $v$ becomes isolated if all the $d$ 
edges created at birth were to nodes that have meanwhile died, and $v$ 
was  never been chosen as neighbor by younger nodes. Because of the 
presence of such isolated nodes, broadcasting a message to all nodes is 
not possible, or at  least it takes at least $\Omega_d(n)$ time in the 
streaming model and  $\Omega_d(n\log n)$ time in the Poisson model. 
Furthermore,  there is a constant probability that a broadcast dies out 
after reaching only $\bigO(1)$ vertices. There is, however, also a 
large constant probability  (that tends to 1 as $d\to \infty$ as 
$1-exp(-\Omega(d))$) that a broadcast will reach, say, $90\%$ of the 
nodes (in general, a constant fraction that tends to 1 as $d\to \infty$ 
as $1-exp(-\Omega(d))$) after $\bigO(\log n)$ time.

To prove this fast convergence we establish two results. One is that, 
in $\bigO(\log n)$ time, a broadcast reaches at least, say, $n/10$ 
nodes. To prove this, we argue that, while the number of informed nodes 
is less than $n/10$, there is a good probability that the number of 
informed nodes grows by a constant factor at each step (and the 
probability that the above condition fails after exactly $t$ steps 
decreases exponentially with $t$, so that we can take a union bound 
over all $t$). The basic idea of this proof is to apply the principle 
of deferred decision to the $d$ edges chosen by each vertex, and assume 
that those edges are chosen after the vertex is informed, so that the 
``frontier'' of newly informed vertices keeps growing. There are two 
difficulties with this approach. One is that older nodes are likely to 
have chosen neighbors that have meanwhile died, and so older nodes are 
unlikely to significantly contribute to the number of nodes that will 
be newly informed at the next step. The second difficulty is that a 
node may become informed by a message coming from one of the $d$ 
neighbors chosen at birth, so that we cannot really apply deferred 
decision in the way that we would like. 

To overcome these difficulties, we only consider nodes that are 
informed through special kinds of paths from the source node (this will 
undercount the number of informed nodes and make our result true for a 
stronger reason). Specifically, we define an ``onion-skin'' process 
that only considers paths that alternate between ``young'' nodes whose 
age is less than the median age and ``old'' nodes whose age is more 
than the median age. Furthermore,  this  process arbitrarily splits the 
$d$ edges chosen by each node at birth into $d/2$ ``type-A'' edges and 
$d/2$ ``type-B'' edges, and   only considers paths that, besides 
alternating between young nodes and old nodes, also alternate between 
type-A edges and type-B edges. With this restrictions and conventions 
in place, we can study what happens for every pair of consecutive steps 
by applying deferred decision.

As sketched above, we are able to show that we inform at least $n/10$ 
nodes after $\bigO(\log n)$ steps. To complete the argument, we show 
that, if $d$ is a sufficiently large constant, all sets of at least 
$n/10$ vertices have constant vertex expansion, which leads to 
informing at least $.9 n$ nodes after another $\bigO(1)$ steps. Above, 
$1/10$ can be replaced by $exp(-\Omega(d))$. This tradeoff is best 
possible because, as argued above, there are $\Omega_d(n)$ isolated 
vertices that we will not be able to inform.

\subsubsection{Informing all nodes in the models with edge regeneration}

In the model with edge regeneration, we show that the graph has, with high probability, constant vertex expansion at each time step. 
Despite the presence of node churn, this  implies that broadcast reaches all nodes in $\bigO(\log n)$ steps.

In the streaming model, the proof of vertex expansion is similar to how expansion is generally proved in random graphs: we bound the probability that a fixed set of $k$ vertices fails to have constant vertex expansion, then we take a union bound by multiplying by ${n \choose k}$ and then by summing over $k$. The only difficulty is in characterizing the probability that an edge exists between a pair of vertices $u,v$, because such probability is a non-trivial function of the age of $u$ and $v$. Then, since in the streaming model the node churn   is  limited and deterministic, we can easily exploit  the vertex expansion to derive the  logarithmic bound on the flooding time.

The analysis becomes considerably more technical in the Poisson model. The main difficulty is that, in order to compute the probability that an edge $(u,v)$ exists, we need to know the age of $u$ and $v$ and so we have to take a union bound over all subsets of vertices of all possible ages. But, at any given time, there are nodes of age up to $n\log n$, and so we end up with ${n\log n \choose k}$ cases in our union bound for sets of size $k$, while the probability that one such set is non-expanding is as high as $1/{n \choose k}^{\bigO(1)}$ for sets that contain mostly young vertices. The point is that most of the  ${n\log n \choose k}$ possible ways of choosing $k$ nodes of all possible ages involve choices of several old nodes, which are unlikely to have all survived. In order to carefully account for the ``demographics'' of all possible sets of edges in our union bound, we look at the logarithm of the probability that a certain set fails to expand, interpret is as the KL divergence of two appropriately defined distributions, and then use inequalities about KL divergence.
Moreover, some more technical care is required in the Poisson model to apply the above expansion property for bounding the flooding time. Indeed, the flooding analysis needs to cope with the presence of a random number of node insertion/deletions during every 1-hop message transmission.

\subsection{Summary and roadmap}
We consider four dynamic graph models, each corresponding to different 
choices as regards the process modelling node churn and the change in 
topology induced by departure of a node. For the former, we 
study both a streaming model of node churn and a more elaborate, 
continuous Poisson model. As for topology dynamics, we consider both 
the case in which a node's departure simply determines failure of all 
incident edges, and a model in which all nodes maintain a constant 
degree, thus regenerating new edges to compensate for the loss of 
edges shared with nodes that died in the interim.
Table \ref{table} summarizes our positive and negative results and refers to 
the formal statements of results in theorems and lemmas below.

\begin{table}[]
\centering
\footnotesize
\begin{tabular}{|c|c|c|c|}
\hline
\multicolumn{2}{|c|}{\multirow{2}{*}{}} & \multicolumn{2}{c|}{\textbf{Poisson/Streaming dynamic graphs}} \\ \cline{3-4} 
\multicolumn{2}{|c|}{} & \textbf{without Edge Regeneration} & \textbf{with Edge Regeneration} \\ \hline
\multirow{2}{*}{\begin{tabular}[c]{@{}c@{}}\\{\bf Expansion}\\ {\bf properties}\end{tabular}} & Negative Results & \begin{tabular}[c]{@{}c@{}}There is a constant fraction\\ of isolated nodes (w.h.p.)\\ \\ Streaming: Lemma \ref{lem:isolated_nodes}\\ Poisson: Lemma \ref{lem:isolated_nodes_poisson}\end{tabular} & --- \\ \cline{2-4} 
 & Positive Results & \begin{tabular}[c]{@{}c@{}}$\Theta(1)$-Expansion of big-size node subsets  (w.h.p.)\\ \\ Streaming: Lemma \ref{lem:exp_large_sub_SDG}\\ Poisson: Lemma \ref{lem:exp_large_subset_PDG}\end{tabular} & \begin{tabular}[c]{@{}c@{}}$\Theta(1)$-Expansion (w.h.p.)\\ \\ Streaming: Theorem \ref{thm:expansion-stream}\\ Poisson: Theorem \ref{thm:exp:pdge}\end{tabular} \\ \hline
\multirow{2}{*}{\begin{tabular}[c]{@{}c@{}}\\ \\  {\bf Flooding}\\ \end{tabular}} & Negative Results & \begin{tabular}[c]{@{}c@{}}Flooding may not complete, \\ with    probability $\Theta_d(1)$ \\ \\ Streaming: Theorem \ref{thm:not_flooding_purestreaming}\\ Poisson: Theorem \ref{lem:flooding_not_terminate_poisson}\end{tabular} & --- \\ \cline{2-4} 
 & Positive Results & \begin{tabular}[c]{@{}c@{}}Flooding informs  a fraction \\ $1-exp(-\Omega(d))$  of the nodes in $\bigO(\log n)$ time, \\ with   probability $1-\Omega(exp(-d))$\\ \\ Streaming: Theorem \ref{apx:thm:aeflooding}\\ Poisson: Theorem \ref{thm:flood_poiss_noreg}\end{tabular} & \begin{tabular}[c]{@{}c@{}} Flooding time is \\   $\bigO(\log n)$    (w.h.p.)\\ \\ Streaming: Theorem \ref{thm:SDGE-flooding}\\ Poisson: Theorem \ref{thm:flooding_terminates_poisson}\end{tabular} \\ \hline
\end{tabular}
\caption{Summary of our results.\label{table}}
\end{table}
The paper is organized as follows.
We provide more details on previous work in Section \ref{sec:related}.
In Section~\ref{sec:streaming},   we define  streaming models of dynamic graphs, we state our results on the convergence of the flooding process, and we provide an overview of the proofs of such results. Our main technical contribution is  
the  analysis via the onion-skin process   in the model without edge regeneration, which  is given  in Subsection \ref{subse:stream_without_flood}.  
In Section \ref{sec:Poisson}, we define the Poisson models of dynamic graphs, we state our results on the flooding process and provide an overview of the proofs. The main technical result, presented in  Subsection \ref{sssec:pdge-expansion}, is to establish vertex expansion for the model with edge regeneration, using a notion of edge subset ``demographics,'' which is quantified via KL divergence.  
Section   \ref{sec:conlc} provides some further overall remarks about our contribution and  poses an open question.
  To highlight our major contributions and to keep Sections \ref{sec:streaming} and \ref{sec:Poisson} reasonably short, all the omitted proofs of such sections are given in Sections \ref{apx:ssec:stream} and \ref{apx:sec:Poisson}, respectively.
Finally,   some mathematical tools 
  we used in the analysis are given in the Appendix.

\section{Related Work} \label{sec:related}

 A first, rough classification of dynamic graphs  can be made according to an important feature:
 whether or not the set of nodes keeps the same along all the graph process.
 In the affirmative case, we have an \emph{edge-dynamic graph} $\{G_t=(V,E_t),\, t \geq 0 \}$ where  the topology dynamics 
 defines the way the edges of a fixed set $V$ of participant nodes change over time. For this class of dynamic graphs,
 several
 models, such as worst-case adversarial changes \cite{KLO10,KO11,M16} and Markovian evolving graphs \cite{CMMPS08,CMPS11},
 have been introduced, their basic connectivity properties have been derived, and, fundamental distributed tasks, such as broadcast and consensus, have been rigorously analyzed.

 In contrast, much less analytical works are currently available when (even) 
 the set of participant nodes can change over time. This 
 class of dynamic graphs $\{G_t=(V_t,E_t),\, t \geq 0 \}$ are often called \emph{dynamic networks with churn} \cite{augustine2016distributed}: in this framework,
 the specific graph dynamics describe both the node insertion/deletion rule for 
 the time sequence $V_t$ and the edge updating rule for the time sequence $E_t$.
 The number of nodes that can join or leave the network at every round is called \emph{churn rate}.
 For brevity's sake, in what follows we will only describe    those previous analytical results on 
 dynamic networks with churn which are related to the models we studied in this paper. In particular, we mainly focus
 on previous work where some  connectivity properties of a  dynamic networks with churn have been rigorously proved.
 
 As remarked in the Introduction, to the best of our knowledge, previous analytical studies focus on   distributed algorithms that are suitably designed to maintain topologies having  good connectivity properties.
 
   Pandurangan et al.~\cite{pandurangan2003building} introduced  a partially-distributed 
protocol that constructs and maintains a bounded-degree graph   which relies on a centralized cache of a constant number of nodes. In more detail, their protocol ensures the network is connected,
has logarithmic diameter, and has always bounded degree. The protocol manages a central cache which maintains a  subset of the current set of vertices. When joining the network, a new node chooses a constant number of nodes in the cache. The insertion/deletion procedures for the central cache follows rather complex rules which   take $\bigO(\log n)$ overhead and delays, w.h.p.

In \cite{DuchonD14}, Duchon et al   presented ad-hoc protocols that  
  maintain a given distribution of random graphs under an arbitrary sequence of vertex insertions and deletions. 
   More in detail, given that the graph $G_t$ is  random uniform    over  the  set of 
   $k$-out-degree  graphs   with $n$ nodes, they provide  suitable  distributed  randomized protocols that
  can  insert (respectively  delete) a node  such that the 
graph $G_{t+1}$ at round $t$ is again  random uniform over the  set of 
   $k$-out-degree  graphs   with $n+1$ (respectively,  $n - 1$) nodes.
 They do not   assume a centralized knowledge of the whole graph but, instead,
 their protocol relies on some  random primitives to 
sample arbitrary-sized subsets of nodes  uniformly at random. For instance,     
once a new node $u$ is inserted,    a random subset of nodes is selected (thanks to 
one of such  centralized primitives),
and each of them   is forced to delete one of its link and to deterministically connect to $u$. 
 The basic versions of their insertion$/$deletion procedures require each node to communicate
with nodes at distance 2, while their more refined version (achieving   optimal  performance) require communications over 
longer paths.

An important and effective approach to keep a dynamic graph with churn having 
good expansion properties is based on the use of ID random walks. Roughly speaking, 
this approach let 
every participating node start  $k$
independent random walks of tokens containing its ID  
and all the other nodes collaborate to perform such random walks for enough time so that
the token is well-mixed over the network. Once a token is mature, it can be used by any
node that, in that step,  needs a new edge by simply asking to connect to it. The probabilistic
analysis then typically shows two main, correlated invariants: on one hand, the edge set, arising from 
the above random-walk process, form a random graph having good expansion properties. On the other hand, after a small
number of steps, the random walks are well-mixed.

   Cooper et Al \cite{cooper2007sampling} consider two  deterministic churn processes: in the first one,
     at every round a new node is inserted while no nodes leave the network, while, in  the second process,
     the size $n$ of the graph does never change since, 
      at every round,  a new node is inserted  and the oldest node leaves the graph (this is in fact the streaming model  we study in this paper).  They
 provide  a protocol where each node $v$ starts $c\cdot m$ 
  independent random walks (containing the ID-label of $v$) until they are picked up, $m$ at a time, by new nodes joining the network.
 The new node connects to the $m$ peers that contributed the tokens it got. The resultant dynamic topology
 is shown to keep diameter $\bigO(\log n)$, and to be fault-tolerant against  adversarial deletion of both edges and vertices. We remark that the tokens in the graph must be constantly circulated in order to ensure that they are well-mixed. Moreover, the rate at which new nodes can join the system is limited, as they must wait while the existing tokens mix before they can use them.
 
  Law and Siu \cite{LS03} provide a distributed algorithm for maintaining
a regular expander in the presence of limited number of insertions/deletions.  The   algorithm is based on a complex procedure that is able to sample uniformly at random  from the space of all possible $2d$-regular graphs formed by $d$ Hamiltonian circuits over the current set of alive nodes. They present
possible distributed implementations of this sample procedure, the best of which, based on random walks,  have $\bigO(\log n)$ overhead and time delay.   Such solutions cannot   manage frequent node churn. 

Further distributed algorithms with different approaches achieving $\bigO(\log n)$ overhead and time delay in the case of slow node churn are proposed in \cite{AS04,RRSST09,MS04,PT14}.
    
    In \cite{APRRU15}, Augustine et al present an efficient randomized distributed protocol that guarantees the maintenance of a bounded degree  topology that, with high probability, contains an expander subgraph whose set of vertices has size $n-o(n)$, where $n$ is the stable network size. This property is preserved   despite the presence of a large oblivious adversarial churn rate
— up to $\bigO(n/ \polylog(n))$. In more detail,  considering the node churn adopted in \cite{APRU12}, i.e.,  an oblivious churn adversary that:   can remove any set of nodes up to the churn limit in every round, and,  at the same time, it should add (an equal amount of) nodes to the network with the following constraints.   A new node
should be connected to at least one existing node and  the number of new nodes added to an existing node
should not exceed a fixed constant (thus, all nodes have constant bounded degree).

The expander maintenance protocol is efficient  even though  it is rather complex and the local overhead for maintaining
the topology is   polylogarithmic in $n$.  A complication of the protocol follows from the fact that, in order to prevent the growth of large clusters of nodes outside the expander subgraph, it  uses   special criteria to ``refresh'' the links of some nodes, even when the latter have not been involved by 
any edge deletion due to the node churn.

   Recently, the flooding process has been analytically studied  over  dynamic graph models
   with churn in \cite{APRU12,augustine2016distributed}. Here, the authors consider the model analysed in \cite{APRRU15},  that we discussed above. Using  the   expansion property proved in \cite{APRRU15}, they show that, for any  fixed   churn rate $C(n) \leq n/ \polylog n$ managed by an oblivious worst-case adversary,  there is  a set $S$ of size  $n - \bigO(C(n))$ of nodes such that, if a source node in $S$ starts the flooding   in round $t$,  then all except $\bigO(C(n)$ nodes  get informed within round $t + \bigO(\log(n/C(n)) \log n)$, w.h.p.

Our models are inspired by the way some unstructured P2P networks maintain a
``well-connected'' topology, despite nodes joining and leaving the network,
small average degree and almost fully decentralized network formation. For example, after an initial bootstrap in which they
rely on DNS seeds for node discovery, full-nodes of the Bitcoin
network~\cite{nakamoto2008bitcoin} running the Bitcoin Core implementation
turn to a fully-decentralized policy to regenerate their neighbors when their
degree drops below the configured threshold~\cite{bitcoincoreP2P}. This allows
them to pick new neighbors essentially at random among all nodes of the
network~\cite{bitnodes}. Notice also that the real topology of the Bitcoin
network is hidden by the network formation protocol and discovering the real
network structure has been recently an active subject of
investigations~\cite{delgado2019txprobe,neudecker2016timing}.

\section{Warm-up: Preliminaries and  the Streaming Model} \label{sec:streaming}

We first recall the notion of vertex expansion of a graph. 

\begin{definition}[Vertex expansion]\label{def:expander} 
The \emph{vertex isoperimetric number} $h_{out}(G)$ of a graph $G = (N,E)$ is
\[
h_{out}(G)=\min_{0 \leq |S| \leq |N|/2} \frac{|\partial_{out}(S)|}{|S|}\,,
\]
where we used $\partial_{out}(S)$ for the outer boundary of $S$
\[
\partial_{out}(S) = \{ v \in N \setminus S \,:\, \{u,v\} \in E \mbox{ for some }
u \in S\}\,.
\]
Given a constant $\varepsilon > 0$, a graph $G$ is a \emph{(vertex)
$\varepsilon$-expander} if $h_{out}(G) \geq \varepsilon$.
\end{definition}

A \textit{dynamic graph} $\mathcal{G}$ is a sequence of graphs $\mathcal{G} =
\{G_t = (N_t, E_t) \,:\, t \in \mathbb{N}\}$ where the sets of nodes and edges
can change at any discrete round. If they can change randomly we call the
corresponding random process a \textit{dynamic random graph}. We call $G_t$ the
\textit{snapshot} of the dynamic graph at round $t$. For a set of nodes $S
\subseteq N_t$, we denote with $\partial_{out}^{t}(S)$ the outer boundary $S$
in snapshot $G_t$; we omit superscript $t$ when it is clear from the context.

In this section we study two dynamic random graph models in which nodes join
and leave the network according to a \textit{deterministic streaming}, (see
Definition~\ref{def:streaming_node_churns}) and edges are created randomly by
nodes with low degree (see Definitions~\ref{def:pure_streaming_without}
and~\ref{def:pure_streaming_with}).

\begin{definition}[Streaming node churn]\label{def:streaming_node_churns}
The set of nodes $N_t$ evolves as follows: It starts with $N_0 = \emptyset$; At
each round $t \geq 1$ a new node joins the network and it stays in the network
for exactly $n$ rounds (i.e., node joining at round $t$ stays up to round $t +
n-1$), then it disappears.  We say that a node has \emph{age} $k$ at round $t$
if it joined the network at round $t-k$. We say that a node $u$ is \emph{older}
(respectively, \emph{younger}) than a node $v$ if $u$ joined the network before
(respectively, after) $v$.
\end{definition}

We are interested in estimating the time a message sent by a node takes to
reach all (or a large fraction of) the nodes. To this end, we formalize the
\textit{flooding process} over a dynamic (random) graphs.

\begin{definition}[Flooding]\label{def:flooding_streaming}
Let $\calG = \{G_t = (N_t, E_t) \,:\, t \in \mathbb{N}\}$ be a dynamic (random)
graph. The flooding process over $\calG$ starting at time $t_0$ from the source node $v_0 \in N_{t_0}$ 
is the sequence of (random) sets of
nodes $\{ I_t \,:\, t \in \mathbb{N} \}$ where, $I_t = \emptyset$ for all $t < t_0$, 
$I_{t_0} = \{v_0\}$ (in this paper we will assume that $I_0$
contains the node joining the network at round $t_0$) and, for every $t
\geqslant t_0$, $I_{t}$ contains all nodes in $N_{t}$ that were neighbor of
some node in $I_{t-1}$ in the snapshot $G_{t-1}$, i.e.,
\[
I_{t} = (I_{t-1} \cup \partial_{out}^{t-1}(I_{t-1}))\cap N_{t} \,.
\]
We say that $I_t$ is the subset of \emph{informed} nodes at round $t$.  We say
that the flooding \emph{completes} the broadcast if a round $t$ exists such
that $I_t \supseteq N_{t-1} \cap N_t$ and, in this case, the number of rounds
$t-t_0$ is the \emph{flooding time} of the source message.
\end{definition}

\subsection{Streaming graphs without edge regeneration} \label{ssec:purestreamwithout}
In this section, we study the streaming model \SDG\  where edges are 
created only when a new node joins the network. We first show that, for 
constant $d$ and for any given round, the corresponding random snapshot of the dynamic graph has a linear 
fraction of isolated nodes, w.h.p. Moreover, we show that flooding fails with constant 
probability. On the other hand, this model still affords a weaker 
notion of epidemic process. In particular, in Subsection 
\ref{subse:stream_without_flood}, we show that, with constant 
probability, flooding can still inform a large, constant fraction of 
the nodes within a time interval of size $\bigO(\log n)$.

\begin{definition}[Streaming graphs without edge
regeneration]\label{def:pure_streaming_without}
A \emph{Streaming Dynamic Graph} (for short, \SDG) $\mathcal{G}(n,d)$ is a
dynamic random graph $\{G_t = (N_t,E_t)\,:\, t \in \mathbb{N} \}$ where the set
of nodes $N_t$ evolves according to Definition~\ref{def:streaming_node_churns},
while the set of edges $E_t$ according to the following \textit{topology
dynamics}
\begin{enumerate}[noitemsep]
\item When a new node appears, it creates $d$ independent connections, each one
with a node chosen uniformly at random among the nodes in the network.
\item When a node dies, all its incident edges disappear.
\end{enumerate}
\end{definition}

\noindent
\textbf{Remark.}
 The
considered graphs are always undirected. However,  given  any active node $v$,
our analysis will need to distinguish between \emph{out-edges} from $v$,
i.e., those requested by $v$, and the \emph{in-edges}, i.e., the ones due
to the  requests from other nodes and accepted by $v$.

\paragraph{Preliminary properties.}
  It is possible to prove that the expected degree of each 
node in  a snapshot $G_t=(V_t,E_t)$ of a \SDG\ $\mathcal{G}(n,d)$ is 
$d$   (see Lemma 
\ref{lem:exp_degree_purestreaming} in Subsection \ref{ssec:lem:exp_degree_purestreaming}), for any $t \geq n$. Thus, the expected number of edges in the  graph is  $nd/2$.

It is well-known    that  a static random graph in which each node 
chooses  $d$ random neighbors is a $\Theta(1)$-expander, w.h.p., for any choice of 
the parameter  $d \geq 3$ (see Lemma \ref{lemma:staticdoutgoing} in Section \ref{ssec:lemma:staticdoutgoing} of the Appendix). 
In the next lemma we instead show that this is not the case for the 
\SDG \  model: w.h.p., there can be a linear fraction of isolated nodes 
at every time steps. Informally speaking, this fact is essentially due to 
the presence of ``older'' nodes that have good chance to see all their 
out-edges disappear and, at the same time, to get no in-edges from 
younger nodes. The formal argument (which is given in    Subsection \ref{ssec:lem:isolated_nodes}) to prove this intuitive fact  requires the 
use of the method of bounded difference to manage the correlations 
among the random variables, each one  indicating whether a given node 
gets isolated or not.

\begin{lemma}[Isolated nodes]
\label{lem:isolated_nodes}
For every positive constant $d$ and for every sufficiently large $n$, let $\{G_t = (N_t, E_t) \,:\, t \in
\mathbb{N}\}$ be an \SDG\ sampled from $\mathcal{G}(n,d)$. For every fixed $t \geqslant n$,  
w.h.p. the number of isolated nodes in $G_t$ is at least   $\frac{1}{6}n 
e^{-2d}$. Moreover, w.h.p., each of these nodes will remain isolated across its
entire lifetime.
\end{lemma}

\subsubsection{Expansion properties}\label{subse:poiss_without_exp}

As Lemma \ref{lem:isolated_nodes} suggests, we have no generalized 
expansion properties in the \SDG\ model. Still, we can prove a weaker 
expansion property, which only applies to sufficiently large subsets of the 
vertices. This property is crucial in proving our positive result about 
flooding in this model and it is stated in Lemma \ref{lem:exp_large_sub_SDG} below.

\begin{lemma}[Expansion of large subsets]
\label{lem:exp_large_sub_SDG}
For every constant $d \geq 20$ and for every sufficiently large $n$, let $\{G_t = (N_t, E_t) \,:\, t \in
\mathbb{N}\}$ be an \SDG\ sampled from $\mathcal{G}(n,d)$. For 
every fixed $t \geq n$, w.h.p. the snapshot $G_t$ satisfies the following:
\[
\min_{S \subseteq N_t \, : \, ne^{-d/10} \leq |S| \leq
n/2}\frac{|\partial_{out}(S)|}{|S|} \geq 0.1\,.
\]
\end{lemma}

The  proof's idea of the above property  is to show that any two disjoint sets $S,T \subseteq N_t$, with $ne^{-d/10} \leq |S| \leq n/2$ and $|T|=0.1|S|$, such that  $\partial_{out}(S) \subseteq T$, exist with negligible probability and then apply a union bound over all possible pairs $S,T \subseteq N_t$. To get the first fact, we derive  a suitable   argument that takes care about the younger/older relationship between any pair of nodes.   We then exploit the fact that  any node
$u$ has probability $1/n$ to send a request to   a node $v$ older than him. The full proof is given in Subsection \ref{ssec:lem:exp_large_sub_SDG}.

\subsubsection{Flooding}\label{subse:stream_without_flood}
We begin with a negative result about flooding in  the \SDG\ model.   We recall that Lemma \ref{lem:isolated_nodes}  shows the existence of a linear fraction of   nodes that keep isolated for all their respective lifetime. This is the key-ingredient in proving the following fact (the full proof  is given in Subsection \ref{ssec:thm:not_flooding_purestreaming}).

\begin{theorem}[Flooding]
\label{thm:not_flooding_purestreaming}
For every positive constant $d$, for every sufficiently large $n$, and for every fixed $t_0 \geqslant n$, the flooding process over an \SDG\ sampled from $\mathcal{G}(n,d)$ starting at $t_0$ satisfies the following two statements:
\begin{enumerate}[noitemsep]
\item With probability $\Omega(e^{-d^2})$, for every $t \geq t_0$, $I_t$ contains at most $d+1$ nodes;
\item W.h.p. The flooding time is $\Omega_d(n)$.
\end{enumerate}
\end{theorem}

 On the other hand, there is  a large constant probability   that a broadcast will reach a large fraction of nodes  within $\bigO(\log n)$ time.

\begin{theorem}[Flooding completes for a large fraction  of nodes] \label{apx:thm:aeflooding}
For every constant $d>200$, for every sufficiently large $n$ and for every fixed 
$t_0\ge n$, there is a $\tau=\bigO(\log n/\log d + d)$, such that the 
flooding process over an \SDG\ sampled from 
$\mathcal{G}(n,d)$ starting at $t_0$ satisfies the following:
\[
	\Prc{|I_{t_0+\tau}| \geq (1-e^{-d/10})n} 
	\geq 1-4e^{-d/100}-o(1) \, ,
\]
\end{theorem}

As remarked in Subsection \ref{sssec:intro-oursnoedge}, the proof of the above result is one of our  major technical contributions: for this reason, in what follows,    we provide its   description.

\subsubsection*{Proof of Theorem \ref{apx:thm:aeflooding}}
The proof consists of two steps. Assuming the source node $s$ joined the 
network in round $t_0$, we first show (Lemma 
\ref{lem:flooding_terminates_part_1}) that, with probability at least $1-4e^{-d/100}$ , a restriction of 
the true topology dynamics establishes a bipartite graph which i) 
contains $s$, ii) only connects nodes with ages in the interval 
$\{1,\ldots , n/2\}$ to nodes with ages in the
interval $\{n/2 + 1,\ldots , n-\log n\}$, iii) has diameter $\bigO(\log n)$, 
iv) includes at least $2n/d$ nodes. This is enough to prove that, with 
probability $1-4e^{-d/100}$, $2n/d$ nodes are informed at time $t_0 + \tau_1$, where 
$\tau_1 = \bigO(\log n)$.

The second step consists in showing (Lemma \ref{le:big_set_exp}) that, 
thanks to the expansion properties established in Lemma 
\ref{lem:exp_large_sub_SDG}, once $2n/d$ nodes have been informed, at 
least $(1-e^{-d/10})n$ nodes will become informed within a constant 
number $\tau_2=\Theta(d)$ of additional steps, w.h.p.

Overall, the above two steps prove that within time $t_0 + \tau_1 + \tau_2$, at least $(1-e^{-d/10})n$  nodes have been informed, with probability at least $1-4e^{-d/100}-o(1)$. We begin with the first part, corresponding to the following lemma.
\begin{lemma}[Flooding completes for a large fraction of nodes, phase 1]
\label{lem:flooding_terminates_part_1} 
Under the hypotheses of Theorem~\ref{apx:thm:aeflooding}, there is a
$\tau_1=\bigO(\log n/\log d)$ such that
\begin{equation}
\label{eq:lem:flo_alm_term}
    \Prc{|I_{t_0+\tau_1}|\geq \frac{2n}{d}} \geq 1 - 4e^{-\frac{d}{100}} \,.
\end{equation}
\end{lemma}
\begin{proof}
We begin by defining the following subsets of $N_{t_0}$:
\begin{itemize}
\item the set of the \emph{young nodes}: $    Y=\{v \in N_{t_0} \ | \ v \hbox{ has life $l$ with } 2 \leq l <\frac{n}{2}\}$
\item the set of the \emph{old nodes}: $
O=\{ v \in N_{t_0} \ | \ v \hbox{ has life $l$ with } \frac{n}{2} \leq  l \leq n-\log n\}$
\item the set of the \emph{very old nodes}:
$    \hat{O} = N_{t_0}-(Y\cup O)=\{v \in N_{t_0} \ | \ v \hbox{ has life $l$ with } n-\log n < l \leq n\}$
\end{itemize}

To prove \eqref{eq:lem:flo_alm_term} we show that 
$G_{t_0}=(N_{t_0},E_{t_0})$ contains a bipartite subgraph with logarithmic 
diameter, containing the informed node $s$ and such that i) links are 
established only between nodes in $Y$ and in $O$ and ii) it contains 
no very old node. The graph in question is the result of the 
\emph{onion-skin} process described below.

\paragraph{The onion-skin process.} The iterative process we consider 
operates in phases, each consisting of two steps. Starting from $s$, 
the onion-skin process builds a connected, bipartite graph, 
corresponding to alternating paths in which young nodes only connect
to old ones. In particular, each realization of this process 
generates a subset of the edges generated by the original topology 
dynamics. 
Moreover, each iteration of the process corresponds to a partial 
flooding in the original graph, in which a new layer of informed nodes 
is added to the subset of already informed ones, hence the term onion-skin. Flooding is partial since i) the 
network uses a subset of the edges that would be present in the 
original graph. 

In the following, we 
denote by $Y_k\subseteq Y$ and $O_k\subseteq O$ the subsets of young 
and old nodes that are informed by the end of phase $k$, respectively. In the remainder, we let 
$O_{-1} = \emptyset$ for notational convenience.

\begin{center}
\fbox{
\begin{minipage}{15cm}
\textbf{Onion-skin process}
	\begin{description}
		\item[Phase $\mathbf{0}$:] $Y_0 = \{s\}$; $O_0$ is obtained as follows:
				$s$ establishes $d$ links. We let $O_0\subset O$ denote the 
				subset of old nodes that are destinations of these links. Links with 
				endpoints in $Y$ or $\hat{O}$ are discarded;
		\item[Phase $\mathbf{k\ge 1}$:] $Y_k$ and $O_{k}$ are 
		iteratively obtained as follows:
		
		\smallskip
		\emph{Step 1}. Each node in $Y - Y_{k-1}$ 
				establishes $d/2$ links. More precisely: 
				\begin{equation}\label{def:Y_k}
					Y_k-Y_{k-1}=\left\{v \in Y-Y_{k-1} \ | \ v \hbox{ 
					connects to $O_{k-1}$ by a request $i \in \{\frac{d}{2}+1,\dots,d$\}} \right\}
				\end{equation}
				Links to nodes not belonging to $O$ are discarded;
				
		\smallskip
		\emph{Step 2}. Each node in $Y_k - Y_{k-1}$ establishes $d/2$ links to 
			nodes in $O - O_{k-1}$. More precisely:
			\begin{equation}\label{def:O_k}
				O_k-O_{k-1}=\left\{v \in O-O_{k-1} \ | \hbox{ some $w \in Y_k$ 
				connects to $v$ by a request $i \in \{1,\dots, \frac{d}{2}\}$}\right\}
			\end{equation}
			Links to nodes not belonging to $O$ are discarded.
	\end{description}
\end{minipage}\label{proc:onion_skin_stream}
}
\end{center}
A couple remarks are in order. It is clear that the links in $E_{t_0}$ 
can be established in any order, as long as they are created from 
younger nodes towards older ones. As a consequence, each realization of 
the onion-skin process produces a subset of $E_{t_0}$. In particular, 
i) nodes in $O$ and $\hat{O}$ do not create any links, though they can still 
be the targets of links originating from $Y$; ii) a node $v\in Y$ 
released at time $\hat{t}$ ($\le t_0$) creates $d$ links, with possible 
destinations the nodes released in the interval $[\hat{t}, t_0]$, but 
only links with destinations in $O$ are retained, the others are 
discarded.

The next claim states  that, at each step, the sets of informed nodes $Y_k
\subseteq Y$ and $O_k \subseteq O$ grow by a constant factor $d/20$. It
analyzes Phase 0 and the generic Phase $k$ separately and it is proved in
Subsection \ref{ssec:claim:O_0}.

\begin{claim}
The following holds for Phase 0, 
\label{claim:O_0}\begin{equation}
    \Prc{|O_0| \ge \frac{d}{20}}\geq 1-e^{-d/100}.
\end{equation} 
In the  generic phase $k \geq 1$,
if $|Y_{k-1}| \leq n/d$ and $|O_{k-1}| \leq n/d$, 
\begin{align*}
\Prc{|Y_k-Y_{k-1}| > \frac{d}{20} y \mid |O_{k-1}-O_{k-2}| \geq y} \geq 
1-e^{-yd/100}  
\end{align*}

\begin{equation}
\Prc{|O_{k}-O_{k-1}| \geq \frac{d}{20}x \mid |Y_k-Y_{k-1}| \geq x} \geq 1-e^{-dx/100} \, .
\end{equation}
\end{claim}

Then, from the above     claim  and using    the chain rule, 
we get that, for each $k \geq 0$,
\begin{align}
&\Prc{|O_k-O_{k-1}| \geq a_{2k+1}} \geq \prod_{i=0}^{2k}\left(1-e^{-a_i(d/100)}\right) \hbox{ and } \Prc{|Y_k-Y_{k-1}| \geq a_{2k} } \geq \prod_{i=0}^{2k}\left(1-e^{-a_i(d/100)}\right)\,,
\end{align}
where $a_k=\left(\frac{d}{20}\right)^k$ and
as long as $a_{2k}$ and $a_{2k+1}$ are smaller than $n/d$.
Then,  after some $\tau_1=\log n/\log d$ rounds, we get    $|Y_{t_0+\tau_1}| \geq n/d$ and $|O_{t_0+\tau_1}|\geq n/d$, with probability at least
\begin{align}
c=\prod_{i=0}^{\infty}\left(1-e^{-a_i(d/100)}\right)\,.
\end{align}
In Subsection \ref{ssec:claimprod}, using standard calculus, we prove the following 
claim, which concludes the proof.

\begin{claim}
\label{claim:product}
For each $d>200$, if $a_i=(d/20)^i$,
\[
	c = \prod_{i=0}^{\infty}\left(1-e^{-a_i(d/100)}\right) \ge 1 - 4e^{-\frac{d}{100}}.
\]
\end{claim}
\end{proof}

\begin{lemma}[Flooding completes for a large fraction of nodes, phase 2]\label{le:big_set_exp}
Under the hypotheses of Theorem \ref{apx:thm:aeflooding}, a constant
$\tau_2=\Theta(d)$ exists such that, for $\tau_1=\bigO(\log n/\log d)$ (as in Lemma 
\ref{lem:flooding_terminates_part_1}) we have:
\begin{equation}
    \Prc{|I_{t_0+\tau_1+\tau_2}| \geq (1-e^{-d/10})n} \geq 1-4e^{-d/100}-o(1)\,.
\end{equation}
\end{lemma}

The proof of the above lemma, which is given in Subsection
\ref{ssec:le:big_set_exp}, heavily relies on the expansion properties of large
subsets proven in Lemma  \ref{lem:exp_large_sub_SDG}. In more detail, we first observe
that Lemma \ref{lem:flooding_terminates_part_1} implies
$|I_{t_0+\tau_1}|\geq 2n/d$.  We can then 
inductively  apply the expansion property stated by Lemma  \ref{lem:exp_large_sub_SDG} to the set of   informed nodes
$I_t$, for each $t \geq t_0+\tau_1$, until the size of this subset becomes $n/2$.
After that, the expansion property and the consequent inductive argument is
instead applied to the set of non-informed nodes.  The process ends when the
size of the set of non-informed nodes falls below $ \leq ne^{-d/10}$, since at
that point we can no longer apply Lemma \ref{lem:exp_large_sub_SDG}. We notice that, in the whole proof,
the oldest $\tau_2$ nodes in $N_{t_0+\tau_1}$ are never considered, since they
all die within the next $\tau_2$ steps.

\subsection{Streaming   graphs with edge regeneration} \label{ssec:streaming-with}
We now consider a variant of the streaming dynamic graph model where a node
creates its outgoing links not only when it joins the network, but also every
time it looses an outgoing link due to one of its neighbors leaving the
network. In this model, at every round $t$ the snapshot $G_t$ is a sparse
random graph having exactly $d n$ edges.

\begin{definition}[Streaming graphs with edge regeneration]\label{def:pure_streaming_with}
A \emph{Streaming Dynamic Graph with edge Regeneration} (for short, \SDGE) 
$\mathcal{G}(n,d)$ is a dynamic random graph $\{G_t = (N_t,E_t)\,:\, t \in \mathbb{N} \}$ 
where the set
of nodes $N_t$ evolves according to Definition~\ref{def:streaming_node_churns},
while the set of edges $E_t$ evolves according to the following \textit{topology
dynamics}: 
\begin{enumerate}[noitemsep]
\item When a new node appears, it creates $d$ independent connections, each one
with a node chosen uniformly at random among the nodes in the network.
\item When a node dies, all its incident edges disappear.
\item When a node has one of its   $d$ outgoing edges disappearing, it creates a new connection
with a node chosen uniformly at random among all nodes in the network.
\end{enumerate}
\end{definition}

\paragraph{Preliminary properties.}
We will prove that the streaming model with edge regeneration yields snapshots
having good vertex expansion. To derive the expansion properties we first
prove a bound on the edge probability. Informally, we show that, despite the
presence of nodes of different ages, making the edge distribution non uniform,
the probability that a fixed node chooses any other active node in the network
is still $\bigO(1/n)$.

For constant $d \geq 20$ and for sufficiently large $n$, let $\{G_t = (N_t, E_t) \,:\, t \in
\mathbb{N}\}$ be an \SDG\ sampled from $\mathcal{G}(n,d)$. For 
every fixed $t \geq n$, w.h.p. the snapshot $G_t$

\begin{lemma}\label{lem:node_destination}
  For every $d \geq 1$ and for every sufficiently large $n$, let $\{G_t = (N_t, E_t) \,:\, t \in
\mathbb{N}\}$ be an \SDGE\ sampled from $\mathcal{G}(n,d)$. For 
every fixed $t \geq n$, consider  the snapshot $G_t$.  Let $k \leq t-1$ and let  $u$ be the  node 
having age $k+1$. Then, if another node $v$ in $N_t$ is born before $u$, the probability that a single request of $u$ has destination $v$ is
\begin{equation} \label{eq:edgeprob1}
    \frac{1}{n-1} \left( 1+\frac{1}{n-1} \right)^k \,,
\end{equation}
while,  if $v$ is born after $u$, the probability that a single request of $u$ has destination $v$ is always $\leq \frac{1}{n-1}$.
\end{lemma}

The almost-uniformity of the destination distribution stated in the lemma above
is essentially due to the fact that, in the streaming model,  every node node
has lifetime $n$ and, hence, it has at most $n$ chances to be chosen as
destination along the regeneration process. In formula, this yields, in the
worst case, the extra factor $(1 + 1/n)^{\bigO(n)}$ in \eqref{eq:edgeprob1}. The
full proof of the lemma is given in Subsection \ref{ssec:lem:node_destination}.

\subsubsection{Expansion properties}

In this subsection, we show that, for a sufficiently large constant $d$, the streaming model with edge regeneration yields snapshots having good vertex expansion.

For constant $d \geq 20$ and for sufficiently large $n$, let $\{G_t = (N_t, E_t) \,:\, t \in
\mathbb{N}\}$ be an \SDG\ sampled from $\mathcal{G}(n,d)$. For 
every fixed $t \geq n$, w.h.p. the snapshot $G_t$ 
 
 \begin{theorem}[Expansion] \label{thm:expansion-stream}
 For every $d \geq 14$ and  
 for every sufficiently large $n$, let $\{G_t = (N_t, E_t) \,:\, t \in
\mathbb{N}\}$ be an \SDGE\ sampled from $\mathcal{G}(n,d)$. Then, w.h.p.,    
  for  every fixed $t \geq n$,   the snapshot $G_t$ 
 is an   $\varepsilon$-expander with parameter $\varepsilon \geq 0.1$.
 \end{theorem}

The full proof of the above result is given in Subsection \ref{ssec:lem:expansion_small_streaming}, while an   overview is given below.
The proof is divided into two parts: the expansion for the big-size sets (with size in the range $[n/4,n/2]$) and the expansion for small-size subsets (with size in the range $[1,n/4]$). As for the first case, the analysis  is identical to that of Lemma \ref{lem:exp_large_sub_SDG} for the  \SDG\ model.  To analyze the expansion of   small-size subsets,  we show that, for every pair of vertex subset $S$, with $|S| \leq n/4$ and $T$, with $S \cap T = \emptyset$ and $|T|=0.1|S|$, the event ``all the out-neighbors of $S$ are in $T$'', i.e. $A_{S,T}=\{\partial_{out}(S)\subseteq T\}$, does happen with negligible probability.  
  To  give an   upper bound on $\Prc{A_{S,T}}$, we   observe that $A_{S,T}$  is bounded by the event that each link request of every  node in $S$ must have destination in  $S \cup T$.  Thanks to Lemma \ref{lem:node_destination}, for any pair of subset $S$ and $T$, we can derive the following bound
\begin{equation}
    \Prc{A_{S,T}} \leq \left(\frac{e}{n-1}\cdot |S \cup T|\right)^{d|S|}\,.
\end{equation}
Since $|S| \leq n/4$, using standard calculus, we show the above equation offers a sufficiently small bound. The theorem then follows from an union bound over all possible pairs $S,T \subseteq N_{t}$.

\subsubsection{Flooding}
An important consequence of the 
expansion property we prove in Theorem \ref{thm:expansion-stream}  is that the flooding process over the \SDGE\ model is fast and reaches all nodes of the network.    The proof of this fact for this streaming model is a simple adaptation of  
the expansion argument which is  typically used in   dynamic graph models with no node churn (see, for example,  \cite{CMPS11}).  The deterministic and limited  node churn   has in fact  a negligible impact in the analysis, only. The  proof is given in  Subsection \ref{ssec:thm:SDGE-flooding}.

\begin{theorem}[Flooding] \label{thm:SDGE-flooding}
For    every  $d \geq 21$, for every sufficiently large $n$, and for every fixed $t_0 \geq n$, w.h.p. the flooding process over an \SDGE\ sampled from $\mathcal{G}(n,d)$ starting at $t_0$ has completion time $ \bigO(\log n)$.
\end{theorem}

\section{The Poisson Model} \label{sec:Poisson}

A \textit{continuous dynamic graph} $\mathcal{G}$ is a continuous family of
graphs $\mathcal{G} = \{G_t = (N_t, E_t) \,:\, t \in \mathbb{R}^+\}$ where the
sets of nodes and edges can change at any time $t \in \mathbb{R}^+$. As in the
discrete case, we call $G_t$ the \textit{snapshot} of the dynamic graph at time
$t$ and, for a set of nodes $S \subseteq N_t$, we denote with
$\partial_{out}^{t}(S)$ the outer boundary of $S$ in snapshot $G_t$ and we omit
superscript $t$ when it is clear from context.

In this section we study expansion properties and flooding over two continuous-time dynamic graph models in which
nodes' arrivals follow a Poisson process and their lifetimes obey an
exponential distribution. 

\begin{definition}[Poisson node churn]\label{def:Poisson_node_churns}
Initially $N_0 = \emptyset$. Node arrivals in $N_t$ follow
a Poisson process with mean $\lambda$. Moreover, once a node joins the 
network, its lifetime  has exponential distribution with parameter $\mu$.
\end{definition}

While the definition of flooding is straightforward in the discrete case
(Definition~\ref{def:flooding_streaming}), where we assume that the sets of
nodes and edges can change and all the neighbors of an informed node gets
informed \textit{in one unit of time}, in the continuous case we need to
specify how the time it takes a message to flow from a node to its neighbors
and the changes in the topology of the graph relate to each other. Since we
want to preserve in the model the fact that a message takes one unit of time to
flow from an informed node to its neighbors, the most natural way to define the
flooding process in a continuous setting would be the following
``asynchronous'' version.


\begin{definition} (``Asynchronous'' Flooding)\label{def.async.flood}
Let $\calG = \{G_t = (N_t, E_t) \,:\, t \in \mathbb{R}^+\}$ be a dynamic
(random) graph. The flooding process over $\calG$ starting at time $t_0$ from
vertex $v_0 \in N_{t_0}$ is the sequence of (random) sets of nodes $\{ I_t
\,:\, t \in \R^+ \}$ where, $I_t = \emptyset$ for all $t < t_0$, $I_{t_0} =
\{v_0\}$ (in this paper we will thus assume that $I_0$ contains the node joining the network at round $t_0$) and, for every $t \geqslant t_0$, $I_{t}$ contains
all nodes in $N_{t}$ that were neighbor of some node in $I_{t-1}$ in the snapshot
$G_{t-1}$, in addition to all previously informed nodes
\[
I_{t} = 
\left(  
\left( \bigcup_{t'<t} I_{t'} \right) \cup \partial_{out}^{t-1}(I_{t-1})
\right) 
\cap N_{t} \, .
\]
We say that the nodes in $I_t$ are \emph{informed} at time $t$.  We say that
the flooding \emph{completes} the broadcast if a time $t$ exists such that $I_t
\supseteq N_t$, in this case the time $t-t_0$ is the \emph{flooding time} of
the source message.
\end{definition}

In order to analyze the process of Definition \ref{def.async.flood}, it will be convenient to define the  discretized process below, in which nodes are informed only at discrete times. 

\begin{definition} (``Discretized'' Flooding) \label{def:flooding_poisson}
Let $\calG = \{G_t = (N_t, E_t) \,:\, t \in \mathbb{R}^+\}$ be a continuous
dynamic (random) graph. The flooding process over $\calG$ starting at time $t_0
\in \R^+$ from vertex $v_0 \in N_{t_0}$ is the sequence of (random) sets of
nodes $\{ I_t \,:\, t \in \mathbb{N} \}$ where, $I_t = \emptyset$ for all $t <
t_0$, $I_{ t_0} = \{v_0\} $ and, for
every $t$ of the form $t_0 + m$ with integer $m$, $I_{t}$
contains all nodes in $I_{t-1}$ that did not die in the time interval $(t-1,
t)$ and all nodes in $N_{t}$ that have been neighbor of some node in $I_{t-1}$
for the whole time interval $(t-1, t)$:
\[
I_{t} = 
\left(  
I_{t-1} \cup \partial_{out}^{t-1}(I_{t-1} \cap N_t)
\right) 
\cap N_{t} \, .
\]
We say that the nodes in $I_t$ are \emph{informed} at round $t$. We say that
the flooding \emph{completes} the broadcast if a round $t$ exists such that
$I_t \supseteq N_t$, in this case the time $t- t_0$ is the
\emph{flooding time} of the source message.
\end{definition}

The discretized process, which is artificial and is defined only for the purpose of the analysis, can be thought of as the asynchronous modified in such a way that an informed node waits until a discrete time before sending messages, Thus, the convergence of the discretized flooding can only be slower than the convergence of the asynchronous flooding, and any upper bound that we prove on the convergence time of the former will apply to the latter.

Our negative results, however, also apply to Definition \ref{def.async.flood}.

\subsection{Poisson node churning} \label{ssec:Poisson-Prely}

In this subsection, we present useful properties of Poisson dynamic
graphs that only depend on the random node churn process and therefore 
apply to both variants of the model, i.e., with and without edge regeneration. 

We remark that,  according to Definition \ref{def:Poisson_node_churns} 
above, the time interval between two consecutive node arrivals is an 
exponential random variable of parameter $\lambda$,
while the number of nodes joining the network in a time interval of 
duration $\tau$ is a Poisson random variable with expectation 
$\tau \cdot \lambda$. We finally note that the stochastic continuous 
process $\{N_t:t \in \mathbb{R}^+\}$  is clearly a continuous Markov Process.

A first important fact our analysis relies on is that  we can
 bound the number of active nodes at every time. In particular, it is easy to show that  
 $\Expcc{|N_t|} \rightarrow \lambda/\mu$ and, moreover, we have the following bound in concentration.

\begin{lemma}[Pandurangan et al. \cite{pandurangan2003building} - Number of nodes in the network]
\label{thm:concentration_nodes}
For every pair of parameters $\lambda$ and $\mu$ such that 
$n=\lambda/\mu$ is sufficiently large, 
consider the  Poisson node churn $\{N_t \,:\, t \in \mathbb{R}^+ \}$       in Definition~\ref{def:Poisson_node_churns}.
Then, for     every fixed    real $t \geq 3n$, w.h.p. $|N_t|=\Theta(n)$ and, more
precisely,
\begin{equation}
\Prc{0.9n \leq |N_t| \leq 1.1n}
\geq 1-2e^{-\sqrt{n}} \,.
\end{equation}
\end{lemma}

Leveraging Lemma \ref{thm:concentration_nodes},  our  analysis of the Poisson 
considers the setting $\lambda=1$ without loss of generality. In the 
remainder, we define the key parameter $n=\frac{1}{\mu}$ representing the  ``expected'' size 
of the network. Moreover,   since the probability that two or more 
churn events occur at the same time is zero, the  points of change 
of the dynamic graph   yield a discrete-time sequence  of   
\emph{events}. In particular, we can observe and prove properties of the dynamic 
graph only when one event changing the graph occurs, namely, the 
arrival of a new node or the death of an existing one. 

\begin{definition}
\label{def:T_i}
Let $\{N_t \,:\, t \in \mathbb{R}^+ \}$ be a Poisson node churn  as in
Definition~\ref{def:Poisson_node_churns}.  We
define   the infinite sequence of random variables \emph{steps} (also called \emph{rounds}) $\{T_r: r \in \mathbb{N}\}$ (with parameters $\lambda$ and $\mu$) as follows:
\[ 
T_0 \, = \, 0 \quad \text{ and } \quad T_{r+1} \, = \, \inf\{t > T_r \, : \, N_t \neq N_{T_r} \, \} \, , \mbox{ for } r =  0, 1, 2, \ldots .
\]
\end{definition}

It is worth mentioning that, since the Poisson stochastic process $\{N_t:t \in \mathbb{R}^+\}$ is a countinuous Markov process, the above defined stochastic process 
$\{N_{T_r}: r \in \mathbb{N}\}$ consistently is a discrete Markov chain.

Thanks to Theorem \ref{thm:minimum_exp} in the Appendix, we can easily 
find the law of the random variables that define the time steps at
which new events occur. The proof of the next lemma is given in 
Subsection \ref{ssec:lem:prop_In}.

\begin{lemma}[Jump process]
\label{lem:prop_Tn} The stochastic process $\{N_{T_r} \,, r \in \mathbb{N}\}$ in Definition \ref{def:T_i} is a discrete Markov chain where, for every fixed integer $r \geq 0$ and for every integer $N \geq 0$, conditional to the event  ``$|N_{T_r}|=N$'', $T_{r+1}$ 
  is a  random variable of exponential distribution with parameter $N\mu + \lambda$. Moreover, 
\begin{align}
    \Prc{|N_{T_{r+1}}|=|N_{T_r}|-1 \mid |N_{T_r}|=N}= \frac{N\mu}{N\mu + \lambda}\,, \\
    \label{eq:bound_birth}
    \Prc{|N_{T_{r+1}}|=|N_{T_r}|+1 \mid |N_{T_r}|=N}= \frac{\lambda}{N\mu + \lambda}\,.
\end{align}
Finally, for every fixed  node $v \in N_{T_r}$, the probability that the decreasing of $N_{T_r}$ is due to the death of $v$ is 
 
\begin{equation}
    \Prc{v \not \in N_{T_{r+1}} \mid v \in N_{T_r}, |N_{T_r}|=N}= \frac{\mu}{N \mu + \lambda}\,.
\end{equation}
\end{lemma}

The  next lemma shows that the probability of the next event being a node 
arrival or death is close to $1/2$ since, for large enough $r$,  
$|N_t|$  is w.h.p. close to  $n$. The proof is deferred to Subsection 
\ref{ssec:lem:N_m+1_and_death}.

\begin{lemma} \label{lem:N_m+1_and_death}
For every sufficiently large $n$,   consider  the Markov chain  $\{N_{T_r} \,, r \in \mathbb{N}\}$     in  Definition \ref{def:T_i} with parameters $\lambda =1$ and $\mu = 1/n$.    Then, for  every fixed integer $r \geq n \log n$,  
\begin{align}
\label{eq:bound_death}
  0.47 \leq \Prc{|N_{T_{r+1}}|=|N_{T_r}|-1 } \leq 0.53 \quad \hbox{and} \quad  0.47 \leq \Prc{|N_{T_{r+1}}|=|N_{T_r}|+1 } \leq 0.53\,.
\end{align}
Moreover, if $v \in N_{T_{r}}$,
\begin{equation}
    \frac{1}{2.2n} \leq \Prc{v \not \in N_{T_{r+1}} \mid v \in N_{T_r}} \leq \frac{1}{1.8n} \, .
    \label{eq:v_not_in_T_r+1}
\end{equation}

\end{lemma}

The next lemma provides a useful bound on the   lifetime of any node in 
the network. The proof is given  in Subsection \ref{ssec:thm:life_of_nodes}.

\begin{lemma}[Lifetime of the nodes]
\label{lem:life_of_nodes} 

For every sufficiently large $n$   consider  the Markov chain  $\{N_{T_r} \,, r \in \mathbb{N}\}$     in  Definition \ref{def:T_i} with parameters $\lambda =1$ and $\mu = 1/n$.
Then, for  every fixed integer $r \geq 7 n \log n$,
 with probability at 
least $1-2/n^{2.1}$, each node in $N_{T_r}$ was born after step 
$T_{r-7n \log n}$.
\end{lemma}

\subsection{Poisson graphs without edge regeneration}
\label{ssec:poisson_without_reg}
We consider two variants of dynamic graphs with node churns governed by Poisson
processes that mirror the two dynamics in
Definitions~\ref{def:pure_streaming_without} and~\ref{def:pure_streaming_with}.
In this subsection, we consider the first variant, in which new edges are
created only when a new node joins the network.

\begin{definition}[Poisson dynamic graphs without edge regeneration]
\label{def:Poisson1}
A \emph{Poisson Dynamic Graph without edge regeneration} (for short, \PDG)
$\mathcal{G}(\lambda, \mu, d)$  is a continuous dynamic random graph $\{G_t =
(N_t,E_t)\,:\, t \in \mathbb{R}^+ \}$ where the set of nodes $N_t$ evolves
according to Definition~\ref{def:Poisson_node_churns}, while the set of edges
$E_t$ according to the following \textit{topology dynamics}:
\begin{enumerate}[noitemsep]
\item When a new node appears, it creates $d$ independent connections, each one
with a node chosen uniformly at random among the nodes in the network.
\item When a node dies, all its incident edges disappear.
\end{enumerate}
\end{definition}

\paragraph{Preliminary properties.}
Similarly to the streaming model, the Poisson model without edge regeneration
may result in the presence of a linear fraction of isolated nodes. The proof of
this negative result proceeds along the same lines as the case of the streaming
model (Lemma~\ref{lem:isolated_nodes}). In more detail, we leverage
Lemma~\ref{thm:concentration_nodes} and Lemma~\ref{lem:life_of_nodes} to
characterize the random churn. The full proof is given in
Subsection~\ref{ssec:lem:isolated_nodes_poisson}.

\begin{lemma}[Isolated nodes]
\label{lem:isolated_nodes_poisson}
For every positive constant $d$ and 
for every sufficiently large $n$, let $\{G_t = (N_t, E_t) \,:\, t \in
\mathbb{R}^+\}$ be a \PDG\ sampled from $\mathcal{G}(\lambda,\mu,d)$ with $\lambda=1$ and $\mu=1/n$. For every fixed integer $r \geqslant 7n \log n$, 
w.h.p. the number of isolated nodes in $G_{T_r}$ is at least   $\frac{1}{18}n 
e^{-2d}$. Moreover, w.h.p., each of these nodes will remain isolated across its
entire lifetime.
\end{lemma}

\subsubsection{Expansion properties}
The lemma that follows highlights weak expansion properties of the 
Poisson model without edge regeneration. In particular, we show that, for any sufficiently large 
$t$, all subsets of $N_t$ including a sufficiently large, constant 
fraction of the nodes exhibit good expansion properties. 

\begin{lemma}[Expansion of large subsets]
\label{lem:exp_large_subset_PDG}
For every constant $d \geq 20$ and for every sufficiently large $n$,  let $\{G_t = (N_t, E_t) \,:\, t \in
\mathbb{R}^+\}$ be a \PDG\ sampled from $\mathcal{G}(\lambda,\mu,d)$ with $\lambda=1$ and $\mu=1/n$. Then, for every fixed integer
$r=\Omega(n \log n)$, with probability at least $1-2/n^2$, the snapshot $G_{T_r}$ satisfies
\begin{equation}
    \min_{n e^{-d/20} \leq |S| \leq |N_{T_r}|/2} \frac{|\partial_{out}(S)|}{|S|}\geq 0.1 \, .
\end{equation}
\end{lemma}

The full 
proof of the above lemma is given in Subsection
\ref{ssec:lem:exp_large_subset_PDG} and it is based on the following idea.
For any fixed pair of subsets 
$S$ and $T$ with $|T|=0.1|S|$, we first observe that the event 
$A_{S,T}$ (defined as in the analogous proof in the \SDG \ model, Lemma 
\ref{lem:exp_large_sub_SDG}) implies that there is no outgoing link 
from $S$ to the subset  $P = N_{T_r}-S-T$, and, using Lemma 
\ref{thm:concentration_nodes},  we know that $|P| \geq 0.9n-1.1s$, 
w.h.p. By using a simple counting argument, we show that each of at 
least half of all the edges in the cut $E(S,T)$ have probability $\geq 
1/\Theta(n)$ to belong to $E_{T_r}$. Then, since $S$ and $P$ are   
large,  the number of such potential edges is large enough to apply a 
standard union bound to all possible subset pairs $S,T \subseteq N_{T_r}$.

\subsubsection{Flooding}

The negative result of Lemma \ref{lem:isolated_nodes_poisson} implies that
the flooding process has non-negligible chances to 
fail in rapidly informing the entire network.  
Its proof, which is given in Subsection 
\ref{ssec:lem:flooding_not_terminate_poisson}, relies on the presence 
of isolated nodes and uses the original Definition 
\ref{def.async.flood}, adapting the   argument we used  in the proof 
of Theorem \ref{thm:not_flooding_purestreaming}  (for the streaming model) to take care about the presence
of nodes having random lifetime that can be of length $\Theta(n\log 
n)$.

\begin{theorem}[Flooding]
\label{lem:flooding_not_terminate_poisson}
For every positive constant $d$, for every sufficiently large $n$ and 
for every fixed $r_0 \geqslant 7n \log n$, the flooding process over a 
\PDG\ sampled from $\mathcal{G}(\lambda,\mu,d)$ with $\lambda=1$, 
$\mu=1/n$ and starting at $t_0=T_{r_0}$ satisfies the following two 
properties:
\begin{enumerate}[noitemsep]
\item With probability $\Omega(e^{-d^2})$, for every $t \geq t_0$, $I_t$ contains at most $d+1$ nodes;
\item W.h.p., the flooding time is $\Omega_d(n)$.
\end{enumerate}
\end{theorem}

We next complement the negative results above by showing that, following the arrival of 
an informed node at some time $t$, a fraction $1 - e^{-\Omega(d)}$ of the vertices of the network
will become informed within the following $\bigO(\log n)$
flooding steps, with probability $1- e^{-\Omega(d)}$ . 

\begin{theorem}[Flooding completes for a large fraction of nodes] \label{thm:flood_poiss_noreg}
For every constant $d \geq 1152$, for every sufficiently large $n$ and 
for every fixed $r_0 \geq 7n \log n$, there is a $\tau = \bigO(\log 
n/\log d + d)$, such that the flooding process over a \PDG \ sampled 
from $\mathcal{G}(\lambda, \mu, d)$ with $\lambda=1$, $\mu=1/n$ 
and  starting at $t_0=T_{r_0}$, satisfies the following:
\[\Prc{|I_{t_0+\tau}| \geq (1 - e^{-\frac{d}{20}})|N_{t_0 + \tau}|} \geq 1 - 2e^{-\frac{d}{576}} - o(1)\,.\]
\end{theorem}

The proof of Theorem \ref{thm:flood_poiss_noreg} is presented in Subsection \ref{ssec:thm:flood_poiss_noreg} and
proceeds along lines similar to those of Section 
\ref{subse:stream_without_flood}, though with some important 
differences, which we briefly discuss below. It should be noted that, 
in order to account for the fact 
that a live node might die at any point of a given flooding interval, the proof of Theorem 
\ref{thm:flood_poiss_noreg} uses the discretized version of the 
flooding process described by Definition \ref{def:flooding_poisson}, 
which clearly provides a worst case scenario when we are interested in proving 
lower bounds on the extent and upper bounds on the speed of flooding.

We first show that, starting with an informed node $s$ joining the network 
at time $t_0$, with probability $1 - 2e^{-\frac{d}{576}} - o(1)$, at least $\frac{n}{10}$ nodes 
are informed at time $t_0 + \tau_1$, where 
$\tau_1 = \bigO(\log n/\log d)$. To prove a 
similar result in the streaming model, we considered a subset of the vertices, inducing a 
topology that remained unchanged within an interval of interest of 
logarithmic size. This way, the proof boiled down to proving diameter 
properties of this induced subgraph, which we did by introducing 
the onion-skin process. Unfortunately, this approach does not trivially 
carry over to the Poisson model, since every node that is in the 
network at any given time $t$ has some probability of dying within each time 
unit. To address this issue, we define a variant of the onion-skin 
process, in which i) flooding proceeds alongside edge creation, in the
sense that a node only establishes its links upon becoming informed 
(deferred decisions), ii) each newly informed node tosses a coin to 
decide whether or not it is going to die before time $t_0 + \tau_1$. In 
order to consider a worst-case scenario, if a node dies before time 
$t_0 + \tau_1$, the node leaves the network immediately upon being reached 
in the flooding process, without generating any links or informing 
any neighbours. We also leverage two facts, namely, since we are 
considering an overall flooding interval spanning a logarithmic number of steps, 
the number of nodes joining the network in this interval is itself 
at most logarithmic, while the probability of a node that is alive 
at time $t_0$ to die before time $t_0 + \bigO(\log n)$ is $\bigO(\log n/n)$. This 
allows us to prove that, with probability at least $1 - 
2e^{-\frac{d}{576}} - o(1)$, 
at least a constant fraction of the nodes are informed within 
time $t_0 + \tau_1$. The second step is similar to the case of the 
streaming model, leveraging expansion of large sets and, in particular, 
Lemma \ref{lem:exp_large_subset_PDG} above. In particular, \emph{if} at least 
$n/10$ are informed, Lemma \ref{lem:exp_large_subset_PDG} allows to show 
that, with high probability, a further, constant number $\tau_2 = 
\bigO(d)$ of flooding steps suffice to reach a fraction $1 - e^{-\frac{d}{20}}$ of the nodes. Choosing 
$\tau = \tau_1 + \tau_2$ allows to prove that within time $t_0 + 
\bigO(\log n/\log d + d)$, a fraction at least $1 - e^{-\frac{d}{20}}$ 
of the nodes is informed, with probability at 
least $1 - 2e^{-\frac{d}{576}} - o(1) - \frac{2}{n^2}$. 

\subsection{Poisson   graphs with edge regeneration } \label{ssec:poissonedgereg}

We now  model graph dynamics where an active node replaces
each of its $d$ outgoing edges that will be deleted. 

\begin{definition}[Poisson dynamic random graphs with edge regeneration]
\label{def:Poisson2}
A \emph{Poisson Dynamic Graph with edge Regeneration} $\mathcal{G}(\lambda,
\mu, d)$ (for short, \PDGE) is a continuous dynamic random graph $\{G_t =
(N_t,E_t)\,:\, t \in \mathbb{R}^+ \}$ where the set of nodes $N_t$ evolves
according to Definition~\ref{def:Poisson_node_churns}, while the set of edges
$E_t$ evolves according to the following \textit{topology dynamics}: 
\begin{enumerate}[noitemsep]
\item When a new node appears, it creates $d$ independent connections, each one
with a node chosen uniformly at random among the nodes in the network.
\item When a node dies, all its incident edges disappear.
\item When a node has one of its $d$ outgoing edges disappearing, it creates a new connection with a node chosen uniformly at random among all the nodes in the network.
\end{enumerate}
\end{definition}

\paragraph{Preliminary properties.}
Similarly to the approach we adopted for the streaming model, our first
technical step is to provide an upper bound on the probability  that a  fixed
node  chooses any other active node in the network as destination of one of its
$d$ requests. However, things in this setting get more complicated essentially
because of  the presence of ``very old'' nodes (i.e. those nodes having age
$\omega(n)$).  Indeed, such old nodes can be selected as destination of a link
request from  a  younger node with probability $\omega(1/n)$. The next lemma
formalizes this fact as function of the age of the nodes.
The  proof follows the same approach we used to get   Lemma \ref{lem:node_destination} and  it  is given in Subsection \ref{sssec:lem:existence_edge_hyp}.

\begin{lemma}
\label{lem:existence_edge_hyp}
For every constant $d \geq 20$ and for every sufficiently large $n$,  let $\{G_t = (N_t, E_t) \,:\, t \in
\mathbb{R}^+\}$ be a \PDGE\ sampled from $\mathcal{G}(\lambda,\mu,d)$ with $\lambda=1$ and $\mu=1/n$. Then, for every fixed integer
$r=\Omega(n \log n)$, consider the snapshot $G_{T_r}$. Let $u \in N_{T_r}$ be the node born in round $T_{r-i}$ for some integer $i \leq r$. Then, if another node $v \in N_{T_r}$ is born before $u$, the probability that a single request of $u$ has destination $v$ is at most
\begin{equation} \label{eq:uppedgeprob-poisson}
    \frac{1}{0.8n}\left(1+\frac{i}{1.7n}\right) \, .
\end{equation}
While, if $v$ is born after $u$, the probability that a single request of $u$ has destination $v$ is always $\leq \frac{1}{0.8n}$.
\end{lemma}

\subsubsection{Expansion properties} \label{sssec:pdge-expansion} 
The expansion property satisfied by the Poisson model with edge regeneration can be stated as follows.

\begin{theorem}[Expansion] \label{thm:exp:pdge}
For every constant $d \geq 35$ and for every sufficiently large $n$,  let $\{G_t = (N_t, E_t) \,:\, t \in
\mathbb{R}^+\}$ be a \PDGE\ sampled from $\mathcal{G}(\lambda,\mu,d)$ with $\lambda=1$ and $\mu=1/n$. Then, for every fixed integer
$r \geq 7n \log n$,  w.h.p.  the snapshot $G_{T_r}$ is an   $\varepsilon$-expander with parameter $\varepsilon \geq 0.1$.

\end{theorem}

 The proof proceeds analyzing     three different 
  size ranges of the    vertex  subset  $S \subseteq N_{T_r}$,  the  expansion of which has to  be shown.

\paragraph{Expansion of small subsets.}

\begin{lemma}[Expansion of small subsets] \label{lem:poi:exp:small1}
Under the hypothesis of Theorem \ref{thm:exp:pdge}, for subsets $S$ of $N_{T_r}$, with probability of at least $1-2/n^2$,
\begin{equation}
    \min_{0 \leq |S| \leq n/\log^2n} \frac{|\partial_{out}(S)|}{|S|}\geq 0.1 \, .
\end{equation}

\end{lemma}

The proof of the above lemma adapts the argument we used in the proof of Lemma \ref{lem:expansion_small_streaming} for the streaming model, and  it is given in the Subsection \ref{ssec:lem:poi:exp:small1}.

 \paragraph{Expansion of middle-size subsets.} The second case deals with   subsets of size in the range 
 $n/\log^2n \leq |S| \leq n/14$ and its analysis definitely represents one of the key technical contributions of this paper. Indeed, departing from the first case, the presence of a large number of subsets in this range    does not allow to use any  rough worst-case counting argument:   for instance, assuming that all nodes in the considered subset $S$ have age $\bigO(n \log n)$ and applying the corresponding edge-probability bound given by  \eqref{eq:uppedgeprob-poisson} would lead to a useless, too large   union bound for  the probability of non-expansion for some subset $S$.
 
 In few words, to  cope with this technical issue, we need to 
 partition and classify the subsets $S$ and $T$ according to their \emph{age profile}. More in detail, we first define a sequence of $\Theta(\log n)$ \emph{slices} of possible nodes ages and then we provide an effective age profile  of  each subset $S$ (and $T$) depending on how large its intersection is with  each of these slices.
 Thanks to the properties of the exponential distributions of the life of every node in the Poisson model (see \eqref{eq:v_not_in_T_r+1} in Lemma \ref{lem:N_m+1_and_death}),
 we show that the existence of a given subset in a given time has a probability that essentially  depends on its  profile. Roughly speaking, the more is the number of old nodes in $S$, the less is the probability of  the presence  of $S$ in $N_{T_r}$. 
 
 Then, combining this profiling with a more refined use of the parameterized 
 bound on the edge probability in   \eqref{eq:uppedgeprob-poisson}, we get 
 a mathematical expression  (see \eqref{eq:kl_div_inequality_touse}) that, in turn, we show to be dominated by the  KL divergence  of two suitably defined probability distributions. Finally, our target probability bound, stated in the next lemma,  is obtained by the 
 standard KL divergence inequality (see Theorem \ref{thm:kullback_inequality}). 
 The arguments above allow us to prove the following result.

\begin{lemma}[Expansion of middle-size subsets] \label{lem:expansion_2}
Under the hypothesis of Theorem \ref{thm:exp:pdge}, for subsets $S$ of $N_{T_r}$, 
with probability of at least $1-2/n^2$,
\begin{equation}
\label{eq:lem:expansion_2}
    \min_{n/\log^2n \leq |S| \leq n/14} \frac{|\partial_{out}(S)|}{|S|}\geq 0.1 \, .
\end{equation}
\end{lemma}
\begin{proof}
  From Lemma \ref{lem:life_of_nodes},   all the nodes in $N_{T_r}$ are born after time $T_{r-7n\log n}$ with probability of at least $1-1/n^2$. So, if we define the event
  \[L_r=\{\hbox{each node in $N_{T_r}$ is born after time $T_{r-7n \log n}$}\}\,,\]
  we get that $\Prc{L_r} \geq 1-1/n^2$:  through the rest of this proof,  we will condition to this event. 
 
  So, if we denote $i$ as the node that joined the network at round $T_{r-i+1}$ (i.e. the node has age of $i$ rounds), conditioning to $L_r$,
\[N_{T_r} \subseteq \{1,2,3,\dots,7n \log n\}\,.\]
As in the previous analysis of small subsets, we have to show that (conditioning to $L_r$) any two disjoint sets $S,T \subseteq \{1,2,\dots,7n \log n\}$, such that $n/\log^2n \leq |S| \leq n/14$,   $|T|=0.1|S|$,     $S,T \subseteq N_{T_r}$, and $\partial_{out}(S) \subseteq T$, may exist only  with negligible probability. To this aim, we define the following event
\begin{equation}
    A_{S,T}=\{\partial_{out}(S) \subseteq T\} \cap \{S,T \subseteq N_{T_r}\}\,.
\end{equation}
For the  law of total probability, 
\begin{equation}
   \Prc{\min_{n/\log^2n \leq |S| \leq n/14}\frac{|\partial_{out}(S)|}{|S|} \leq 0.1} \leq \sum_{\substack{n/\log^2n \leq |S| \leq n/2,|T|=0.1|S| \\ S,T \subseteq \{1,2,\dots,7n\log n\}}} \Prc{A_{S,T} \mid L_r}+\frac{1}{n^2}\,,
   \label{eq:final_prob}
\end{equation}
   and, hence, our next goal is to upper bound the quantity $\Prc{A_{S,T} \mid L_r}$. For each $i \in S$, let $B_i$ be  the event ``Each of the $d$ requests of   node $i$ has destination in $S \cup T$''. Then,  we can write 
   \[ A_{S,T}=\cap_{i \in S}B_i \cap \{S,T \in N_{T_r}\} \, , \] and for Bayes' rule
\begin{align}
\label{eq:A_S,T}
    &\Prc{A_{S,T} \mid L_r}= \Prc{\bigcap_{i \in S}B_i \mid S,T \subseteq N_{T_r},L_r}\Prc{S,T \subseteq N_{T_r} \mid L_r} \,.
\end{align}
From Lemma \ref{lem:existence_edge_hyp}, conditional to the event $\{S,T \subseteq N_{T_r}\}$, we get   
\begin{equation}
    \Prc{B_i \mid S,T \subseteq N_{T_r},L_r}= \left[\frac{|S \cup T|}{0.8n}\left(1+\frac{i}{1.7n}\right)\right]^d\,.
    \label{eq:pr_S_i}
\end{equation}
Since we will use \eqref{eq:A_S,T} to bound $\Prc{A_{S,T} \mid L_r}$, we need an upper bound   for 
   $ \Prc{S,T \subseteq N_{T_r}}$.
To this aim,  we can use the bound \eqref{eq:v_not_in_T_r+1} in Lemma \ref{lem:N_m+1_and_death}. However, according to the definition of  \emph{round} in Definition \ref{def:T_i}, we know that the death of one node in one single round is not independent of the death of the others. Indeed, if we know that, in   a given round, the node $v$ dies, we will also know   that in this round no other event occurs, and, so, noone of the other nodes dies. Moreover, if we know that one node does not die in a given round, the probability to die of the other nodes will be larger. To cope with this issue, we consider the probability that a fixed set of node survives in one round. From  Lemma \ref{lem:N_m+1_and_death},   for an arbitrary set of $r$ nodes it holds 
\begin{equation}
    \Prc{v_1,\dots,v_r \in N_{T_{r}} \mid v_1,\dots,v_r \in N_{T_{r-1}},L_r} \leq 1-\frac{r}{2.2n} \leq \left(1-\frac{1}{2.2n}\right)^{r}\,,
    \label{eq:v_1,v_k_in_N_Tr}
\end{equation}
This is in fact  the probability that the next step is not characterized by the death of any of the $k$ considered nodes. The last inequality in \eqref{eq:v_1,v_k_in_N_Tr} follows from the binomial inequality. So, thanks to  \eqref{eq:v_1,v_k_in_N_Tr} and to  the memoryless property of the exponential distribution,  
\begin{equation}
    \Prc{S,T \subseteq N_{T_r} \mid L_r} \leq \prod_{i \in S \cup T}\left(1-\frac{1}{2.2n}\right)^i \leq \prod_{i \in S \cup T}e^{-i/2.2n}\,,
    \label{eq:pr_S,T_in_N_T}
\end{equation}
where,   in \eqref{eq:pr_S,T_in_N_T} we used the fact that, from \eqref{eq:v_1,v_k_in_N_Tr}, each node contributes in the product with a factor $1-1/(2.2n)$ for each round of its life.
Since each node chooses the destination of its out-edges independently of the other nodes, we can place  \eqref{eq:pr_S_i} and  \eqref{eq:pr_S,T_in_N_T} into  \eqref{eq:A_S,T},  and  obtain
\begin{align}
    \Prc{A_{S,T} \mid L_r}\leq \prod_{i \in S \cup T}e^{-i/2.2n}\cdot \prod_{i \in S} \min \left\{1, \left[\frac{|S \cup T|}{0.8n}\left(1+\frac{i}{1.7n}\right)\right]^d \right\}\,.
    \label{eq:bound_A_S,T}
\end{align}
   For each set $R \subseteq N_{T_r}$,  we define the sequence $(K_1^R,\dots,K_L^R)$ (where $L= 7 \log n$), whose goal is to classify the nodes of the set according to their \emph{age profile}:
\begin{align}
    &K_1^R=|R \cap \{1,2,\dots,n\}| \notag \\ &K_2^R=|R \cap \{n+1,\dots 2n\}| \notag \\ &\cdots \notag \\ \notag &K_{L}^R=|R \cap \{(L-1)n+1,\dots, L n\}|\,.
\end{align}
Notice that, if $|R|=r$ and $K_1^R=r_1,\dots,K_L^R=r_L$,  then it must holds $\sum_{m=1}^L r_m=r$. For each set $R \subseteq N_{T_r}$, we denote the vector of random variables $(K_1^R,\dots,K_L^R)$ as $\mathbf{K}^R$.
According to this definition, by setting $\mathbf{k}=(k_1,\dots,k_L)$ and $\mathbf{h}=(h_1,\dots,h_L)$, we can rewrite \eqref{eq:final_prob} as follows:
\begin{align}
\label{eq:final_prob_2}
    &\Prc{\min_{n/\log^2n \leq |S| \leq n/14} \frac{|\partial_{out}(S)|}{|S|} \leq 0.1} \\ &\leq \sum_{k=n/\log^2n}^{n/14} \sum_{\substack{k_1+\dots+k_L=k \\ h_1+\dots+h_L=0.1k}} \sum_{\substack{S,T: \ \mathbf{K}^S=\mathbf{k} \\ \mathbf{K}^T=\mathbf{h}}} \Prc{A_{S,T} \hbox{ s.t. } \mathbf{K}^S=\mathbf{k},\ \mathbf{K}^T=\mathbf{h} \mid L_r}+\frac{1}{n^2}\,.
\end{align}
Indeed, we have to sum over all the possible size  $k=n/\log^2n,\dots,n/14$ of the set $S$, all the possible vectors $\mathbf{k}$ and $\mathbf{h}$ whose sum of the elements is equal to $k$ and $0.1k$, respectively (i.e. the characterization of the age profiles of $S$ and $T$ with $|S|=k$ and $|T|=0.1|S|=0.1k$), and, finally, over all the possible sets $S,T$ characterized by $\mathbf{K}^S=\mathbf{k}$ and $\mathbf{K}^T=\mathbf{h}$, respectively.

\noindent
From \eqref{eq:bound_A_S,T}, we get

\begin{align}
    &\Prc{A_{S,T} \hbox{ s.t. } \mathbf{K}^S=\mathbf{k},\ \mathbf{K}^T=\mathbf{h} \mid L_r}  \leq p(\mathbf{k},\mathbf{h}) \label{def:p(k,h)} \\
    &= \prod_{m=1,\dots,L}e^{-0.4(m-1)k_m}  \prod_{m=1,\dots,L}e^{-0.4(m-1)h_m}\prod_{m=1,\dots,L}\min \left\{1,\left[\frac{|S \cup T|}{0.8n}\left(1+0.6m \right) \right]^{dk_m} \right\}\,.
  \notag
\end{align}
The number of    subsets $S,T \subseteq \{1,2,\dots,7n\log n\}$ such that  $(K_1^S,\dots,K_L^S)=(k_1,\dots,k_L)$ and  $(K_1^T,\dots,K_L^S)=(h_1,\dots,h_L)$
is bounded by
\begin{equation}
    n(\mathbf{k},\mathbf{h})=\binom{n}{k_1}\cdot\binom{n}{h_1} \cdots \binom{n}{k_2}\cdot \binom{n}{k_2}\cdots \binom{n}{k_L}\cdot \binom{n}{h_L}\,.
    \label{def:n(k,h)}
\end{equation}
So, we introduce  the quantity $s(\mathbf{k,h})$ and get the  following bound from \eqref{def:p(k,h)} and \eqref{def:n(k,h)}: 
\begin{equation}
\label{def:s(k,h)}
     s(\mathbf{k,h})=\sum_{\substack{S,T:\mathbf{K}^S=\mathbf{k}\\\mathbf{K}^T=\mathbf{h}}}\Prc{A_{S,T} \hbox{ s.t. } \mathbf{K}^S=\mathbf{k},\ \mathbf{K}^T=\mathbf{h}} \leq n(\mathbf{k},\mathbf{h}) \cdot p(\mathbf{k},\mathbf{h})\,.
\end{equation}
We place  \eqref{def:n(k,h)} and \eqref{def:p(k,h)} into \eqref{def:s(k,h)},  and, since $|S \cup T|=1.1k$, then: 
\begin{align}
    s(\mathbf{k,h}) \leq  \prod_{m=1}^{L}\left(\binom{n}{h_m}e^{-0.4(m-1)h_m}\cdot \binom{n}{k_m} e^{-0.4(m-1)k_m}\min \left\{1,\left(\frac{1.1k}{0.8n}(1+0.6m)\right)^{dk_m}\right\}\right)\,.
\end{align}
The next step is to prove that $s(\mathbf{k},\mathbf{h})\leq 2^{-0.15k}$ and, to this aim, we split $s(\mathbf{k},\mathbf{h})$ in two factors, $s_1(\mathbf{k,h})$ and $s_2(\mathbf{k,h})$: 
\begin{align}
    s_1(\mathbf{k,h})& =\prod_{m=1}^L \binom{n}{h_m}e^{-0.4(m-1)h_m}\, ; \\
 s_2(\mathbf{k,h})& =\prod_{m=1}^L \binom{n}{k_m}e^{-0.4(m-1)k_m}\min \left\{1,\left(\frac{1.1k(1+0.6m)}{0.8n}\right)^{dk_m}\right\} \, .
\end{align}
To give an  upper bound on $s(\mathbf{k,h})$,
 we provide separate upper bounds for  $\log( s_1(\mathbf{k,h}))$ and $\log(s_2(\mathbf{k,h}))$. In particular, we want to  show that 
\begin{equation}
    \log( s(\mathbf{k,h}))\leq -0.15k \, ,
    \label{eq:aim_log_s}
\end{equation}
which implies that
\begin{equation}
    s(\mathbf{k,h}) \leq 2^{-0.15k} \, .
    \label{eq:aim_bound_final_s(k,h)}
\end{equation}
We will start    bounding $\log(s_1(\mathbf{k,h}))$. Using  $\binom{n}{k}\leq \left(\frac{n\cdot e}{k}\right)^{k}$,  
\begin{align}
    \log(s_1(\mathbf{k,h})) \leq \sum_{m=1}^{L}h_m\log\left(\frac{n }{h_m}e^{-0.4m+1.4}\right)\,.
    \label{eq:log_s_1}
\end{align}
Since $\log(x)$ is a concave function, we can apply Jensen's inequality. In detail, for any concave function $\varphi$, numbers $x_1,\dots,x_L$ in its domain, and positive weights $a_1,\dots, a_L$,  it holds that
\[\frac{\sum_{m=1}^L a_m \varphi(x_m)}{\sum_{m=1}^L a_m} \leq \varphi\left(\frac{\sum_{m=1}^L a_m x_m}{\sum_{m=1}^L a_m}\right)\,.\]
So, taking $a_m=h_m$  $x_m=\frac{n}{h_m}e^{-0.4m+1.4}$ and recalling that $\sum_{m=1}^L h_m=0.1k$,  we obtain  
\begin{align}
    \frac{\sum_{m=1}^L h_m \log\left(\frac{n}{h_m}e^{-0.4m+1.4}\right)}{\sum_{m=1}^L h_m} \leq \log \left(\frac{n\sum_{m=1}^L e^{-0.4m+1.4}}{0.1k}\right)\,.
    \label{eq:jensen_log_s_1}
\end{align}
Since  $\sum_{m=1}^L e^{-0.4m+1.4}\leq 7$, combining    \eqref{eq:log_s_1}, \eqref{eq:jensen_log_s_1} and  since $k \leq n/14$,  we get 
\begin{equation}
    \log(s_1(\mathbf{k,h})) \leq 0.1k \log\left(\frac{7n}{0.1k}\right) \leq k \log \left(\frac{n}{7k}\right) \, ,
    \label{eq:bound_log_s1}
\end{equation}
where the last inequality follows by a simple calculation.

As for   $\log(s_2(\mathbf{k,h}))$,
\begin{align}
    \log(s_2(\mathbf{k,h})) \leq \sum_{m=1}^L k_m \log\left(\frac{n}{7k}\cdot \frac{n \cdot e}{k_m}e^{-0.4(m-1)}\left(\min \left\{1,\frac{1.1k(0.6m+1)}{0.8n}\right\}\right)^d\right) -k\log\left(\frac{n}{7k}\right) \, .
    \label{eq:bound_log_s_2}
\end{align}
Then, since $\log (s(\mathbf{k},\mathbf{h}))=\log( s_1(\mathbf{k},\mathbf{h}))+\log( s_2(\mathbf{k},\mathbf{h}))$,  from \eqref{eq:bound_log_s1} and \eqref{eq:bound_log_s_2},
\begin{align}
    \log(s(\mathbf{k,h}))\leq \sum_{m=1}^L k_m \log\left(\frac{0.6n^2}{k\cdot k_m}e^{-0.4m}\left( \min \left\{1,\frac{1.1k(0.6m+1)}{0.8n}\right\}\right)^d\right)\,.
    \label{eq:bound_log_s}
\end{align}
So,   from the above inequality,   
\begin{align}
    -\frac{\log(s(\mathbf{k,h}))}{k} \geq \sum_{m=1}^L \frac{k_m}{k}\log\left(\frac{k_m}{k}\frac{9}{10}\cdot \frac{k^2}{0.6n^2}e^{0.4m}\left(\min \left\{1,\frac{1.1k(0.6m+1)}{0.8n}\right\}\right)^{-d}\right)+\log(10/9)
    \label{eq:bound_log(s(k,h)/k}
\end{align}
Now, notice that, if we prove that
\begin{equation}
    \sum_{m=1}^L\frac{k_m}{k}\log\left(\frac{k_m}{k}\frac{9}{10}\cdot \frac{k^2}{0.6n^2}e^{0.4m}\left(\min \left\{1,\frac{1.1k(0.6m+1)}{0.8n}\right\}\right)^{-d}\right) \geq 0 \, ,
    \label{eq:kl_div_inequality_touse}
\end{equation}
then,  from \eqref{eq:bound_log(s(k,h)/k}, we would get  \eqref{eq:aim_log_s}, since $\log(10/9) \geq 0.15$.

So, we want  to prove \eqref{eq:kl_div_inequality_touse}. Thanks to the \emph{KL divergence inequality} (see Theorem \ref{thm:kullback_inequality}), it is sufficient  to show that the following functions are density mass functions over $\{1,2,\dots,L\}$:
\begin{equation}
    p_m=\frac{k_m}{k} \quad \text{and} \quad q_m=\frac{10}{9}\cdot \frac{0.6n^2}{k^2}e^{-0.4m}\min\left\{1,\left(\frac{1.1k(0.6m+1)}{0.8n}\right)^d\right\}\,.
\end{equation}
Notice that $\sum_{m=1}^L p_m=1$, and 
\begin{align}
    \sum_{m=1}^L q_m=&\sum_{m=1}^{0.9\frac{n}{k}-1}\frac{10}{9}\frac{0.6n^2}{k^2}e^{-0.4m}\left(\frac{9}{10}\right)^{d-2}\left(\frac{1.1k(0.6m+1)}{0.8n}\right)^2+\sum_{r=0.9\frac{n}{k}}^L\frac{10}{9}\frac{0.6n^2}{k^2}e^{-0.4m}  \\ &\leq 1.1\left(\frac{1.1}{0.8}\right)^2\left(\frac{9}{10}\right)^{d-3}+\frac{10}{9}\cdot \frac{0.6n^2}{k^2}\cdot e^{-0.36\frac{n}{k}} \leq 1\,,
\end{align}
where the last inequality holds   we taking $d$ large enough ($d \geq 30$) and $k \leq \frac{n}{14}$.
So, we have proved that $q_m$ and $p_m$ are density mass functions over $\{1,2,\dots,L\}$ and so, thanks to Theorem \ref{thm:kullback_inequality}, \eqref{eq:kl_div_inequality_touse} holds and  implies  \eqref{eq:aim_bound_final_s(k,h)}.

By placing \eqref{def:s(k,h)} in \eqref{eq:final_prob_2} and using  \eqref{eq:aim_bound_final_s(k,h)},

\begin{equation}
\Prc{\min_{n/\log^2 \leq |S| \leq n/14}\frac{|\partial_{out}(S)|}{|S|} \leq 0.1} \leq     \sum_{k=n/\log^2n}^{n/14}\sum_{\substack{k_1+\dots+k_L=k \\ h_1+\dots+h_L=0.1k}} s(\mathbf{k,h})+\frac{1}{n^2} \leq \frac{2}{n^2} \, ,
\end{equation}
  where the last inequality holds   since the number of integral sequences $k_1,\dots,k_L$ that sum up $k$ is bounded by $\binom{k+L}{L}$ (and the same holds for $h_m$), and hence, from simple calculations and recalling that $L=7n \log n$, 
\begin{align}
\label{eq:final_sum_s(k,h)}
      &\sum_{k=n/\log^2n}^{n/14}\sum_{\substack{k_1+\dots+k_L=k \\ h_1+\dots+h_L=0.1k}} s(\mathbf{k,h}) \leq \sum_{k=n/\log^2n}^{n/14}\binom{L+0.1k}{L}\binom{L+k}{L}2^{-0.15k} \leq \frac{1}{n^2} \, .
\end{align}
   
\end{proof}

\paragraph{Expansion of big subsets.}
The last case of our analysis of the vertex expansion of the \PDGE\ model considers subsets of big size $|S| \geq n/14$. It analysis is much simpler than the that of the previous case and proceeds exactly as the proof of Lemma \ref{lem:exp_large_subset_PDG} about the expansion of large subsets in the \ \PDG \ model. Indeed, in both the \PDG \ and \PDGE \ models,   we use the fact that any node $u \in N_{T_r}$ chooses any fixed older node $v \in N_{T_r}$ with probability $\geq 1/1.1n$ (thanks to Lemma \ref{thm:concentration_nodes}).
The proof is omitted since it is identical to that of Lemma \ref{lem:exp_large_subset_PDG} in Subsection \ref{ssec:lem:exp_large_subset_PDG}.

\begin{lemma}[Expansion of large subsets] \label{lem:expansion_3}
Under the hypothesis of Theorem \ref{thm:exp:pdge}, for subsets $S$ of $N_{T_r}$, 
with probability of at least $1-2/n^2$,
\begin{equation}
\label{eq:lem:expansion_3}
    \min_{n/14 \leq |S| \leq |N_{T_r}|/2} \frac{|\partial_{out}(S)|}{|S|}\geq 0.1 \, . 
\end{equation}. 
\end{lemma}


\subsubsection{Flooding} \label{sssec:Poisson-flooding}

We now consider the flooding process over the \PDGE\ model introduced in Definition 
\ref{def:flooding_poisson}. The 
vertex expansion property we derived in Theorem \ref{thm:exp:pdge} is here exploited to 
obtain a logarithmic   bound on the   time required by this  process to inform all the nodes of the graph. Notice that, according to the considered topology dynamics,  if there is 
a time in which   all the alive nodes are informed, then every  successive
snapshot of the dynamic graph will have all its nodes informed as well, w.h.p.

\begin{theorem}[Flooding]
\label{thm:flooding_terminates_poisson}
For every constant $d \geq 35$, for every sufficiently large $n$ and for every fixed $r_0 \geq 7n \log n$,  consider  the flooding process over a \PDGE \ sampled from $\mathcal{G}(\lambda, \mu, d)$, with $\lambda=1$ and $\mu=1/n$, and  starting at $t_0=T_{r_0}$. Then, w.h.p.,    the flooding    time is $\bigO(\log n)$. 
\end{theorem}

As remarked before, in dynamic networks without node churn, it has already been  shown \cite{CMPS11}
that  the good  vertex expansion of every snapshot   implies  fast flooding time (see for instance \cite{CMPS11}). 
In the Poisson models, the  presence of random node churn requires to consider some new technical issues.  
Indeed, once we observe the set of informed nodes $I_t$   at a given snapshot $G_t=(N_t,E_t)$, the expansion of $I_t$ refers to topology $E_t$ while the 1-hop
message transmissions take one  unit of time. So, during this time interval, some topology changes may take place  affecting the expansion observed at time $t$.
To cope with this    issue,  our analysis  splits the process 
  into  three consecutive phases and prove they all have
logarithmic length, w.h.p.
The details of this approach are given in Subsection \ref{sssec:thm:flooding_terminates_poisson}.

\section{Overall Remarks and Open Questions} \label{sec:conlc}
We   studied 
two    models of fully-random dynamic networks with node churns. We analysed their  expansion properties   and gave   results about  the performances of the flooding process. We essentially show that such important aspects depends on the specific adopted topology dynamic, namely, on whether or not,  edge regeneration takes place along the time process.

While our models are too simplified to predict all properties of realistic
networks, the Poisson model with edge regeneration bears a certain similarity
to the way peer-to-peer networks such as Bitcoin are formed. In particular,
although the random choices over the current node set $N_t$ the nodes make to
establish connections is not the connection mechanism adopted in standard
Bitcoin implementations, the set of IP addresses of the active full-nodes of
the Bitcoin network can be easily discovered by a crawler (see, e.g.,
\cite{bitnodes}). This implies that, potentially, nodes can implement a good
approximation of the fully-random strategy by picking random elements from such
on-line table.

We see an interesting future research direction related to our work. The
topology dynamics we considered yield sparse graphs at every round, however,
the maximum node degree can be of magnitude $\bigO(\log n)$. For some real
applications this bound is too large, and finding natural, fully-random
topology dynamics that yield bounded-degree snapshots of good expansion
properties is a challenging issue which has strong theoretical and practical
motivations~\cite{allen2016expanders,augustine2016distributed,mao2020perigee}.

\section{Omitted Proofs for the Streaming Model} \label{apx:ssec:stream}
In this section we present the proofs of the results we obtained for  the streaming model.

\subsection{Omitted Proofs for the streaming model without edge regeneration}

\subsubsection{Lemma \ref{lem:exp_degree_purestreaming}}
\label{ssec:lem:exp_degree_purestreaming}

We first observe that the expected degree of each node in this graph
is $d$. Thus, the expected number of edges in the graph $nd/2$.

\begin{lemma}[Expected degree]
\label{lem:exp_degree_purestreaming}
Let  $G_t=(N_t,E_t)$ be the snapshot of a  \SDG\ $\mathcal{G}(n,d)$. Then, for any $t \geq n$, every    node in $N_t$ has expected degree $d$.
\end{lemma}

\begin{proof}
We first fix $t \geq n$, and let $v_1,\dots,v_n$ be the nodes n the network at round $t$, where $v_i$ is the node with age $i$. We define the following Bernoulli random variable, for each $i,j \in [n]$ with $i<j$ ($v_i$ joined the network after $v_j$)
\begin{equation}
 z^{(k)}(v_i,v_j)=
\begin{cases}
1 \text{ if the node $v_i$ at the time of its arrival has connected its $k$-th request to $v_j$} \\
0 \text{ otherwise}
\end{cases}
\end{equation}
We notice that for each $k$ and $i,j$ s.t. $i< j$ the random variables $z^{(k)}(v_i,v_j)$ are independent.
We indicate with $\Delta_i^t$ the degree of the node $v_i$ at time $t$, where $i \in [n]$. Then, we can say that $\Delta_i^t$ is the random variable where
\begin{equation}
    \Delta_i^t = \sum_{k=1}^{d} \left( \sum_{j=1}^{i-1}z^{(k)}(v_j,v_i)+\sum_{j=i+1}^{n}z^{(k)}(v_i,v_j)\right)
\end{equation}
and, since for each $i,j$ with $i<j$ and $k$ we have that $\Prc{z^{(k)}(v_i,v_j)=1}=1/(n-1)$, it holds that for each $t \geq n$ and $i=1,\dots,n$
\begin{equation}
    \Expcc{\Delta_i^t}=d.
\end{equation}
\end{proof}

\subsubsection{Proof of Lemma \ref{lem:isolated_nodes}}
\label{ssec:lem:isolated_nodes}

Let $\varepsilon$ be an arbitrary value with $0 < \varepsilon \leq 1/3$.   We first define the set $H$ of the oldest $\varepsilon n$ nodes in $N_t$. We now define the random variable  
\begin{equation}
    X =\{\text{number of nodes in $H$ that are isolated at round $t$ and for the rest of their lifetime}\}\,.
\end{equation}
  Our goal is to show that, w.h.p., $X \geq \frac{1}{2}\varepsilon n e^{-2d}$.  First, to evaluate the expectation of $X$, we introduce the following random variables, for each $v  \in N_t$:
\begin{equation}
    \Delta^{in}_v=\{\text{maximum in-degree of the node $v_i$ for all its lifetime}\};
\end{equation}
\begin{equation}
    \Delta^{out}_v=\{\text{out-degree of the node $v_i$ at round $t$}\}.
\end{equation}
Since, if $\Delta_v^{out}=0$, then the node $v$ will have not out-edges from round $t$ to the rest of its lifetime, we can then  write $X$ as a function of $\Delta^{in}_v$ and $\Delta^{out}_v$, with $v \in H$, i.e., 
\begin{equation}
\label{eq:X_sum_nodes_isolated}
    X=\sum_{v \in H}\mathbb{1}_{\{\Delta^{in}_i=0\}}\mathbb{1}_{\{\Delta^{out}_i=0\}}\,.
\end{equation}
Since each node sends its requests independently of the others, we get   
\begin{equation}
    \Expcc{X}=\sum_{v \in H} \Prc{\Delta^{in}_v=0}\Prc{\Delta^{out}_v=0}\,.
    \label{eq:exp_isolated_nodes_S}
    \end{equation}
The probability that a node $v \in H$ has always in-degree $0$, i.e. it does not receive any request by other nodes for all its lifetime is 
\[ \Prc{\Delta_i^{in}=0}=\left(1-\frac{1}{n}\right)^{nd}\geq e^{-3d/2} \, . \]
The probability that  node $v \in H$ have no out-edges in the current round is 
\[ 
    \Prc{\Delta_v^{out}=0}=\left(1-\varepsilon \right)^d \geq e^{- d/2} \, , \] since $\varepsilon \leq 1/3$. So,   from \eqref{eq:exp_isolated_nodes_S} and since $|H|=\varepsilon n$
    \begin{equation}
        \Expcc{X} \geq \varepsilon n e^{-2d}\,.
    \end{equation}
Observe that, in \eqref{eq:X_sum_nodes_isolated}, $X$ is not a sum of independent random variables, so to get a concentration result,  we   use the method of bounded differences. We introduce the random variables $Y^j_v$, returning the index of the node to which the  node $v$ sends its $j$-th request. Notice that the random variables $\{Y_v^j : v \in N_t\cup N_{t+1}\cdots \cup N_{t+\varepsilon n},j \in [d]\}$ are independent, and they represent the destination of the requests of the nodes in the network and of the $\varepsilon n$ nodes that will join the network after time $t$.   Then, considering the vector  $\mathbf{Y}$ of these random variables, we can easily express $X$ as a function of $\mathbf{Y}$,
\[X=f(\mathbf{Y}).\]
Moreover, if any two vectors $\mathbf{Y}$ and $\mathbf{Y'}$ differs only in one coordinate, it holds  $|f(\mathbf{Y})-f(\mathbf{Y'})| \leq 2 \, . $
Indeed, in the worst case, an isolated node can change its destination from a leaving node to another isolated node: so, the number of isolated nodes decreases (or increases) by at most  $2$ units.
By  applying   Theorem \ref{thm:bounded_differences},  we get that, if $\mu$ is a lower bound to $\Expcc{X}$ and $M>0$,
\begin{equation}
    \Prc{X \leq \mu - M} \leq e^{-\frac{2M^2}{4nd(1+\varepsilon)}} \, .
\end{equation}
Hence, we can fix $\mu=\varepsilon n e^{-2d}$ and $M=\frac{1}{2}\varepsilon n e^{-2d}$, and get  
\begin{equation}
    \Prc{X \leq \frac{1}{2}\varepsilon n e^{-2d}}\leq e^{-n\frac{\varepsilon^2 e^{-4d}}{8d(1+\varepsilon)}} \, .
\end{equation}
Finally, the lemma is proved by setting  $\varepsilon=1/3$.

\subsubsection{Proof of Lemma \ref{lem:exp_large_sub_SDG}} \label{ssec:lem:exp_large_sub_SDG}

We  show that any two disjoint sets $S,T \subseteq N_t$, with $n e^{-d/10} \leq |S| \leq n/2$ and $|T|=0.1|S|$, such that  $\partial_{out}(S) \subseteq T$, exist with negligible probability. If we denote
\begin{equation}
    A_{S,T}=\{\partial_{out}(S) \subseteq T\}
\end{equation}
we have that
\begin{align}
    &\Prc{\min_{n e^{-d/10} \leq |S| \leq n/2}\frac{|\partial_{out}(S)|}{|S|} \leq 0.1} \leq \sum_{\substack{n e^{-d/10} \leq |S| \leq n/2 \\ |T|=0.1|S|}}\Prc{A_{S,T}}\,.
    \label{eq:prob_complementare_2_2_noreg}
    \end{align}
To upper bound $\Prc{A_{S,T}}$ we define $P$ as the set $P=N_t-S-T$, and notice that the event $A_{S,T}$ is like saying that between $S$ and $P$ there are no edges. Clearly, it holds that
\begin{equation}
    \left|\{(a,b) \hbox{ s.t. } a \in S, b \in P\}\right|=|S|\cdot |P| \,.
\end{equation}
Two cases may arise: 
\begin{enumerate}
    \item $
        \left|\{(a,b) \hbox{ s.t. } a \in S, b \in P \hbox{ and $a$ is younger than $b$}\}\right|\geq |S| \cdot |P|/2$;
    \item 
        $\left|\{(a,b) \hbox{ s.t. } a \in S, b \in P \hbox{ and $b$ is younger than $a$}\}\right|\geq |S| \cdot |P|/2$.
\end{enumerate}
As for the first case, for each $a \in S$,
we define  $N_a$ as the number of nodes in $P$ older then $a$. We get that $\sum_{a \in S}N_a \geq |S|\cdot |P|/2$. Since a fixed request of a node $u$ has probability $1/n$ to get  a node $v$ older than him as destination (this is in fact the probability that $u$ connects  to $v$ when  $u$ joins the network), we get

\begin{equation} \label{eq:AST1_noreg}
    \Prc{A_{S,T}} \leq \prod_{a \in S}\left(1-\frac{N_a}{n}\right)^d \leq e^{-d \sum_{a \in S}N_a/n} \leq e^{-d|S|\cdot |P|/2n}.
\end{equation}
As for the second case, we get the same bound  above   by proceeding with a similar argument. 
Then, from  \eqref{eq:prob_complementare_2_2_noreg} and \eqref{eq:AST1_noreg},      
\begin{align}
    &\Prc{\min_{n e^{-d/10} \leq |S| \leq n/2}\frac{|\partial_{out}(S)|}{|S|} \leq 0.1} \leq \sum_{s=ne^{-d/10}}^{n/2} \binom{n}{s}\binom{n-s}{0.1s} e^{-ds\frac{n-1.1s}{2n}} \leq \frac{1}{n^4}\, ,
    \end{align}
    where the last inequality holds for large enough $n$  and for any $d \geq 20$. It  easily follows    by bounding each binomial coefficient with the   $\binom{n}{k}\leq \left(\frac{n \cdot e}{k}\right)^{k}$ and by standard calculations.
    
\subsubsection{Proof of Theorem \ref{thm:not_flooding_purestreaming}} \label{ssec:thm:not_flooding_purestreaming}

Let $s_0$ be the source node of the flooding process. We bound the  probability that $s_0$ has all its out-edges towards $d$ nodes that are isolated at round $t_0$ and keep so  for the rest of their lifetime. Define the events $A=$"the source node $s_0$ connects  to $d$ nodes that are isolated at round $t_0$ and for the rest of their lifetime" and $B=$"there are at least $\frac{ne^{-2d}}{6}$ nodes in $N_{t_0}$ that are isolated at time $t_0$ and for the rest of their lifetime".
From Lemma \ref{lem:isolated_nodes},
\begin{equation}
    \Prc{A}  \geq \Prc{B} \Prc{A \mid B} \geq  \frac{1}{2}\cdot \left(\frac{e^{-2d}}{6}\right)^d .
    \label{eq:pr_A_poisson_noflod}
\end{equation}
Let $s_1,\dots,s_d$ be the $d$ out-neighbors      of  $s_0$.
Notice that the events $A$ imply that $s_0,s_1,\dots,s_d$ are isolated  and  after $n$ rounds   all of them will leave the network.   Then, from \eqref{eq:pr_A_poisson_noflod},
\begin{equation}
    \label{pr:I_t_0+n}
\Prc{ I_{t_0+n} =\emptyset} \geq \Prc{A} \geq  \frac{1}{2}\cdot \left(\frac{e^{-2d}}{6}\right)^d\,.
\end{equation}
Finally, as for the stated  linear bound on the flooding time $\tau$, we observe this is an easy consequence of Lemma \ref{lem:isolated_nodes}. Indeed, in the network there are at least $\frac{1}{6}e^{-2d}n$ isolated nodes that will remain isolated for the rest of their lifetime: to have all the nodes informed, we have to wait that these nodes leave the network, and so $\tau=\Omega_{d}(n)$.


\subsubsection{Proof of Claim \ref{claim:O_0}}
\label{ssec:claim:O_0}

As for the claim for Phase 0, the proof of the first inequality of the claim
proceeds as follows. For each $i=1,\dots,d$ and for every $v\in N_{t_0}$ that
joined the network at time $\hat{t}\le t_0$, we define the variable
$A_v^{(i)}\in N_{\hat{t}}$ as follows:
\begin{equation}
\label{def:A_v^i}
    A_v^{(i)} = w \quad \text{if } w\in N_{\hat{t}} \text{ is the destination of the $i$-th link of $v$.}
\end{equation}
Assume $w\in O$. We have:
\[
	\Prc{\exists i\in [d]: A_s^{(i)} = w} = 1 - \left(1 - 
	\frac{1}{n}\right)^d,
\]
which implies 
\[
	\Expcc{|O_0|}  \geq |O|\left(1 - \left(1 - 
	\frac{1}{n}\right)^d\right)\ge |O|\frac{d}{2n} \geq \frac{d}{5}, 
\]
where the last equality follows since $|O|=n/2-\log n$. We next bound the
probability that $O_0$ is smaller than $d/10$. To this purpose, we cannot
simply apply a Chernoff bound to the binary variables that describe whether or
not a node $w\in O$ was the recipient of at least one link originating from
$s$, since these are not independent. We instead resort to Theorem
\ref{thm:bounded_differences}. In particular, we consider the random variables
$A_s^{(1)},\ldots , A_s^{(d)}$ and we define the function $f(A_s^{(1)},\ldots ,
A_s^{(d)}) = |O_0|$. Clearly, $f$ is well-defined for each possible realization
of the $A_s^{(i)}$'s.  Moreover, $f$ satisfies the Lipschitz condition with
values $\beta_1 =\cdots = \beta_d = 1$, since changing the destination of one
link can affect the value of $|O_0|$ by at most $1$. We can thus apply
Theorem~\ref{thm:bounded_differences} to obtain:
\[
	\Prc{|O_0| < \frac{d}{10}}\le\Prc{f < \Expcc{f} - 
	\frac{d}{10}}\le e^{-d/50}.
\]

To uniform the results, in the claim we give a weaker bound following from the
bound above.

As for the first inequality of the generic phase $k \geq 1$, we proceed as
follows.  For each $i=1,\dots,d$, for each node $v \in N_{t_0}$ and for each
set $A \subseteq N_{t_0}$, we define the Bernoulli random variable $R_{v,A}$ as
follows:
\begin{equation}\label{def:R_v,A_streaming}
    R_{v,A}=\left\{
    \begin{array}{l}
		1\quad \hbox{if $x\ge 1$ links in }\{\frac{d}{2},\dots,d\} \hbox{ from $v$ have destination in 
		$A$}\\
		0 \quad \hbox{otherwise}
	\end{array}\right.
\end{equation}
We remark that 
\begin{fact}\label{fa:onion_dest}
If $v \in Y-Y_{k-1}$ establishes a link with destination  $w 
\in O_{k-1}$ in phase $k$, then $w \not \in O_{k-2}$.
\end{fact}
\begin{proof}
If this were the case, the definition of the onion-skin process would 
imply $v\in Y_j$, for some $j\le k-1$, a contradiction.
\end{proof}
From the fact above and from definition of $Y_k-Y_{k-1}$ given above we have:
\begin{equation}
    |Y_{k}-Y_{k-1}|=\sum_{v \in Y-Y_{k-1}}R_{v,O_{k-1}-O_{k-2}}.
\end{equation}
Moreover, for each $v \in N_{t_0}$
\[
	\Prc{R_{v,O_{k-1}-O_{k-2}} = 1 \mid |O_{k-1}-O_{k-2}| \geq y}\ge 1 
	- \left(1 - \frac{y}{n}\right)^{\frac{d}{2}},
\]
whence
\begin{equation}
     \Expcc{|Y_k-Y_{k-1}| \mid |O_{k-1}-O_{k-2}| \geq y} \geq 
     |Y-Y_{k-1}| \left(1 - \left(1 - 
     \frac{y}{n}\right)^{\frac{d}{2}}\right).
\end{equation}
Since $y\le n/d$, we have:
\[
	\Expcc{|Y_k-Y_{k-1}| \mid |O_{k-1} - O_{k-2}| \ge y} \ge |Y - 
	Y_{k-1}|\frac{yd}{4n}\ge\frac{yd}{10}, 
\]
where in the last inequality we used the assumption that $|Y_{k-1}|\le 
n/d$.
The $R_{v,O_{k-1}-O_{k-2}}$'s are independent, we can 
therefore apply Chernoff's Bound (Theorem \ref{thm:chernoff} in the 
Appendix) to obtain 
\begin{align}
    &\Prc{|Y_k-Y_{k-1}| \leq \frac{yd}{20} \mid |O_{k-1}-O_{k-2}| \geq y}
     \leq e^{-yd/100}.
\end{align}

As for the last inequality of the claim, we proceed similarly to the case $k=0$.
In this setting, the  following fact holds.
 
\begin{fact}
If $v \in Y_k$ establishes a link with destination $w \in O-O_{k-1}$ in phase
$k$, then $v \not \in Y_{k-1}$.
\end{fact}
\begin{proof}
If this were the case, the definition of the onion-skin process would imply $w
\in O_j$ for some $j\leq k-1$, a contradiction.
\end{proof}
Assume $w \in O$. Recalling the definition of the random variables $A_v^{(i)}$
in \eqref{def:A_v^i}, we have
\begin{equation}
\Prc{\exists i \in \left[\frac{d}{2}\right], \exists v \in Y_k-Y_{k-1}  :
A_{v}^{(i)}=w\mid |Y_{k}-Y_{k-1}| \geq x}
=1-\left(1-\frac{1}{n}\right)^{\frac{dx}{2}}\,.
\end{equation}
From the fact above and from the definition of $|O_k-O_{k-1}|$ in the
onion-skin process we have
\begin{equation}
\label{eq:E(Ok-Ok-1)_bound}
    \Expcc{|O_k-O_{k-1}| \mid |Y_{k}-Y_{k-1}| \geq x}= |O-O_{k-1}|\left(1-\left(1-\frac{1}{n}\right)^{\frac{dx}{2}}\right)\,.
\end{equation}
Since $x \leq n/d$, we have
\begin{equation}
    \Expcc{|O_{k}-O_{k-1}| \mid |Y_k-Y_{k-1}| \geq x} \geq |O-O_{k-1}|\frac{dx}{4n}\geq  \frac{xd}{10} 
\end{equation}
where the last inequality follows from the fact that $|O|= n/2- \log n$ and
$|O_{k-1}| \leq n/d$.  However, in this case, as in the proof of the case
$k=0$, we cannot simply apply a Chernoff bound to the binary random variables
that describe whether or not a node $w \in O$ was the recipient of at least one
link in $\{1,\dots, d/2\}$ originating from one node $v \in Y_k$, since these
are not independent. We will use instead the method of bounded differences
(Theorem~\ref{thm:bounded_differences}). In particular, we consider the random
variables $\{A_v^{(i)}, i \in [d/2],v \in Y_k\}$ and we define the function $g$
depending on these variables which returns $|O_k-O_{k-1}|$. Clearly, $g$ is
well-defines for each possible realization of the $A_{v}^{(i)}$'s and satisfies
the Lipschitz condition with values $\beta_1=\dots=\beta_{\frac{d}{2}\cdot
|Y_k-Y_{k-1}|}=1$, since changing the destination of one link can affect the
value of $|O_{k}-O_{k-1}|$ of at most $1$. We can thus apply
Theorem~\ref{thm:bounded_differences} to obtain
\begin{equation} \Prc{|O_k-O_{k-1}|
\geq \frac{dx}{20} \mid |Y_k-Y_{k-1}| \geq x}\leq e^{-dx/100}\,.
\end{equation}

\subsubsection{Proof of Claim \ref{claim:product}} \label{ssec:claimprod}
From $a_i=(d/20)^i$ we have:
\begin{equation}
\label{eq:log_product}
	\log c = \log \left(\prod_{i=0}^{\infty}\left(1-e^{-(d/20)^i 
	(d/100)}\right)\right)=\sum_{i=0}^{\infty} \log 
	\left(1-e^{-(d/20)^i(d/100)}\right)=-\sum_{i=0}^{\infty}\log\left(\frac{1}{1-e^{-(d/20)^i 
	(d/100)}}\right)\,.
\end{equation}
Moreover, since $\log \left(\frac{1}{1-x}\right) \leq 2x$ for each $x \leq 1$,
\begin{align}
	&\sum_{i=0}^{\infty} \log 
	\left(\frac{1}{1-e^{-(d/20)^i(d/100)}}\right) \leq 
	\sum_{i=0}^{\infty} 2e^{-(d/20)^i(d/100)} = 2e^{-(d/100)} + 
	\sum_{i=1}^{\infty} 2e^{-(d/20)^i(d/100)}\\
	&< 2e^{-(d/100)} + 2e^{-(d/100)}\sum_{i=1}^{\infty}e^{-(d/20)^i} < 4e^{-(d/100)}.
    \label{eq:double_exp}
\end{align}
The third inequality holds since $d> 200$, which implies
$\frac{d}{100}\left(\frac{d}{20}\right)^i > \frac{d}{100} +
\left(\frac{d}{20}\right)^i$ for $i\ge 1$, while the last inequality follows
since the double exponential is dominated by a simple one, summing to a
constant not exceeding $1$. This in turn implies So, from
\eqref{eq:log_product} and \eqref{eq:double_exp} we have that
\begin{align}
    \log c \geq -4e^{-(d/100)}\,,
\end{align}
whence:
\begin{align}
    c\ge e^{-4e^{-(d/100)}}\ge 1 - 4e^{-\frac{d}{100}}\,,
\end{align}
where the last inequality follows since $e^{-x} > 1 - x$ for $x\ge 0$.

\subsubsection{Proof of Lemma \ref{le:big_set_exp}} \label{ssec:le:big_set_exp}
We recall that, for Lemma \ref{lem:flooding_terminates_part_1},
$|I_{t_0+\tau_1}|\geq 2n/d$.  In this proof, we  show that the size of the set
of informed nodes grows by a constant factor at each step, reaching size
$(1-e^{-d/10})n$ in $\tau_2$ steps. In our analysis, we can fix
$\tau_2=\Theta(d)$.  In the whole proof, we will not consider the oldest
$\tau_2$ nodes in $N_{t_0+\tau_1}$, because they will die in the next $\tau_2$
steps. We call the set of such nodes $V$, and let $V_t$ be the intersection
between $V$ and $N_t$. To show that the set of informed nodes grows at each
step by a constant factor, it is sufficient to notice that, for $d \geq 20$,
the graph $G_t=(N_t,E_t)$ is an expander for sets of large size w.h.p. (see
Lemma \ref{lem:exp_large_sub_SDG}), i.e.,
\begin{equation}
\label{eq:expansion_set_large_size}
    \min_{ne^{-d/10} \leq |S| \leq n/2} \frac{|\partial_{out}(S)|}{|S|} \geq 0.1 \, .
\end{equation}
Indeed, the size of the set of informed nodes $I_t$, for $t \geq t_0+\tau_1$,
grows by a constant factor at each step, as long as $|I_t| \leq n/2$. We have
from the definition of flooding that $I_{t+1}=(I_t \cup \partial_{out}(I_{t}))
\cap N_t$ and since, at each step, at most $1$ informed node leaves the
network, $|I_{t+1}| \geq |\partial_{out}(I_t)|+|I_t|-1$. Since
$|I_{t_0+\tau_1}| \geq n/d$, from \eqref{eq:expansion_set_large_size} follows
that for each $t \geq t_0+\tau_1$ \[|I_{t+1}| \geq 1.1|I_t|-1,\] as long as
$|I_t| \leq n/2$. This consideration implies that in $\tau_2'=\Theta(\log d)$
steps, we will have $|I_{t_0+\tau_1+\tau_2'}| \geq n/2$. 

Now, we consider $S_t=N_t-I_t$ as the set of non-informed nodes in the graph at
time $t$, and we will show that the size of this set decreases by a constant
factor at each step, as long as $|S_t| \geq ne^{-d/10}$. First of all, we
notice that $\partial_{out}(S_{t+1}) \subseteq (S_{t}-S_{t+1}) \cap N_t$,
because $\partial_{out}(S_{t+1}) \subseteq I_{t+1}$ are nodes reachable in one
edge from the non-informed nodes, and so they were not informed in the previous
time. Since, at each step, at most $1$ node joins the network, we have that
$|S_{t}|-|S_{t+1}|+1 \geq |\partial_{out}(S_{t+1})|$. Since
$|S_{t_0+\tau_1+\tau_2'}| \leq n/2$, from \eqref{eq:expansion_set_large_size}
follows that for each $t \geq t_0+\tau_1+\tau_2'$, $|\partial_{out}(S_{t+1})|
\geq 0.1 |S_{t+1}|$, as long as $|S_{t+1}| \geq ne^{-d/10}$: so, \[|S_{t+1}|
\leq \frac{1}{1.1}(|S_{t}|+1).\] This consideration implies that in
$\tau_2=\Theta(d)$ steps, we will have $|S_{t_0+\tau_1+\tau_2}| \leq
ne^{-d/10}$.  Then, conditional at the event $\{|I_{t_0+\tau_1}| \geq n/d\}$,
w.h.p. $|I_{t_0+\tau_1+\tau_2}| \geq n(1-e^{-d/10})$. Since from Lemma
\ref{lem:flooding_terminates_part_1}, $\Prc{|I_{t_0+\tau_1}| \leq n/d}\geq
1-4e^{-d/100}$, the lemma is proved.

\subsection{Omitted Proofs for the streaming model with edge regeneration}
\subsubsection{Proof of Lemma \ref{lem:node_destination}}\label{ssec:lem:node_destination}

 If $u$ is younger than $v$, then the request of $u$ can choose $v$ only if some previous neighbor of $u$ leave the network: the probability of this event is clearly $\leq \frac{1}{n-1}$. 
 
 If, instead, $u$ is older than $v$, then the analysis leading to \eqref{eq:edgeprob1} needs to take care of more chances    $u$ has   to 
get $v$ as destination of its request because of the edge-regeneration process. The argument we adopt here for the streaming process is also a good warm-up for the more complex Poisson process we analyze in Subsection \ref{ssec:poissonedgereg}.

 According to the \SDGE\ model, when  a node leaves the network, all its incident edges are removed and, in the same step, if an edge of a fixed request of $u$ fails, $u$ instantly reassigns it by choosing its destination uniform at random (with replacement) over the current set of nodes: in this proof,  this action will be denoted as \emph{assignment}. For example, $u$ chooses $v$ at the $1$-st assignment if the request of $u$ gets destination $v$   in the first attempt; if, instead, $u$ first chooses a node that will die  at some successive  round and then $u$ chooses $v$, then we say $v$ is chosen at the $2$-nd assignment, and so on.
Since $u$ has age $k+1$,   $u$ must have selected $v$ in one of the  $(k+1)$ possible assignments. Then, for any fixed request of $u$,   
\begin{equation}
    \Prc{\text{$v$ is the   $u$-request destination}}  \,
    = \, \sum_{i=1}^{k+1}\Prc{\text{$v$ is the $u$-request destination  in  the $i$-th assignment}} \, . 
    \label{eq:prob_uv_sum}
\end{equation}
Since there are $\binom{k}{i-1}$ ways to choose the $i-1$ destinations of the request of $u$ before the choice of $v$, the probability of getting a fixed subset of  $i$  nodes in the reassignments is $\left(\frac{1}{n-1}\right)^i$, and such  choices in each round are mutually  independent. So, 
\begin{equation} \label{eq:wayski}
    \Prc{\text{$v$ is the $u$-request destination  in  the $i$-th assignment}}= \binom{k}{i-1}\left(\frac{1}{n-1}\right)^{i}\, .
\end{equation}
From \eqref{eq:prob_uv_sum}, \eqref{eq:wayski} and by Newton's Binomial Theorem
\begin{align*}
    &\Prc{\text{$v$ is the $u$-request destination  in  the $i$-th assignment}} \\ &= \sum_{i=1}^{k+1}\binom{k}{i-1}\left(\frac{1}{n-1} \right)^i=\sum_{j=0}^{k} \binom{k}{j}\left(\frac{1}{n-1}\right)^{j+1} 
    =\frac{1}{n-1}\left(1+\frac{1}{n-1}\right)^{k}\,.
\end{align*}

\subsubsection{Proof of Theorem  \ref{thm:expansion-stream}} \label{ssec:lem:expansion_small_streaming}

The theorem   is consequence of the next two lemmas. The first one shows the claimed expansion for subsets of  size $\leq n/4$.

\begin{lemma}[Expansion of ``small'' subsets]
 Under the hypothesis of Theorem \ref{thm:expansion-stream},
 for subsets $S$ of $N_t$, it holds
\[\min_{0 \leq |S| \leq n/4}\frac{|\partial_{out}(S)|}{|S|} \geq 0.1 \, ,\]
with probability at least $ 1-1/n^4$.
\label{lem:expansion_small_streaming}
\end{lemma}

\begin{proof}

We proceed similarly to the proof of the expansion for big subset in the \SDG \ model (Lemma \ref{lem:exp_large_sub_SDG})
To prove the lemma, it is sufficient to show that any two disjoint sets $S, T \subseteq N_t$, with $|S| \leq n/4$ and $|T|=0.1|S|$, such that $\partial_{out}(S) \subseteq T$, exist with negligible probability.
For  any  $S$ and any $T \subseteq N_t-S$,  we define $A_{S,T}$ as in Lemma \ref{lem:exp_large_sub_SDG},
\begin{equation}
    A_{S,T}=\left\{\partial_{out}(S) \subseteq T\right\}\,.
    \label{def:A_S,T}
\end{equation}
Therefore, 
\begin{align}
    &\Prc{\min_{0 \leq |S| \leq n/4}\frac{|\partial_{out}(S)|}{|S|} \leq 0.1} \leq \sum_{\substack{|S| \leq n/4 \\ |T|=0.1|S|}}\Prc{A_{S,T}} \,.
    \label{eq:prob_complementare_1_2}
    \end{align}
The quantity  $\Prc{A_{S,T}}$ is upper bounded by the probability that each request of  the nodes in $S$ has destination in  $S \cup T$. From Lemma \ref{lem:node_destination} we know that a request of a node $u$ with age $k+1$  has probability at most $\frac{1}{n-1}$ to have a node $v$ younger than $u$ as destination and probability $\frac{1}{n-1}\left(1+\frac{1}{n-1}\right)^k$ to have a node $v$ older than $u$ as destination. Since  $k \leq n-1$,   the probability that a single request of $u$ has an arbitrary, fixed  node $v$ as destination is at most $\frac{e}{n-1}$. Since to have $\partial_{out}(S) \subseteq T$, each request of $u \in S$ must have destination in $S \cup T$, it holds
\begin{equation}
\label{eq:prc_S,T_small_set_streaming}
    \Prc{A_{S,T}} \leq \left(\frac{e}{n-1}\cdot |S \cup T|\right)^{d|S|}\,.
\end{equation}
So, from \eqref{eq:prob_complementare_1_2} and \eqref{eq:prc_S,T_small_set_streaming} we have that
\begin{align}
    \Prc{\min_{0 \leq |S| \leq n/4}\frac{|\partial_{out}(S)|}{|S|}\leq 0.1} \leq \sum_{s=1}^{n/4}\binom{n}{s}\binom{n-s}{0.1s} \left(\frac{1.1s \cdot e}{n-1}\right)^{ds} \leq \frac{1}{n^4} \,.
    \label{pr_upper_bound_sum}
\end{align}
By standard calculus, it can be proved that, for $d \geq 21$, the equation above is upper bounded by $1/n^4$. This is obtained by bounding each binomial coefficient in \eqref{pr_upper_bound_sum} with the bound $\binom{n}{k} \leq \left(\frac{n \cdot e}{k}\right)^k$ and by computing the derivative of the function $f(s)$ (representing each term of the sum), obtaining that each of these terms attained its maximum at the "boundaries", i.e. or in $s=1$ or in $s=n/4$. 
\end{proof}

The second lemma provides  the expansion property for subsets of  size $\geq n/4$.
Its proof  is omitted since it is identical to that of Lemma \ref{lem:exp_large_sub_SDG} about the expansion of large subsets in the \SDG\ model. 

The only difference between the two proofs is that now we need to consider sets of size in the range $[n/4,n/2]$, and so we get a weaker condition on the value of $d$.

\begin{lemma}[Expansion of large subsets]
\label{lem:expansion_big_streaming}
Under the hypothesis of Theorem \ref{thm:expansion-stream},
 for subsets $S$ of $N_t$, it holds
\[\min_{n/4 \leq |S| \leq n/2}\frac{|\partial_{out}(S)|}{|S|} \geq 0.1 \, , \]
with probability of at least $1-1/n^4$.
\end{lemma}

\subsubsection{Proof of Theorem \ref{thm:SDGE-flooding}} \label{ssec:thm:SDGE-flooding}

We first show that the size
of the set of informed nodes $I_t$, for $t \geq t_0$, grows by a constant factor at each step, as long as $|I_{t}| \leq n/2$. From the definition of flooding it holds $I_{t+1}=(I_t \cup \partial_{out}(I_t))\cap N_t$ and, since, at each round, one single node leaves the network, $|I_{t+1}| \geq |\partial_{out}(I_t)| + |I_t| - 1$.  Since $d \geq 21$,    Theorem \ref{thm:expansion-stream} implies that the graph $G_{t}$ is an $(1/10)$-expander, w.h.p., and so, as long as $|I_{t}| \leq n/2$, w.h.p.
\[|I_{t+1}| \geq 1.1|I_{t}|-1\,.\]
Then, a $\tau_1=\bigO(\log n)$ exists such that  $|I_{\tau_1+t_0}|\geq n/2$,
w.h.p.

To reach $n$ informed nodes, from time $\tau_1+t_0$, we consider the set of non-informed nodes in the network at each time $t \geq \tau_1+t_0$, i.e. $S_t=N_t-I_t$, and we next show that the size of this set decreases by a constant factor at each step. Notice that, since every node $v$ in $\partial_{out}^{t+1}(S_{t+1}) \subseteq I_{t+1}$ is   reachable in 1-hop from the set of non-informed nodes at time $t+1$,  $v$ was not informed at time $t$. This implies that  $\partial_{out}^{t+1}(S_{t+1}) \subseteq (S_{t}-S_{t+1}) \cap N_{t+1}$. Since, at each step, one single node joins the network, we have that $|S_{t}|-|S_{t+1}|+1 \geq |\partial_{out}^{t+1}(S_{t+1})|$. Since $S_{t_0+\tau_1} \leq n/2$, from the expansion of the graph $G_{t+1}$ (Theorem \ref{thm:expansion-stream}) it holds w.h.p. that, for each $t \geq t_0+\tau_1$,

\[|S_{t+1}| \leq \frac{1}{1.1}(|S_{t}|+1) \, .\]
The above equation implies that a time $\tau_2=\bigO(\log n)$ exists such that $|S_{t_0+\tau_1+\tau_2}|<1$.

\section{Omitted Proofs for the Poisson Model } \label{apx:sec:Poisson}

\subsection{Omitted proofs for the preliminary properties} 

\subsubsection{Proof of Lemma \ref{lem:prop_Tn}} \label{ssec:lem:prop_In}

It is sufficient to apply Theorem \ref{thm:minimum_exp} in Appendix to the exponential random variables that represent the time arrival of a new node and the lifetime of the node in the network, for each node $v \in N_{T_n}$.

\subsubsection{Proof of Lemma \ref{lem:N_m+1_and_death}}
 \label{ssec:lem:N_m+1_and_death}
 
The lemma easily follows from Lemma \ref{lem:prop_Tn} and from the concentration of the nodes (Lemma \ref{thm:concentration_nodes}). 
We first define, for each $r \geq n \log n$, the following event
\begin{equation}
    C_{r}=\{|N_{T_r}| \in [0.9n,1.1n]\}
\end{equation} and, for Lemma \ref{thm:concentration_nodes}, we have that $ \Prc{C_{r}} \geq 1-1/n^2$.

We will first show the upper bound for the first inequality in \eqref{eq:bound_death}.
We notice that, for each $r \geq n \log n$ we have
\begin{align}
    & \Prc{|N_{T_{r+1}}|=|N_{T_r}|-1} \\ &= \Prc{|N_{T_{r+1}}|=|N_{T_r}|-1 \mid C_{r}}\Prc{C_{r}}+\Prc{|N_{T_{r+1}}|=|N_{T_r}|-1 \mid C_{r}^C}\Prc{C_{r}^C} 
    \label{eq:bound_leaving_node}
\end{align}
and we have that, for the law of total probability, and from the  equation above
\begin{align}
    \Prc{|N_{T_{r+1}}|=|N_{T_r}|-1}\leq \sum_{N=0.9n}^{1.1n}\Prc{|N_{T_{r+1}}|=|N_{T_r}|-1 \mid |N_{T_r}|=N}\Prc{|N_{T_r}|=N \mid C_{T_r}} + \frac{1}{n^2} .
    \label{eq:bound_leaving_node_2}
\end{align}
For Lemma \ref{lem:prop_Tn}, we have that $\Prc{|N_{T_{r+1}}|=|N_{T_r}|-1 \mid |N_{T_r}|=N}=(N/n)/(N/n+1)$, so, from the inequality above, we get
\begin{align}
     \Prc{|N_{T_{r+1}}|=|N_{T_r}|-1} \leq &\sum_{N=0.9n}^{1.1n}\frac{N}{N+n}\Prc{|N_{T_r}|=N \mid C_{T_r}}+\frac{1}{n^4} \leq 0.53\,.
     \label{eq:bound_leaving_node_3}
\end{align}
To show the lower bound, we utilize the upper bound above, getting
\begin{equation}
    \Prc{|N_{T_{r+1}}|=|N_{T_r}|+1}=1-\Prc{|N_{T_{r+1}}|=|N_{T_r}|-1} \geq 1-0.53=0.47.
\end{equation}

By a similar argument, we can show also the other inequalities in the statement of the lemma.

\subsubsection{Proof of Lemma \ref{lem:life_of_nodes}} \label{ssec:thm:life_of_nodes}

Let $r \geq7n \log n$. We know from Lemma \ref{thm:concentration_nodes} that $|N_{T_r}| \in [0.9n,1.1n]$ with probability at least $1-1/n^2$: we call this event $C_{T_r}$. For Lemma \ref{lem:prop_Tn} and for the memoryless property of the exponential distribution, we have that
\begin{equation}
\Prc{v \in N_{T_r} \mid v \in N_{T_{m-7n \log n}}} \leq \left(1-\frac{1}{2.2n}\right)^{7n \log n} \leq e^{-3.1 \log n}=\frac{1}{n^{3.1}}
\label{eq:L_v>4n log n}
\end{equation}
Now we want to prove that every node in $N_{T_r}$ has joined the network after time $T_{m-7n \log n}$ and, to do that, we have to do an union bound over all the nodes in the network. To doing that, we have to know how nodes in the network they are. So, for the law of total probability and for the concentration
\begin{align}
    &\Prc{\hbox{there exists a node $v \in N_{T_r}$ born before $T_{m-7n \log n}$}} \\ & \leq \Prc{\hbox{there exists a node $v \in N_{T_r}$ born before $T_{m-7n \log n}$} \mid C_{T_r}} + \frac{1}{n^2} 
\end{align}
and, since $C_{T_r}$ guarantees that the nodes in the network at time $T_{r}$ are at most $1.1n$, from equation \eqref{eq:L_v>4n log n} we get the lemma.

\subsection{Omitted proofs for  the Poisson model without edge regeneration}

\subsubsection{Proof of Lemma \ref{lem:isolated_nodes_poisson}}
\label{ssec:lem:isolated_nodes_poisson}
 
Let $r \geqslant 7n \log n$.  We define the following event
\begin{equation}
    \label{def:L_r}
    L_r=\{\hbox{each node in $N_{T_r}$ is born after time $T_{r-7n \log n}$}\}
\cap \{|N_{T_i}| \in [0.9n,1.1n] \hbox{ with }i=r-7n \log n,\dots,r\}\end{equation}
From Lemma \ref{thm:concentration_nodes} and Lemma \ref{lem:life_of_nodes} we have that $\Prc{L_r} \geq 1-1/n^2$.
From the event $L_r$ follows that, when each node in $N_{T_r}$ joined the network, the network was composed by at least $0.9n$ nodes and at most $1.1n$ nodes. 

Let $\varepsilon$ be an arbitrary value with $0 < \varepsilon  \leq 1/3$: we consider the set of the $\varepsilon n$ oldest nodes in $N_{T_r}$, and we call that set $H$. We define the following random variables
\begin{equation}
    A=\{v \in H \text{ s.t. $v$ has lifetime of at most $2n$ rounds}\}.
\end{equation}
Recalling the equation \eqref{eq:v_not_in_T_r+1} of Lemma \ref{lem:N_m+1_and_death}, which gives a bound on the lifetime (in rounds) of a node, we can apply a standard concentration argument (Theorem \ref{thm:chernoff}) getting that $|A| \geq \varepsilon n /2$ w.h.p.
Similarly to the proof of Lemma \ref{lem:isolated_nodes}, we first define the random variable below
\begin{equation}
    X=\{\text{number of nodes in $A$ that are isolated at time $T_r$ and for the rest of their lifetime}\}\,.
\end{equation}
The purpose is to show that, w.h.p., $X \geq \frac{1}{6}\varepsilon n e^{-2d}$. As in Lemma \ref{lem:isolated_nodes}, we will utilize the method of bounded differences, expressing $X$ as a function of $2n\cdot d$ independent random variables. First, to calculate the expectation of $X$, we introduce the following random variables, for each $v \in N_{T_r}$
\begin{equation}
    \Delta^{in}_v=\{\text{maximum in-degree of the node $v$ for all its lifetime}\} 
\end{equation}
\begin{equation}
    \Delta^{out}_v=\{\text{out-degree of the node $v$ at time $T_r$}\}.
\end{equation}
Since, if $\Delta_v^{out}=0$, then the node $v$ will have not out-edges from round $t$ to the rest of its lifetime, we can then write $X$ as a function of $\Delta^{in}_v$ and $\Delta^{out}_v$, with $v \in A$:
\begin{equation}
    X=\sum_{v \in A}\mathbb{1}_{\{\Delta^{in}_v=0\}}\mathbb{1}_{\{\Delta^{out}_v=0\}}\,.
\end{equation}
As each node sends its requests independently to the others, we have that 
\begin{equation}
    \Expcc{X \mid L_r}=\sum_{v \in A} \Prc{\Delta^{in}_v=0 \mid L_r}\Prc{\Delta^{out}_v=0 \mid L_r}\,.
    \label{eq:exp_isolated_nodes}
    \end{equation}
The probability of a node $v \in N_{T_r}$ to have in-degree $0$ for all their lifetime, conditional to $L_r$, is $
    \Prc{\Delta_v^{in}=0 \mid L_r}=\left(1-\frac{1}{0.9n}\right)^{2n d}\geq e^{-3d},
$ since each node $v \in A$ has lifetime of at most $2n$ rounds.
The probability of a node $v \in A$ to not have out-edges in the current round (conditional to $L_r$) is $
    \Prc{\Delta_i^{out}=0 \mid L_r}=\left(1-\frac{\varepsilon n}{0.9n}\right)^d \geq e^{- d/2}$, since $\varepsilon \leq 1/3$. So, since $|A| \geq \varepsilon n /2$ w.h.p. it follows from \eqref{eq:exp_isolated_nodes} that
    \begin{equation}
        \Expcc{X \mid C_t} \geq \frac{\varepsilon n}{3}e^{-2d}\,.
    \end{equation}
We now introduce the random variables $Y^j_v$, returning the index of the node to which the node $v$ send its $j$-th request. The random variables $\{Y_v^j : v \in N_{T_r}\cup N_{T_{r+1}}\cdots \cup N_{T_{r+2n}},j \in \{1,\dots,d\}\}$ are independent. We take $\mathbf{Y}$ as the vector of these random variables. Apparently, we can express $X$ as a function of $\mathbf{Y}$,
\begin{equation}
    X=f(\mathbf{Y}).
\end{equation}
Moreover, if the vectors $\mathbf{Y}$ and $\mathbf{Y'}$ differs only in one coordinate, we have that $|f(\mathbf{Y})-f(\mathbf{Y'})| \leq 2$.
This is because, in the worst case, an isolated node can change its destination from a dead node to an other isolated node: so, the number of isolated nodes decrease (or increase) of only $2$ units. 
So, applying Theorem \ref{thm:bounded_differences} we have that, if $\mu$ is a lower bound to $\Expcc{X \mid C_t}$, 
\begin{equation}
    \Prc{X \leq \mu - M \mid L_r} \leq e^{-\frac{2M^2}{4n d}}.
\end{equation}
So, taking $\mu=\frac{1}{3}\varepsilon n e^{-2d}$ and $M=\frac{1}{6}\varepsilon n e^{-2d}$ we get that
\begin{equation}
    \Prc{X \leq \frac{1}{6}\varepsilon n e^{-2d} \mid L_r}\leq e^{-n\frac{\varepsilon^2 e^{-4d}}{72d}}.
    \label{eq:upper_bound_isolated_nodes_conditioned}
\end{equation}
So, the number of isolated nodes is w.h.p. $X \geq \frac{1}{6}\varepsilon n e^{-2d}$.
This is because for the law of total probability and from \eqref{eq:upper_bound_isolated_nodes_conditioned}, if $n$ is large enough
\begin{align}
    &\Prc{X \leq \frac{1}{6}\varepsilon n e^{-d}}=\Prc{X \leq \frac{1}{6}\varepsilon n e^{-d} \mid L_r} + \frac{1}{n^2} \leq e^{-n\frac{\varepsilon^2e^{-2d}}{72d}}+\frac{1}{n^2} \leq \frac{2}{n^2}\,.
\end{align}
Taking $\varepsilon=1/3$ we get the lemma.
 
\subsubsection{Proof of Lemma 
\ref{lem:exp_large_subset_PDG}}
\label{ssec:lem:exp_large_subset_PDG}

We proceed as in the proof of Lemma \ref{lem:exp_large_sub_SDG}
for the \SDG \ model. We show that any two disjoint sets $S,T \subseteq N_{T_r}$, such that $n e^{-d/20} \leq |S| \leq |N_{T_r}|/2$, $|T|=0.1|S|$, and   $\partial_{out}(S) \subseteq T$, exist with negligible probability. For any $S$ and $T \subseteq N_{T_r}-S$ we again consider $A_{S,T}=\{\partial_{out}(S) \subseteq T\}$. 
Consider   the event
\[L_r=\{\hbox{each node in $N_{T_r}$ is born after time $T_{r-7n \log n}$}\}\cap \{|N_{T_i}| \in [0.9n,1.1n] \hbox{ with } i=r-7n \log n,\dots,r\},\]
and notice that Lemma \ref{lem:life_of_nodes} and Theorem 
\ref{thm:concentration_nodes} imply   $\Prc{L_r} \geq 1-1/n^2$. Thus,  from the law of total probability, 
\begin{align}
    \Prc{\min_{n e^{-d/20} \leq |S| \leq n/2}\frac{|\partial_{out}(S)|}{|S|}\leq 0.1} \leq \sum_{\substack{n e^{-d/20} \leq |S| \leq |N_{T_r}|/2 \\ |T|=0.1|S|}}\Prc{A_{S,T}\mid L_r}+\frac{1}{n^2} \, ,
    \label{eq:main_exp_sum}
\end{align}
To upper bound $\Prc{A_{S,T} \mid L_r}$ we define $P$ as the set $P=N_{T_r}-S-T$ and notice that  $|P| \geq 0.9n-1.1s$. The event $A_{S,T}$ implies that all the  edges   coming 
from   $S$ must go to $T$:   this is 
equivalent to say that there are no edges between $S$ and $P$. Since
\begin{equation}
    \left|\left\{(a,b) \mid a \in S, \ b \in P\right\} \right|=|S| \cdot |P|\, ,
\end{equation}
    two cases may arise: either 
\begin{enumerate}
    \item $\left| \left\{(a,b) \mid a \in S, \ b \in P, \hbox{$a$ younger then $b$}\right\}\right| \geq |S| \cdot |P| /2$, or 
    \item $\left| \left\{(a,b) \mid a \in S, \ b \in P, \hbox{$b$ younger then $a$}\right\}\right| \geq |S| \cdot |P| /2$.
\end{enumerate}
For each $a \in S$, let  $N_a$ be  the number of nodes in $P$ older than $a$. In the first case, we clearly have  
that $\sum_{a \in S}N_a \geq |S| \cdot |P| /2$. We can prove that  
\begin{equation}
\label{eq:bound_A_s_t_mid_L_r}
    \Prc{A_{S,T} \mid L_r}\leq \prod_{a \in S}\left(1-\frac{N_a}{1.1n}\right)^d \leq e^{-d\sum_{a \in S}N_a/(1.1n)} \leq e^{-d|S| \cdot |P|/2.2n}\, .
\end{equation}
 Indeed, as for the first inequality above,  for each $a \in S$,
 we considered the probability that a fixed request of node $a$ does not choose any node  in $P$ which is older than $a$ and we used the fact that, conditional to  the event $L_r$,     the probability that a node $a$ chooses any fixed older node  $v \in N_{T_r}$    is $\geq 1/1.1n$ (thanks to the event  ``$a$ chooses $v$ when it joins the network'').
Using a symmetric argument, we get the same claim for the second case.  Hence, placing \eqref{eq:bound_A_s_t_mid_L_r} into  \eqref{eq:main_exp_sum},
\begin{align}
&\Prc{\min_{n e^{-d/20} \leq |S| \leq n/2} \frac{|\partial_{out}(S)|}{|S|} \leq 0.1}=\sum_{s=n e^{-d/20}}^{n/2}\binom{1.1n}{s}\binom{1.1n-s}{0.1s}e^{-ds\frac{0.9n-1.1s}{2.2n}}+\frac{1}{n^2} \leq \frac{2}{n^2}  \, , 
\label{eq:final_sum}
\end{align}
 where the last inequality holds for a large enough $n$ and for any
 $d \geq 20$.  It can be easily proved by bounding each binomial coefficient with the bound $\binom{n}{k}\leq \left(\frac{n \cdot e}{k}\right)^k$ and by standard calculation.

\subsubsection{Proof of Theorem \ref{lem:flooding_not_terminate_poisson}}
\label{ssec:lem:flooding_not_terminate_poisson}
We begin by noting that the proof of Theorem 
\ref{lem:flooding_not_terminate_poisson} uses the original Definition 
\ref{def.async.flood} of asynchronous flooding.
Let $s_0$ be the   source node that joins the network at time $t_0=T_{r_0}$. Consider the event
    \begin{equation}
        C_{r_0}^{r_0+n}=\{|N_{T_i}| \in [0.9n,1.1n] \hbox{ with } i=r_0,\dots,r_0+n\}\, ,
    \end{equation}
	and notice that   Lemma \ref{thm:concentration_nodes} implies 
	that $\Prc{C_{r_0}^{r_0+n}}\geq 1-1/n^2$. We then consider the 
	events $A$ and $B$  we  defined in the proof of Lemma 
	\ref{lem:isolated_nodes}: $A=$"the source node $s_0$ has all its out-edges to nodes that are isolated at time $t_0$ and for the rest of their lifetime" and $B=$"there are at least $\frac{ne^{-2d}}{18}$ nodes in $N_{t_0}$ that are isolated at time $t_0$ and for the rest of their lifetime". 
Observe that Lemma \ref{lem:isolated_nodes_poisson} implies
    \begin{equation}
		\Prc{A \mid C_{r_0}^{r_0+n}} \geq \Prc{A \cap B \mid 
		C_{r_0}^{r_0+n}}=\Prc{A \mid B,C_{r_0}^{r_0+n}}\Prc{B \mid 
		C_{r_0}^{r_0+n}} \geq \frac{1}{2}\left(\frac{e^{-2d}}{18\cdot 
		1.1}\right)^{d} \, . \label{eq:lower_bound_B_cap_C_poisson}
    \end{equation}
We define  the event
\begin{equation}
        E = \{\hbox{the node $s_0$ will not get any in-edges for all its lifetime}\} \, .
    \end{equation}
Notice that the 
event $A \cap E$ imply that all the informed nodes $s_0,s_1,\dots,s_d$ are isolated for all their lifetime, and so: 
\begin{equation}
    \Prc{\hbox{$|I_t| \leq d+1$ for all $t \geq t_0$}}\geq \Prc{ A \cap E}.
    \label{eq:pr_I_T_r+n}
\end{equation}
Let $D_{s_0}$ the random variable which indicates the lifetime (in rounds) of the node $s_0$.
We recall that, for any two events $P,Q$ it holds $\Prc{P \cap Q} \geq \Prc{P}+\Prc{Q}-1$.
Since the life of each nodes follows an exponential random variable of parameter $1/n$, we have that, if $n$ is large enough, 
\begin{equation}
\label{eq:pr_D_s0_mid_Cr}
\Prc{D_{s_0} \leq n \mid C_{r_0}^{r_0+n}}\geq \Prc{D_{s_0} \leq n, C_{r_0}^{r_0+n}}\Prc{C_{r_0}^{r_0+n}} \geq \left((1-e^{-1})+\frac{1}{n^2}\right)\left(1-\frac{1}{n^2}\right) \geq \frac{1-e^{-1}}{2}\,.
\end{equation}
Moreover, 
\begin{align}
\label{eq:pr_E_mid_Cr_E}
    \Prc{E \mid C_{r_0}^{r_0+n}} \geq \Prc{E \mid C_{r_0}^{r_0+n},D_{s_0}\leq n}\Prc{D_{s_0}\leq n \mid C_{r_0}^{r_0+n}}
\end{align}
Since each edge chooses its destination uniform at random among the nodes in the network, we have that  
\begin{equation}
\label{eq:pr_E_mid_Cr_D_s0}
    \Prc{E \mid C_{r_0}^{r_0+n},D_{s_0}\leq n} \geq \left(1-\frac{1}{0.9n}\right)^{dn} \geq e^{-2d} 
\end{equation}
Replacing  and \eqref{eq:pr_E_mid_Cr_D_s0} and  \eqref{eq:pr_E_mid_Cr_E} in \eqref{eq:pr_D_s0_mid_Cr}  we get
\begin{equation}
\label{eq:pr_E_mid_C_r0}
    \Prc{E \mid C_{r_0}^{r_0+n}} \geq \frac{(1-e^{-1})e^{-2d}}{2}
\end{equation}

Finally, from \eqref{eq:pr_I_T_r+n}, since $A$ and $E$ are independent and because of \eqref{eq:lower_bound_B_cap_C_poisson} and \eqref{eq:pr_E_mid_C_r0}
\begin{equation*}
    \Prc{|I_t| \leq d+1 \hbox{ for all $t \geq t_0$}} \geq \Prc{A \cap E \mid C_{r_0}^{r_0+n}}\Prc{C_{r_0}^{r_0+n}}\geq \frac{(1-e^{-1})e^{-2d}}{8}\cdot\left(\frac{e^{-2d}}{20}\right)^d =c(d)\,.
\end{equation*}

Finally, as for the stated linear bound on the flooding time $\tau$, we observe this is an easy consequence of Lemma \ref{lem:isolated_nodes_poisson}. Indeed, in the network there are at least $\frac{1}{18}e^{-2d}n$ isolated nodes that will remain isolated for the rest of their lifetime: to have all the nodes informed, we have to wait that these nodes leave the network, and so $\tau=\Omega_d(n)$.

\subsubsection{Proof of Theorem \ref{thm:flood_poiss_noreg}}\label{ssec:thm:flood_poiss_noreg}
We begin by reminding the reader that, in order to account for the fact 
that a live node might die at any point of a given flooding interval, the proof of Theorem 
\ref{thm:flood_poiss_noreg} uses the discretized version of the 
flooding process described by Definition \ref{def:flooding_poisson}, 
which clearly provides a worst case scenario when we are interested in proving 
lower bounds on the extent and upper bounds on the speed of flooding.
The first observation is that, without edge regeneration, the 
distribution of links created by a node $v$ is uniform over the set of 
nodes that were in the network as $v$ joined. In particular, this 
distribution does not depend on past history of the network as is the 
case in the model with edge regeneration, where death of a node 
triggers reallocation of incoming links. 
We begin by showing a number of preliminary facts that will be useful in 
the remainder of this proof. The first is the following lemma, which is 
just a variant of Lemma \ref{lem:life_of_nodes} that is of easier use 
here.

\begin{lemma}[Nodes' lifetimes]\label{le:lifetime}
\label{lem:life_<nlogn}
Let $G_t=(N_t,E_t)$ be a $\mathcal{G}(n,d)$ Poisson random 
graph. If $n$ is large enough, for every $t\ge 4n \log n$, each 
node in $N_t$ has life $\leq 4n \log n$ with probability at least 
$1-1/n^2$.
\end{lemma}
\begin{proof}
We condition on the event $\mathcal E = (|N_{t-4\log n}|\le 2n)$. From Lemma 
\ref{thm:concentration_nodes}:
\[
\Prc{|N_{t-4\log n}|\le 2n}\ge 1 - \frac{1}{2n^2}.
\]
Next, for every $i\in N_t$ we define a binary variable 
$L_i$, such that $L_i = 1$ if $i$ is alive at time $t$. Note that, 
conditioned on the event $\mathcal E$, each $\Prc{L_i} = 1$ is 
exponential with parameter $\mu = 1/n$, hence:
\[
	\Prc{L_i = 1 \mid \mathcal{E}} = e^{-\frac{4n\log n}{n}} = 
	\frac{1}{n^4},
\] 
which implies:
\[
	\Prc{\exists i\in N_{t-4n\log n}: L_i = 1 \mid \mathcal{E}}\le 
	\frac{2}{n^3}.
\]
Denote by $\mathcal{O}$ the event that there exists a node with age 
higher than $4n\log n$ at time $t$. We have:
\[
	\Prc{\mathcal{O}}\le \Prc{\exists i\in N_{t-4n\log n}: L_i = 1 \mid 
	\mathcal{E}}\Prc{\mathcal{E}} + 
	\Prc{\neg\mathcal{E}}\le\frac{1}{n^2}.
\]
\end{proof}

The following fact is instead a simple corollary
of Lemma \ref{thm:concentration_nodes}.

\begin{fact}\label{fa:no_nodes}
With probability at least $1 - 1/n^2$, $0.9n\le N_t\le 1.1n$, for every 
$t\in [t_0 - n^2, t_0]$.
\end{fact}

Now, assume $N_{t_0} = m$ and recall that $m\in [0.9n, 1.1n]$ from Fact 
\ref{fa:no_nodes}.
Next, we consider the number of nodes that die in the interval $[t_0, 
t_0 + \log n]$. 
\begin{lemma}\label{le:no_deaths}
With probability at least $1 - 1/n^c$, at most $4\log n$ nodes die in 
the interval $[t_0, t_0 + \log n]$.
\end{lemma}
\begin{proof}
Consider the generic node $i\in N_{t_0}$. We define a binary variable 
$L_i(\tau)$, such that $L_i(\tau) = 1$ if $i$ is alive at time $t_0 + \tau$, 
$L_i(\tau) = 0$ otherwise. Clearly, since the death process follows the 
exponential distribution with parameter $\mu = 1/n$ and is thus memoryless, we have:
\[
	\Prc{L_i(\tau) = 1} = e^{-\frac{\tau}{n}}\ge 1 - \frac{\tau}{n}.,
\]
Denote by $Z$ the number of nodes that die in the interval $[t_0, t_0 + a\log n]$. Setting $\tau = a\log n$ we 
immediately have:
\[
	\Expcc{Z}\le m - m + \frac{m\log n}{n}\le 1.1\log n,
\]
where we used $m\le 1.1n$ with probability $1 - 1/n^2$. Finally, since 
$Z = \sum_{i\in N_{t_0}}L_i(\log n)$ and since the $L_i(\log n)$'s are 
independent, a simple application of Chernoff's bound yields:
\[
	\Prc{Z > 4\log n}\le e^{-1.1\log n},
\]
which proves our claim.
\end{proof}

Finally, we bound the number of nodes' arrivals in an interval of 
logarithmic duration starting at time $t_0$.
\begin{lemma}\label{le:no_births}
With probability at least $1 - 1/n^c$, at most $4\log n$ nodes join 
the network in the interval $[t_0, t_0 + \log n]$.
\end{lemma}
\begin{proof}
Nodes enter the system according to a Poisson process with rate $1$ in 
each step, which corresponds to a Poisson process with rate $4\log 
n$ over an interval of $4\log n$ steps. This fact and tail bounds for 
the Poisson distribution imply the claim. 
\end{proof}

We next prove that, starting from $s$ at time $t_0$, with probability $1 - 2e^{-\frac{d}{576}} - 
o(1)$ we reach a fraction $1 - e^{-\frac{d}{20}}$ of the nodes that are in the network at time 
$t_0 + T$, where $T = \bigO(\log n/\log d + d)$. To this purpose, we apply the ``onion skin'' technique as 
we did in Section \ref{subse:stream_without_flood}, although with a 
number of more or less significant changes. To begin, differently from 
the \SDG\ model, we consider all nodes that are in the system when the informed node joins the network 
at time $t_0$. On the other hand, we completely disregard nodes that were born in the 
interval $[t_0, t_0 + \log n]$ (i.e., we don't count them as hits, though 
we keep into account that they remove probability mass from 
destinations in $N_{t_0}$), since their number is negligible with high 
probability from Lemma \ref{le:no_births}. Conversely, we do need to 
consider nodes that die in the 
interval $[t_0, t_0 + \log n]$, since their failure might in principle 
significantly affect flooding. Removing them at the onset or at the end 
may be tricky, since in both cases we need to argue that their removal 
has no significant topological effects (we cannot simply assume an adversarial 
removal). To sidestep these challenges, we remove each of them with 
probability $\log n/n$ (this is an upper bound on the 
probability that a node dies in the interval $[t_0, t_0 + \log n]$) as we 
proceed.

In particular, we define an ``extended'' \emph{onion-skin process}, which is dominated by 
the real flooding process. In this process, new links are generated (with 
restrictions with respect to the original process) as flooding unfolds 
(deferred decisions), while each new informed node has a chance 
to independently die with probability $\log n/n$ upon receiving the 
information for the first time, thus without any chance 
of informing other nodes of the network, which clearly is a worst-case 
scenario.

\paragraph{The onion-skin process.}
Denote by $S$ the set of nodes that were in the network at time $t_0$ and 
let $m = |S|$. From Fact \ref{fa:no_nodes}, we know that $m \in [0.9n, 
1.1n]$ with probability at 
least $1 - 1/n^2$. We now build a map $h: S\rightarrow [m]$, so that 
for $v\in S$, $h(v) = i$ if $v$ is the $i$-th youngest node in the 
system at time $t$. Note that $h(s) = 1$ by definition. We next define $Y 
= \{v\in N_t: h(v)\le m/2\}$ as the subset of \emph{young} nodes 
and $O = \{v\in N_t: h(v)\ge m/2 + 1\}$ as the subset of \emph{old} nodes. 

Starting from $s$, the onion-skin process builds a 
connected, bipartite graph, so that young nodes are only connected 
to old ones. In particular, each realization of this process 
generates a subset of the edges generated by the original topology 
dynamics. 
Moreover, each iteration of the process corresponds to a partial 
flooding in the original graph. Flooding is partial since i) the 
network uses a subset of the edges that would be present in the 
original graph and ii) every newly informed node tosses a coin and 
dies, with probability equal to the overall probability of dying in the 
interval $[t_0, t_0 + \log n]$.

The process unfolds over a suitble number $k$ of phases, with $k = 
\bigO(\log n/\log d)$. Each phase corresponds to $2$ 
flooding rounds, each consisting of two steps. In the following, we 
denote by $Y_k\subseteq Y$ and $O_k\subseteq O$ the subsets of young 
and old nodes that are informed by the end of phase $k$, respectively. In the remainder, we let 
$O_{-1} = \emptyset$ for notational convenience and, without loss of 
generality, we use the interval $[d]$ to number the links established by each 
vertex.

\begin{center}
\fbox{
\begin{minipage}{15cm}
\textbf{Onion-skin process (extended version)}:
	\begin{description}
		\item[Phase $\mathbf{0}$:] $Y_0 = \{s\}$; $O_0$ is obtained as follows:
			\begin{enumerate}
				\item $s$ establishes $d$ links. We let $Z_0\subset O$ denote the 
				subset of old nodes that are destinations of these links. Links with 
				endpoints in $Y$ are discarded;
				\item Let $R$ the subset obtained by removing each vertex in $Z_0$ (and 
				the just established links) independently, with probability $\log n/n$. $O_0 = Z_0 - R$;
			\end{enumerate}
		\item[Phase $\mathbf{k}\ge 1$:] (two flooding steps)
			\begin{enumerate} 
				\item[1.a] Each node in $Y - Y_{k-1}$ establishes links $\{1,\ldots, 
				d/2\}$. Denote by $W_k\subseteq Y - Y_{k-1}$ the subset of nodes with 
				at least one link to nodes in $O_{k-1} - O_{k-2}$. Links to nodes in 
				$Y$ are again discarded;
				\item[1.b] Let $R$ the subset obtained by removing each vertex in $W_k$ (and 
				the just established links) independently, with probability $\log 
				n/n$. $Y_k - Y_{k-1} = W_k - R$;
				\item[2.a] Each node in $Y_k - Y_{k-1}$ establishes links $\{d/2 + 
				1,\ldots , d\}$. Denote by $Z_k$ the subset of nodes in $O - O_{k-1}$ 
				that are reached by at least one such link. Links to nodes in $Y$ are 
				discarded;
				\item[2.b] Let $R$ the subset obtained by removing each node in $Z_k$ 
				(and the links just established) independently, with probability $\log 
				n/n$. $O_k - O_{k-1} = Z_k - R$;
			\end{enumerate}
	\end{description}
\end{minipage}\label{proc:onion_skin}
}
\end{center}
A couple remarks are in order. First of all, we are using the principle 
of deferred decisions, delaying decision as to the establishment of a 
link $(u, v)$ to the moment one of its endpoints is informed in the 
flooding process. On the other hand, this means that the probability 
that the $j$-th link originating from $u$ has $v$ as destination is 
equal to $1/N_{T_r}$ if $u$'s arrival in the network corresponds to the $r$-th event. We use 
Lemma \ref{le:lifetime} and Fact \ref{fa:no_nodes} to claim that, with 
probability at least $1 - 1/n^2$ the above probability fell in the interval $\left[\frac{1}{1.1n}, 
\frac{1}{0.9n}\right]$ for every node in the network at time $t_0$. 
Second, the aforementioned probabilities do not depend on nodes that 
joined or left the network in the interval $[t_0, t_0 + \log n]$. 
Accordingly, all events we consider in the remainder of this proof are 
conditioned on $\{N_t\in [0.9n, 1.1n], \forall t\in [t_0 - n^2, t_0]\}$ and 
on the ages of all nodes alive at time $t_0$ belonging to the interval 
$[1,\ldots , 4n\log n]$. We omit these conditionings for ease of notation.

\medskip

We next analyze Phase 0 and the generic Phase $k$ separately. 
To this purpose, we use the same random variables we defined in
the proof of Theorem \ref{thm:not_flooding_purestreaming}, whose 
definition is repeated here to make the proof self-contained.
For each $i=1,\dots,d$, for each node $v \in N_{t_0}$ 
and for each set $A \subseteq N_{t_0}$, we define the Bernoulli random 
variable $R_{v,A}$ as follows:
\begin{equation}\label{def:R_v,A}
    R_{v,A}=\left\{
    \begin{array}{l}
		1\quad \hbox{if $x\ge 1$ links in }\{\frac{d}{2},\dots,d\} \hbox{ from $v$ have destination in 
		$A$}\\
		0 \quad \hbox{otherwise}
	\end{array}\right.
\end{equation}
Moreover, for each $i=1,\dots,d$ and for every $v\in N_{t_0}$ that 
joined the network at time $\hat{t}\le t_0$, we 
define the variable $A_v^{(i)}\in N_{\hat{t}}$ as follows: $A_v^{(i)} = w$ if 
$w\in N_{\hat{t}}$ is the destination of the $i$-th link of $v$.

\paragraph{Analysis of phase $0$.} For phase $0$ we begin with the 
following claim.
\begin{claim}\label{cl:phase_0}
The following holds at the end of phase $0$:
\begin{equation}\label{eq:phase_0}
	\Prc{|O_0|\ge\frac{d}{16}}\ge\left(1 - \frac{2\log 
	n}{n}\right)\left(1 - e^{-\frac{d}{128}}\right).
\end{equation}
\end{claim}
\begin{proof}
We begin with step 1. Assume $v\in O$. We have:
\[
	\Prc{\exists i\in [d]: A_s^{(i)} = v} = 1 - \left(1 - \frac{1}{m}\right)^d,
\]
which implies 
\[
	\Expcc{Z_0} = |O|\left(1 - \left(1 - 
	\frac{1}{m}\right)^d\right)\ge |O|\frac{d}{2m} = \frac{d}{4}, 
\]
where the last equality follows since $|O| = m/2$. We next bound the 
probability that $Z_0$ is smaller than $d/8$. To this purpose, we 
cannot simply apply a Chernoff bound to the binary variables that 
describe whether or not a node $v\in O$ was the recipient of at least 
one link originating from $s$, since these are not independent. We 
instead resort to Theorem \ref{thm:bounded_differences}. In particular, 
similarly to what we did in the proof of Theorem 
\ref{thm:not_flooding_purestreaming}, we define the function 
$f(A_s^{(1)},\ldots , A_s^{(d)}) = |Z_0|$. Clearly, $f$ is well-defined and it 
satisfies the Lipschitz condition with values $\beta_1 =\cdots = 
\beta_d = 1$, since changing the destination of one link can affect the 
value of $|Z_0|$ by at most $1$. We can thus apply Theorem 
\ref{thm:bounded_differences} to obtain:
\[
	\Prc{|Z_0| < \frac{d}{8}}\le\Prc{|Z_0| < \Expcc{|Z_0|} - 
	\frac{d}{8}}\le e^{-\frac{d}{128}}.
\]

We next obtain $O_0$. If $|Z_0| = x$ and $R$ is the number of nodes in $Z_0$ that are 
removed, we have:
\[
	\Expcc{R}\le \frac{ax\log n}{n}.
\]
Applying Markov's inequality immediately yields
\[
	\Prc{R\ge\frac{x}{2}}\le\frac{2\log n}{n},
\]
whence the thesis immediately follows.
\end{proof}

\paragraph{Analysis of Phase $k$ - step 1.}
We next examine the number 
of nodes in $Y - Y_{k-1}$ that connect to nodes in $|O_{k-1} - 
O_{k-2}|$ using their links belonging to the subset $\{d/2 + 1,\ldots , 
d\}$. 
\begin{claim}\label{cl:phase_k_step_1}
Assume that $|O_{k-1} - O_{k-2}| = y$ and $|Y_{k-1}|\le m/10$.
For sufficiently large $n$, the following holds at the end of phase $k$:
\begin{align}\label{eq:phasek_step1}
	&\Prc{|Y_k - Y_{k-1}|\ge\frac{yd}{48} \mid |O_{k-1} - O_{k-2}| = 
	y}\ge\left(1 - \frac{2\log n}{n}\right)\left(1 - 
	e^{-\frac{yd}{48}}\right), & y\le\frac{1.1n}{d}\\
	&\Prc{|Y_k - Y_{k-1}|\ge\frac{m}{20} \mid |O_{k-1} - O_{k-2}| = 
	y}\ge\left(1 - \frac{2\log 
	n}{n}\right)\left(1 - e^{-\frac{n}{240}}\right), & y > \frac{1.1n}{d}.
\end{align}
\end{claim}
\begin{proof}
It should be noted that i) no links originating from nodes in $Y - Y_{k-1}$ have been established so 
far and ii) at this point we know that, for each $u\in Y - Y_{k-1}$, 
none of its links in $\{d/2 + 1,\ldots , d\}$ had destination in $O_{k-2}$, thus the 
probability of pointing to some vertex in $O_{k-1} - O_{k-2}$ can only 
be magnified.
With these premises, if $u\in Y - Y_{k-1}$ we have:
\[
	\Prc{R_{v, O_{k-1} - O_{k-2}} = 1 \mid |O_{k-1} - O_{k-2}| = y} = 
	1 - \left(1 - \frac{y}{\ell}\right)^{\frac{d}{2}} > 1 - 
	e^{-\frac{yd}{2\ell}},
\]
where $\ell\in [0.9n, 1.1n]$ denotes the number of nodes when $u$ 
joined the network. If $y\le 1.1n/d$ we have $\frac{yd}{2\ell} < 
0.61$, which implies $e^{-\frac{yd}{2\ell}} < 
\frac{2}{3}\frac{yd}{2\ell}$, whence:
\[
	\Prc{R_{v, O_{k-1} - O_{k-2}} = 1 \mid |O_{k-1} - O_{k-2}| = y} > 
	\frac{yd}{3\ell} > \frac{yd}{3.3 n}.
\]
As a consequence, if $|Y_{k-1}|\le m/10$:
\[
	\Expcc{|W_k| \mid |O_{k-1} - O_{k-2}| = y} > |Y - 
	Y_{k-1}|\frac{yd}{3.3 n}\ge \frac{2}{5}\frac{myd}{3.3 n} > 
	\frac{yd}{12},
\]
where in the last inequality we used $m\ge 0.9n$.
On the other hand, the $R_{v, O_{k-1} - O_{k-2}}$'s are independent. 
We can thus apply a standard Chernoff bound to obtain:
\[
	\Prc{|W_k| < \frac{yd}{24} \mid |O_{k-1} - O_{k-2}| = y}\le 
	e^{-\frac{yd}{48}}.
\]
We next consider step 1.b. Assume $|W_k| = x$. If $R$ denotes the 
number of vertices removed from $W_k$ we can argue exactly like for the 
analysis of step 2 of Phase $0$ to obtain:
\begin{equation*}
	\Prc{|Y_k - Y_{k-1}|\ge\frac{yd}{48} \mid |O_{k-1} - O_{k-2}| = 
	y}\ge\left(1 - \frac{2\log 
	n}{n}\right)\left(1 - e^{-\frac{yd}{48}}\right).
\end{equation*}

Conversely, if $y > 1.1n/d$ we have $\frac{yd}{2\ell} > 1/2$, so that
\[
	\Prc{R_{v, O_{k-1} - O_{k-2}} = 1 \mid |O_{k-1} - O_{k-2}| = y} > 
	1 - e^{-1/2}.
\]
In this case:
\[
	\Expcc{|W_k| \mid |O_{k-1} - O_{k-2}| = y}\ge |Y - 
	Y_{k-1}|\left(1 - \frac{1}{\sqrt{e}}\right)\ge\frac{2}{5}m\left(1 - 
	\frac{1}{\sqrt{e}}\right) > \frac{0.78}{5}m,
\]
which is both larger than $(2/15)m$ and $(2/15)n$ (the latter follows 
since $m\ge 0.9n$). We therefore have:
\begin{equation*}
	\Prc{|W_k|\ge\frac{m}{10} \mid |O_{k-1} - O_{k-2}| = 
	y}\ge\Prc{|W_k|\ge\frac{3}{4} \Expcc{|W_k|}\mid |O_{k-1} - O_{k-2}| = 
	y}\le e^{-\frac{n}{240}}.
\end{equation*}
Hence, arguing as we did before, we conclude:
\begin{equation*}
	\Prc{|Y_k - Y_{k-1}|\ge\frac{m}{20} \mid |O_{k-1} - O_{k-2}| = 
	y}\ge\left(1 - \frac{2\log 
	n}{n}\right)\left(1 - e^{-\frac{n}{240}}\right).
\end{equation*}

\end{proof}

\paragraph{Analysis of Phase $k$ - step 2.}
In this case, we are interested in nodes from $Y_k - Y_{k-1}$ that connect to nodes 
in $O - O_{k-1}$ using their first $d/2$ links. 
\begin{claim}\label{cl:phase_k_step_2}
Assume $|Y_k - Y_{k-1}| = x$ and $|O_{k-1}|\le m/10$.
For sufficiently large $n$, the following holds at the end of phase $k$:
\begin{align}\label{eq:phasek_step2}
	&\Prc{|O_k - O_{k-1}|\ge \frac{xd}{48} \mid |Y_k - Y_{k-1}| = x}\ge \left(1 - \frac{2\log 
	n}{n}\right)\left(1 - e^{-\frac{xd}{576}}\right), &
	x\le\frac{1.1n}{d},\\
	&\Prc{|O_k - O_{k-1}|\ge\frac{m}{20} \mid |Y{k} - Y_{k-1}| = 
	x}\ge\left(1 - \frac{2\log n}{n}\right)\left(1 - 
	e^{-\sqrt{n}}\right), & x > \frac{1.1n}{d}.
\end{align}
\end{claim}
\begin{proof}
Assume $v\in O - O_{k-1}$. We have:
\[
	\Prc{\exists u\in Y_k - Y_{k-1}, \exists i\in [d/2]: A_u^{(i)} = v \mid |Y_k - Y_{k-1}| = x}\ge 
	1 - \left(1 - \frac{1}{1.1n}\right)^{\frac{xd}{2}} > 1 - 
	e^{-\frac{xd}{2.2n}}.
\]
As a consequence:
\[
	\Expcc{|Z_k| \mid |Y_k - Y_{k-1}| = x} > |O - O_{k-1}|(1 - e^{-\frac{xd}{2.2n}}).
\]
We consider two cases, as we did in the proof of Claim 
\ref{cl:phase_k_step_1}. If $x\le\frac{1.1n}{d}$ we have:
\[
	\Expcc{|Z_k| \mid |Y_k - Y_{k-1}| = x}\ge |O - 
	O_{k-1}|\frac{xd}{3.3n}\ge\frac{2}{5}\frac{mxd}{3.3n} > 
	\frac{xd}{12},
\]
where similarly to Claim \ref{cl:phase_k_step_1}, we used $|O_{k-1}|\le 
m/10$ and $m\ge 0.9n$.
where the second inequality follows since i) $|O| = m/2$ and ii) we are 
assuming $|Y_{k-1}|\le m/10$. Again and differently from Claim 
\ref{cl:phase_k_step_1}, we cannot simply concentrate, 
since the events $\{A_u^{(i)} = v\}$ are negatively 
correlated as $v$ varies over $O - O_{k-1}$. We again resort to Theorem 
\ref{thm:bounded_differences}. In this case, we have $\frac{xd}{2}$ 
links that are established (independently of each other) from vertices 
in $Y_k - Y_{k-1}$. Consider the $xd/2$ variables $\{A_u^{(i)}\}_{u\in 
Y_k - Y_{k-1}, i\in [d/2]}$. The domain of $A_u^{(i)}$ is the set $N_t$ 
if $u$ joined the system at time $t$, where $0.9n\le |N_t|\le 1.1n$ 
with high probability, from  Lemma \ref{le:lifetime} and Fact 
\ref{fa:no_nodes}. We then define the function $f(\{A_u^{(i)}\}_{u\in 
Y_k - Y_{k-1}, i\in [d/2]}) = |Z_k|$. Like in 
the analysis of step 1 of Phase 0, we note that $f$ satisfies the 
Lipschitz condition with constants $\beta_1 =\cdots =\beta_{\frac{xd}{2}} = 1$. We can thus 
apply Theorem \ref{thm:bounded_differences} to obtain:
\begin{align*}
	&\Prc{|Z_k| < \frac{xd}{24} \mid |Y_k - Y_{k-1}| = 
	x}\le\Prc{|Z_k| < \Expcc{|Z_k| \mid |Y_k - Y_{k-1}| = x} - 
	\frac{xd}{24} \mid |Y_k - Y_{k-1}| = x}\\
	&\le e^{-\frac{xd}{576}}.
\end{align*}
Finally, we remove nodes from $Z_k$ exactly as we did in Phase 0 and in 
step 1.b of Phase $k$. The analysis proceeds exactly the same, so that 
we can conclude:
\begin{equation*}
	\Prc{|O_k - O_{k-1}|\ge \frac{xd}{48} \mid |Y_k - Y_{k-1}| = x}\ge \left(1 - \frac{2\log 
	n}{n}\right)\left(1 - e^{-\frac{xd}{576}}\right).
\end{equation*}

Conversely, if $x > 1.1n/d$ we have $\frac{xd}{d} > \frac{1}{2}$, 
so that:
\[
	\Prc{\exists u\in Y_k - Y_{k-1}, \exists i\in [d/2]: A_u^{(i)} = v 
	\mid |Y_k - Y_{k-1}| = x} > 1 -	e^{-1/2},
\]
whence:
\[
	\Expcc{|Z_k| \mid |Y_{k} - Y_{k-1}| = x}\ge |O - 
	O_{k-1}|\left(1 - \frac{1}{\sqrt{e}}\right)\ge\frac{2}{5}m\left(1 - 
	\frac{1}{\sqrt{e}}\right) > \frac{2}{15}m.
\]
Next, application of Theorem Theorem \ref{thm:bounded_differences} 
yields in this case:
\begin{align*}
	&\Prc{|Z_k| < \frac{m}{10} \mid |Y_{k} - Y_{k-1}| = 
	x}\le \Prc{|Z_k| < \Expcc{|Z_k|} - \frac{m}{30} \mid |Y_{k} - Y_{k-1}| = 
	x}\\
	&\le e^{-\frac{2m^2}{900xd}}\le e^{-\frac{m}{225d}} < 
	e^{-\sqrt{n}},
\end{align*}
where the third inequality follows from $x\le |Y|\le m/2$, while the 
fourth follows since $d$ is a constant and $m\ge 0.9$, so it holds for 
sufficiently large $n$. Finally, proceeding like in Claim 
\ref{cl:phase_k_step_1} (analysis of step 1.b) we obtain:

\begin{equation*}
	\Prc{|O_k - O_{k-1}|\ge\frac{m}{20} \mid |Y_{k} - Y_{k-1}| = 
	x}\ge\left(1 - \frac{2\log 
	n}{n}\right)\left(1 - e^{-\sqrt{n}}\right).
\end{equation*}
\end{proof}

Finally, Claims \ref{cl:phase_0}, \ref{cl:phase_k_step_1} and 
\ref{cl:phase_k_step_2} imply the following result.
\begin{lemma}[Flooding completes, part 1]\label{le:flood_without_part_1}
For sufficiently large and constant $d\ge 1152$, there exist a constant $c > 0$ 
and $k = \bigO(\log n/\log d)$, such that:
\begin{equation}\label{eq:flood_without_part_1}
	\Prc{|Y_k|\ge \left(\frac{m}{20}\right)\cap 
|O_k|\ge \left(\frac{m}{20}\right)}\ge 1 - 2e^{-\frac{d}{576}} - o(1).
\end{equation}
\end{lemma}
\begin{proof}
Consider the generic $i$-th phase and assume i) $|Y_{i-1}|, |O_{i-1}| < 
1.1n/d\le n/10$ (i.e., $d\ge 11$) and ii) $|O_{i-1} - O_{i-1}|\ge 
\left(\frac{d}{48}\right)^{2i-1}$. Then, Claims \ref{cl:phase_0}, 
\ref{cl:phase_k_step_1} and \ref{cl:phase_k_step_2} imply that, conditioned 
to i) and ii) we have:

\begin{align}\label{eq:joint_prob}
	&\Prc{|Y_i - Y_{i-1}|\ge \left(\frac{d}{48}\right)^{2i}\cap 
	|O_i - O_{i-1}|\ge \left(\frac{d}{48}\right)^{2i+1}}\\
	&\ge\left(1 - \frac{2\log n}{n}\right)^2\left(1 - 
	e^{-\frac{1}{12}\left(\frac{d}{48}\right)^{2i}}\right)\left(1 - 
	e^{-\frac{1}{12}\left(\frac{d}{48}\right)^{2i+1}}\right).
\end{align}
Since $|Y_i| \ge |Y_{i} - Y_{i-1}|$ and $|O_i|\ge |O_i - O_{i-1}|$, 
using the chain rule and \eqref{eq:joint_prob}, for
\[
	k\ge\frac{\log\left(\frac{n}{20}\right)}{2\log\left(\frac{d}{48}\right)},
\]
at the end of phase $k$ we have:
\begin{equation}\label{eq:chain_rule}
	\Prc{|Y_k|\ge\frac{m}{20}\cap|O_k|\ge \frac{m}{20}}\ge\left(1 - e^{-\sqrt{n}}\right)^2\left(1 - \frac{2\log n}{n}\right)^{2k+1}\left(1 - 
	e^{-\frac{d}{128}}\right)\prod_{i=1}^{2k+1}\left(1 - 
	e^{-\frac{1}{12}\left(\frac{d}{48}\right)^{i}}\right).
\end{equation}

Let $P = \prod_{i=1}^{2k+1}\left(1 - 
e^{-\frac{1}{12}\left(\frac{d}{48}\right)^{i}}\right)$. We next show 
that $P\ge c$ for a suitable constant $c$. This is equivalent to 
showing that $-\log P\le c'$, where $c' = \log\frac{1}{c}$. We have:
\begin{align*}
	&-\log P = -\sum_{i=1}^{2k+1}\log\left(1 - e^{-\frac{1}{12}\left(\frac{d}{48}\right)^{i}}\right) = 
	\sum_{i=1}^{2k+1}\log\frac{1}{1 - 
	e^{-\frac{1}{12}\left(\frac{d}{48}\right)^{i}}}\le 2\sum_{i=1}^{2k+1}
	e^{-\frac{1}{12}\left(\frac{d}{48}\right)^{i}}\\
	& = 2\sum_{i=0}^{2k}e^{-\frac{d}{576}\left(\frac{d}{48}\right)^{i}} 
	< 2e^{-\frac{d}{576}}\sum_{i=0}^{\infty}e^{-\left(\frac{d}{48}\right)^{i}} < 2e^{-\frac{d}{576}},
\end{align*}
where the third inequality follows from $x\ge\log\frac{2}{2-x}$ for 
$0\le x\le 1$ which, considered that $i\ge 1$, in our case is certainly 
satisfied if $\left(\frac{d}{48}\right)^2\ge \frac{1}{12}$, i.e., for 
$d$ a sufficiently large, absolute constant. The fifth inequality 
follows since $\frac{d}{576} + \left(\frac{d}{48}\right)^{i}\le 
\frac{d}{576}\left(\frac{d}{48}\right)^{i}$, whenever $d/576\ge 2$, 
while the last inequality follows since the double exponential is dominated by a simple one, 
summing to a constant not exceeding $1$. Proceeding like in the final 
steps of Claim \ref{claim:product} we obtain $P > 1 - 
2e^{-\frac{d}{576}}$. Since
\[
	\Prc{|Y_k|\ge \left(\frac{m}{20}\right)\cap |O_k|\ge 
	\left(\frac{m}{20}\right)}\ge \left(1 - e^{-\sqrt{n}}\right)^2\left(1 - \frac{2\log n}{n}\right)^{2k+1}\left(1 - 
	e^{-\frac{d}{128}}\right)P,
\]
the proof follows for sufficiently large $n$.
\end{proof}

\paragraph{Information spreading via the expansion of large subsets.}
We finally show that, if at least $m/10$ nodes become informed (which 
occurs with probability at least $1 - 2e^{-\frac{d}{576}} - o(1)$ within $\bigO(\log n/\log d)$ flooding 
steps from Lemma \ref{le:flood_without_part_1}), then at least $(1 - 
e^{-\frac{d}{20}})m$ nodes become informed within a further, constant number of flooding 
steps, w.h.p. 

To prove this, we leverage Lemma \ref{lem:exp_large_subset_PDG} and we 
proceed along the same lines as Theorem \ref{apx:thm:aeflooding} and in 
particular Lemma \ref{le:big_set_exp}, albeit with the following 
difference: in each flooding step, we need to account for the fact that 
a node that was present at time $t_0$ might die before being informed 
(or right upon being informed) and thus be unable to contribute to the 
flooding process. We have the following
\begin{lemma}[Flooding completes, part 
2]\label{le:big_set_flood_without_poiss}
Under the hypotheses of Theorem \ref{thm:flood_poiss_noreg}, 
for some $\tau_2=\bigO(d)$ and for $\tau_1=\bigO(\log n/\log d)$ as in 
Lemma \ref{le:flood_without_part_1}, we have:
\begin{equation}
    \Prc{|I_{t_0+\tau_1+\tau_2}| \geq (1 - 
e^{-\frac{d}{20}})m} \geq 1 - 2e^{-\frac{d}{576}} - o(1)\,.
\end{equation}
\end{lemma}
\begin{proof}
The proof proceeds along the very same lines as Lemma 
\ref{le:big_set_exp}, hence we only discuss the differences. 

If we start with $n/10$ informed nodes, in each flooding step the set of 
informed nodes increases by a constant factor, exactly like in the 
proof of Lemma \ref{le:big_set_exp}, but this time we use Lemma \ref{lem:exp_large_subset_PDG}
for the expansion. The main difference is that, this 
time, each newly informed node has a chance to die. Assume $S$ is the 
set of the newly informed nodes at the end of the generic 
expansion/flooding step. Since we are interested in a constant number 
of rounds, we can handle nodes' deaths in a simplified way with respect 
to what we did in the proof of Lemma \ref{le:flood_without_part_1}. In 
particular, with high probability, at most $4\log n$ nodes are removed 
from $S$ in worst-case fashion, but this still implies that, with high 
probability, $|S| - 4\log n$ new nodes have been informed, 
thus the set of informed nodes has increased by a constant factor, 
since $|S| = \Omega(n)$. If we iterate over a sufficiently large, 
\emph{constant} number $\tau_2 = \bigO(d)$ of flooding steps, we can conclude, 
like in the proof of Lemma \ref{le:big_set_exp}, that $(1 - 
e^{-\frac{d}{20}})m$ nodes have been informed within time $t_0 + \tau_1 
+ \tau_2$.
\end{proof}

This last step concludes the proof of Theorem 
\ref{thm:flood_poiss_noreg}.

\subsection{Proofs for the Poisson model with edge regeneration}

\subsubsection{Proof of Lemma \ref{lem:existence_edge_hyp}}
\label{sssec:lem:existence_edge_hyp}

We define the following event avoiding the use of 
\begin{equation}
    A_{u,v}=\{\text{a fixed request   of $u$ has  destination $v$ at time $T_r$}\}\, .
\end{equation}
Notice that we avoid to index the specific  request since the considered graph process is perfectly symmetric w.r.t. the $d$ random requests of every node. 
We first bound the probability that a fixed request of $u$ has destination $v$ when $v$ is younger than $u$.
Calling    $L_{r}$    the event 
\begin{equation}
    L_r=\{\hbox{each node in $N_{T_r}$ is born after time $T_{r-7n \log n}$}\}\cap \{|N_{T_i} \in [0.9n,1.1n] \hbox{ with } i=r-7n \log n,\dots,r\}
\end{equation}
from Lemma \ref{thm:concentration_nodes} and Lemma \ref{lem:life_of_nodes} we get that $\Prc{L_r} \geq 1-1/n^2$. We notice that the event $L_r$ means that, when each node in $N_{T_r}$ joined the network, the network was composed by at least $0.9n$ nodes and at most $1.1n$ nodes. From the law of total probability, we have
\begin{align}
    \Prc{A_{u,v}} \leq \Prc{A_{u,v} \mid L_r}+\frac{1}{n^2} \leq \frac{1}{0.9n}+\frac{1}{n^2} \leq \frac{1}{0.8n}\,,
\end{align}
where $\Prc{A_{u,v} \mid L_r} \leq 1/0.9n$ since $u$ can choose $v$ only after a death of one of its neighbours, being $v$   younger than $u$. 

We   now analyze the case in which $v$ is older than $u$, where $u$ is born at step $T_{r-i}$.
For the law of total probability,    
\begin{align}
\label{eq:A_u,v_cond}
    \Prc{A_{u,v}} \leq \Prc{A_{u,v} \mid L_r}+\frac{1}{n^2}\,.
\end{align}
So, the next step is to evaluate    $\Prc{A_{u,v} \mid L_r}$. For each $k \geq 1$ and $w \in N_{T_k}$,   define the following event:
\begin{align}
    D_{w,k}=\{w \text{ dies at time } T_k\} \, .
\end{align}
To bound   $\Prc{D_{w,k} \mid L_r}$, for each $k=r-i,\dots,r$ and $w \in N_{T_k}$, we use   Lemma \ref{lem:N_m+1_and_death} to get  $\Prc{D_{w,k}} \leq 1/(1.8n)$, and, hence,   for the Bayes' rule,
\begin{align}
    \Prc{D_{w,k} \mid L_r}= \frac{\Prc{D_{w,k} \cap L_r}}{\Prc{L_r}}\leq \frac{\Prc{D_{w,k}}}{1-1/n^2}=\frac{1/1.8n}{1-1/n^2} \leq \frac{1}{1.7n}\, .
    \label{eq:prob_v_dies_cond}
\end{align}
Now, for each $j=r-i, \dots, r$,   define the following events
\begin{equation}
    A_{u,v}^{j}=\{\text{a fixed request of $u$ connects   to $v$ at time $T_j$}\}\, ,
\end{equation}
and     write  $A_{u,v}=\cup_{j=r-i}^r A_{u,v}^j$. Notice that there is some difference between the probability distribution of $A_{u,v}^{r-i}$ and that of $A_{u,v}^j$ for each  $j >r-i$. Indeed, it holds that 
\begin{equation}
    \Prc{A_{u,v}^{r-i} \mid L_r} \leq \frac{1}{0.9n}\,,
    \label{eq:bound_A_u,v^1}
\end{equation}
since this is the probability that the request of $u$ has   destination $v$ at the time of $u$'s arrival (since $v$ is older than $u$). On the other hand, for each $j=r-i+1,\dots,r$,  thanks to  the memoryless property 
of the exponential distribution, 
\begin{equation}
\label{eq:bound_A_u,v^j}
    \Prc{A_{u,v}^j \mid L_r} \leq 1 \cdot \frac{1}{1.7n}\cdot \frac{1}{0.9n} \,.
\end{equation}
The above bound holds since   any fixed request of $u$ can choose $v$ as destination at round $T_j$ only if, at round $T_{j-1}$, $u$ is not connected to $v$. So, the first factor $1$ in the r.h.s. of \eqref{eq:bound_A_u,v^j} is an  upper bound on the probability that, at time $T_{j-1}$, $u$ is not connected to $v$. The second factor, $1/(1.7n)$, is the upper bound on the probability (conditional to $L_r$ from \eqref{eq:prob_v_dies_cond}) that the node to which $u$ is connected dies at time $T_j$. Moreover,  $1/(0.9n)$ is the probability, conditional to $L_r$, that the request of $u$ connects  to $v$ at time $T_j$, if its neighbour is died at time $T_j$.

So, recalling that $A_{u,v}=\cup_{j=r-i}^{r}A_{u,v}^j$, from  \eqref{eq:bound_A_u,v^1} and \eqref{eq:bound_A_u,v^j},  
\begin{equation}
\label{eq:A_u,v_cond_2}
    \Prc{A_{u,v} \mid L_r} \leq  \sum_{j=r-i}^{r}\Prc{A_{u,v}^j \mid L_r} \leq  \frac{1}{0.9n}\left(1+\frac{i}{1.7n}\right)\,.
\end{equation}
Finally, since conditional to $L_r$ we have that $i \leq 7n \log n$, using  \eqref{eq:A_u,v_cond_2} into  \eqref{eq:A_u,v_cond},   the proof is completed.

\subsubsection{Proof of Lemma \ref{lem:poi:exp:small1}}
\label{ssec:lem:poi:exp:small1}

We proceed as in the proof of the analogous lemma in the \SDGE \ model (Lemma \ref{lem:expansion_small_streaming}): we want to show that   two disjoint sets $S,T \subseteq N_t$, with $|S| \leq n/\log^2n$ and $|T|=0.1|S|$, such that $\partial_{out}(S) \subseteq T$, exist with negligible probability.

We recall the definition
\begin{equation}
    A_{S,T}=\{\partial_{out}(S) \subseteq T\}\,.
\end{equation} 
 Then,  as for  the event
\begin{equation}
L_r=\{\text{each node in $N_{T_r}$ is born after time $T_{r-7n \log n}$}\}\cap\{|N_{T_r}|\in [0.9n,1.1n]\}\, ,
\end{equation}
 from  Lemma \ref{thm:concentration_nodes} and Lemma \ref{lem:life_of_nodes}, we obtain   $\Prc{L_r}\geq 1-1/n^2$.
So, for the law of total probability,  
\begin{align}
    \Prc{\min_{0 \leq |S| \leq n/\log^2n}\frac{|\partial_{out}(S)|}{|S|}\leq 0.1} \leq \sum_{\substack{|S|\leq n/\log^2n \\ |T|=0.1|S|}}\Prc{A_{S,T} \mid L_r} +\frac{1}{n^2} \, .
    \label{eq:main_prob_exp1}
\end{align}

The next step of the proof is to upper bound $\Prc{A_{S,T} \mid L_r}$.
From Lemma \ref{lem:existence_edge_hyp}, since $L_r$ implies that all the active nodes  were born after time $T_{r-7n \log n}$,
\begin{equation}
\label{eq:A_s,t_mid_Lr}
    \Prc{A_{S,T} \mid L_r} \leq \left( \frac{|S\cup T|}{0.8n}\left(1+\frac{7n\log n}{1.7n}\right)\right)^{d|S|} \leq \left( \frac{|S\cup T|}{0.8n}\left(1+5n\log n\right)\right)^{d|S|} \,.
\end{equation}
Notice that, since $|S| \leq n/\log^2n$, the above equation offers a   sufficiently small bound. 
So, combining  \eqref{eq:A_s,t_mid_Lr} with  \eqref{eq:main_prob_exp1}, we obtain
\begin{align}
 &\Prc{\min_{0 \leq |S| \leq n/\log^2n}\frac{|\partial_{out}(S)|}{|S|} \leq 0.1}\leq \sum_{s=1}^{n/\log^2n}\binom{1.1n}{s}\binom{1.1n-s}{0.1s}\left(\frac{1.1s}{0.8n}(1+5n \log n)\right)^{ds}+\frac{1}{n^2} \, .
 \label{eq:final_bound_exp1}
\end{align}
In the equation above, we bounded each binomial coefficient with the inequality $\binom{n}{k} \leq \left(\frac{n\cdot e}{k}\right)^k$ for each $k \leq n$ and $n \geq 2$. Then,   by calculating the derivative of the function $f(s)$ that represents each term of the sum, we derive that       each of such  terms reaches its maximum at the extremes, i.e.   in $s=1$ or in $s=n/\log^2n$. So, we  get that the sum in \eqref{eq:final_bound_exp1}, if $d \geq 35$, is bounded by $2/n^2$.

\subsubsection{Proof of Theorem \ref{thm:flooding_terminates_poisson}} \label{sssec:thm:flooding_terminates_poisson}
As remarked in Subsection \ref{sec:Poisson}, the discretized version of the flooding process (Definition \ref{def:flooding_poisson})   is always slower than the original one (Definition \ref{def.async.flood}), so  we can  analyze the former version   along three consecutive phases.

\smallskip
\noindent \emph{Phase 1: The Boostrap.}
  The first phase lasts until   the source information 
  reaches  a subset of size $n^{\epsilon}$, for some 
  constant $\epsilon < 1$ (in our analysis we fix $\epsilon  = 1/10$).

\begin{lemma}[Phase 1: The Bootstrap]
\label{lem:flooding_poisson_first}
Under the hypotheses of Theorem \ref{thm:flooding_terminates_poisson} there is a  $\tau_1 = \bigO(\log n)$ such that, w.h.p.
\[|I_{t_0 + \tau_1}|\geq n^{1/10}\,.\]
\end{lemma}
\begin{proof}
Let $I_t$ be the set of informed nodes, with $t \geq t_0=T_{r_0}$ and $|I_t| \leq n^{1/10}$. We want to prove that, w.h.p., $I_t = I_t \cap N_{t+1}$, i.e. all the nodes in $I_t$ survive for a time interval equal to $1$. Since the life of a node follows an exponential distribution of parameter $1/n$, and since $|I_t| \leq n^{1/10}$, this event has probability $e^{-|I_t|/n} \geq 1-1/n^{9/10}$. According to  Definition \ref{def:flooding_poisson}, we let $I_{t+1}=(I_t \cup \partial_{out}^{t}(I_{t} \cap N_{t+1})) \cap N_{t+1}$ be the set of informed nodes at time $t+1$. Since the graph $G_t$ is an expander of parameter $0.1$, w.h.p. (Theorem \ref{thm:exp:pdge}) and since $I_t \cap N_{t+1}=I_t$ w.h.p., it holds w.h.p.
\[|\partial_{out}^t(I_{t} \cap N_{t+1})| \geq 0.1 |I_{t}|\,.\]
Since $|I_t| \leq n^{1/10}$, all the nodes in $\partial_{out}^{t}(I_t \cap N_{t+1})$ survive for a time interval equal to $1$ with probability $e^{-0.1|I_t|/n} \geq 1-1/n^{9/10}$, so w.h.p.
\[|I_{t+1}| \geq |(\partial_{out}^t(I_{t} \cap N_{t+1})) \cap N_{t+1}| \geq 1.1|I_t|\,.\]
It follows that, after a phase of length   $\tau_1=\bigO(\log n)$, we get    $|I_{t_0+\tau_1}|\geq n^{1/10}$,  w.h.p.
\end{proof}

\smallskip
\noindent \emph{Phase 2: Exponential growth of the informed nodes.}
   In the next lemma, we show that, after the bootstrap,  the flooding process yields an exponential increase of the number of informed nodes until it reaching  half of the nodes in the network.

\begin{lemma}[Phase 2]
\label{lem:flooding_poisson_second}
Under the same hypotheses of Theorem \ref{thm:flooding_terminates_poisson}, there is a  $\tau_2 = \bigO(\log n)$ such that, for $\tau_1=\bigO(\log n)$ (as in Lemma \ref{lem:flooding_poisson_first}), w.h.p.
\[|I_{t_0 + \tau_1+\tau_2}|\geq \frac{|N_{t_0 +\tau_1+ \tau_2}|}{2}\,.\]
\end{lemma}

\begin{proof}
Observe first that in any interval of time equal to $1$, w.h.p. at most $2\log n$ nodes leave the network. Indeed, the number of nodes that leave the network in the  time interval $[t,t+1]$ is a random variable
\[D=\sum_{v \in N_t}D_v\,,\]
where each $D_v$ is a  Bernoulli random variable, such that $\Prc{D_v=1}=1-e^{-1/n}$,  which indicates if the node $v \in N_t$ leaves the network before $t+1$. So, from  Lemma \ref{thm:concentration_nodes} and  the Chernoff Bound (Theorem \ref{thm:chernoff}),  
\begin{equation}
    \Prc{D \geq 2 \log n} \leq \frac{1}{n^{1/3}}\,.
\end{equation}

We recall that the set of infected nodes at time $t+1$ is $I_{t+1}=(I_t \cup \partial_{out}^t(I_t \cap N_{t+1}) \cap N_{t+1}$. Since we have shown that, in the interval $[t,t+1]$, at most $2 \log n$ nodes leave the network w.h.p. and since the graph $G_t$ is a expander with parameter $0.1$ w.h.p. (Theorem  \ref{thm:exp:pdge}), it holds, w.h.p.,
\[|I_{t+1}|\geq |I_t| + 0.1\left(|I_t|-2\log n\right)-2 \log n\,.\]
So, for each $t \geq t_0+\tau_1$, since $|I_{t_0+\tau_1}| \geq n^{1/10}$, w.h.p.
\[|I_{t+1}| \geq 1.09|I_t|\,.\]
We thus have  an exponential growth of the set of the informed nodes and, so,   there exists $\tau_2= \bigO(\log n)$ such that $|I_{t_0+\tau_1+\tau_2}| \geq |N_{t_0+\tau_1+\tau_2}|/2$, w.h.p.
\end{proof}

 \smallskip
\noindent \emph{Phase 3: Exponential decrease of the non-informed nodes.} The analysis of this phase considers  
the subset $S_t \subseteq N_t$ of the non-informed nodes. More precisely, we prove that $S_{t+1}$ w.h.p. decreases by a constant factor despite the node churn.

\begin{lemma}[Phase 3]
\label{lem:flooding_poisson_third}
Under the same hypotheses of Theorem \ref{thm:flooding_terminates_poisson}, there is a $\tau_3 = \bigO(\log n)$ such that, for  $\tau_1=\bigO(\log n)$ (as in Lemma \ref{lem:flooding_poisson_first}) and $\tau_2=\bigO(\log n)$ (as in Lemma \ref{lem:flooding_poisson_second}), we have  w.h.p.
\[I_{t_0+\tau_1+\tau_2+\tau_3}=N_{t_0 +\tau_1+\tau_2+\tau_3}\,.\]
\end{lemma}

\begin{proof}
To prove this lemma, we will consider the set $S_t \subseteq N_t$ of non informed nodes at time $t$, i.e. $S_t=N_t-I_t$. Notice that, since every node $v$ in $\partial_{out}^{t+1}(S_{t+1}) \subseteq I_{t+1}$ is reachable in 1-hop to the set of non-informed nodes at time $t+1$, $v$ was not informed at time $t$. This implies that
\[\partial_{out}^{t+1}(S_{t+1}) \subseteq (S_t-S_{t+1}) \cap N_{t+1}\,.\]
Consider the random variable $J_{t,t+1}$ that indicates the number of nodes that join the network in the time interval $[t,t+1]$. Then, the above consideration implies that
\begin{equation}
    |S_{t}|-|S_{t+1}|+J_{t,t+1} \geq |\partial_{out}^{t+1}(S_{t+1})|.
    \label{eq:S_t-S_t+1_in_partial}
\end{equation}
Since $J_{t,t+1}$ is a Poisson random variable of parameter $1$, for the tail bound for the Poisson distribution (Theorem \ref{thm:tail_bound_poisson_distribution}),    $J_{t,t+1} \leq \log n$, w.h.p.
Since $|S_{t_0+\tau_1+\tau_2}| \leq |N_{t_0+\tau_1+\tau_2}|/2$ w.h.p., from the expansion of the graph $G_{t+1}$ (Theorem \ref{thm:exp:pdge}) it holds w.h.p. that, for each $t \geq t_0+\tau_1+\tau_2$,
\[|S_{t+1}|\leq \frac{1}{1.1}\left(|S_t|+\log n\right).\]
Then, a time  $\tau_3'=\bigO(\log n)$ exists such that $|S_{t_0+\tau_1+\tau_2+\tau_3'}| \leq \log^2n$.

After the process reaches the above small number of non-informed nodes, we   consider
the set of non-informed nodes at time $t$ without including   the set of nodes that join the network after time $t_0+\tau_1+\tau_2+\tau_3'$: we call the latter set as $S_{t}^*$, for each $t \geq t_0+\tau_1+\tau_2+\tau_3'$. As in the first part of the proof, we get that $\partial_{out}^{t+1}(S_{t+1}^*) \subseteq (S_{t}^*-S_{t+1}^*)\cap N_{t+1}$. Since the graph $G_{t+1}$ is an expander w.h.p. (Theorem \ref{thm:exp:pdge}), we have that w.h.p.
\[|S_{t+1}^*| \leq \frac{1}{1.1}|S_{t}^*|.\]
Since $|S_{t_0+\tau_1+\tau_2+\tau_3'}^*| \leq \log^2n$, there is a $\tau_3=\bigO(\log n)$ such that, $|S_{t_0+\tau_1+\tau_2+\tau_3}^*|<1$ w.h.p.

In conclusion, let $J_{\tau_3',\tau_3}$ be the number of nodes that join the network from time $t_0+\tau_1+\tau_2+\tau_3'$ to time $t_0+\tau_1+\tau_2+\tau_3$. Since the arrival of the nodes during an interval of length $\tau_3-\tau_3'$ is a Poisson process of mean $\tau_3-\tau_3'$, for the tail bound for the Poisson distribution (Theorem \ref{thm:tail_bound_poisson_distribution}), w.h.p. $J_{\tau_3',\tau_3}=\bigO(\log n)$. Moreover, from Lemma \ref{thm:concentration_nodes}, each of these new nodes connect to the set of informed nodes with probability at least $(1-(2\log^2n/n)^d)(1-1/n^2)$. Moreover, each informed node to which the new nodes have connected survive for the 1-hop  transmission with probability $e^{-1/n}\geq 1-\frac{1}{n}$. So,  each   node that joins the network after time $t_0+\tau_1+\tau_2+\tau_3'$ gets informed within time $t_0+\tau_1+\tau_2+\tau_3$, w.h.p.
\end{proof}

\bibliographystyle{plain}
\bibliography{njl}

\newpage
\appendix
\begin{center}
\LARGE{\textbf{Appendix}}
\end{center}

\section{Mathematical tools} \label{apx:sec:tools}

\begin{theorem}[Chernoff Bound, \cite{dubhashipanconesi09}]
\label{thm:chernoff}
Let $X_1,\dots,X_n$ be independent Poisson trials such that $\Prc{X_i=1}=p_i$. Let $X=\sum_{i=1}^n X_i$, $\mu=\Expcc{X}$ and suppose $\mu_L \leq \mu \leq \mu_H$. Then, for all $0 <\varepsilon \leq 1$ the following Chernoff bounds hold
\begin{equation}
\Prc{X \geq (1+\varepsilon)\mu_H} \leq e^{-\frac{\varepsilon^2}{3}\mu_H}
\end{equation}
\begin{equation}
    \Prc{X \leq (1-\varepsilon)\mu_L} \leq e^{-\frac{\varepsilon^2}{2}\mu_L}\,.
\end{equation}
\end{theorem}
\begin{theorem}[Method of bounded differences,\cite{dubhashipanconesi09}]
\label{thm:bounded_differences}
Let $\mathbf{Y}=(Y_1,\dots,Y_m)$ be independent random variables, with $Y_j$
taking values in the set $A_j$. Suppose the real-valued function $f$ defined on
$\prod_j A_j$ satisfies the Lipschitz condition with coefficients $\beta_j$,
i.e.
\[
|f(\mathbf{y})-f(\mathbf{y'})| \leq \beta_j
\]
whenever vectors $\mathbf{y}$ $\mathbf{y'}$ differs only in the $j$-th
coordinate. Then, for any $M>0$, it holds that
\[
\Prc{f(\mathbf{Y}) \geq \Expcc{f(\mathbf{Y})} + M }
\leq e^{-\frac{2M^2}{\sum_{j=1}^{m}\beta_j^2}}\, ,
\]
and
\[
\Prc{f(\mathbf{Y}) \leq \Expcc{f(\mathbf{Y})} - M }
\leq e^{-\frac{2M^2}{\sum_{j=1}^{m}\beta_j^2}}\,.
\]
\end{theorem}

\begin{theorem}[Kullback-Leibler divergence inequality]
\label{thm:kullback_inequality}
Let $p_m$ and $q_m$ be two discrete probability mass functions, with $m \in \{1,\dots,L\}$. We have that
\begin{equation}
    \sum_{r=1}^L p_m \log_2 \left(\frac{p_m}{q_m}\right) \geq 0 \,.
\end{equation}
\end{theorem}

\section{Static Random Graphs}
\label{ssec:lemma:staticdoutgoing}

\begin{lemma}
\label{lemma:staticdoutgoing}
The static random graph in which each node picks $d$ random neighbors is a $\Theta(1)$-expander w.h.p., for each $d \geq 3$.
\end{lemma}
\begin{proof}
We consider the static random graph $G=(N,E)$  and $S \subset N$ a subset of the nodes with $|S|=s$. Let $T \subseteq N-S$ be an arbitrary set disjointed from $S$, with $|T|=0.1s$. In this model we know that an edge starting from a node $v$ has destination $u$ with probability $\frac{1}{n-1}$. We know that the probability that the edges from $S$ are in $S \cup T$
\[\left(\frac{|S \cup T |}{n-1}\right)^{d|S|}\,.\]
So, the probability that the outer boundary of $S$ is at most in $T$ is
\[\Prc{\partial_{out}(S)\subseteq T} \leq \left(\frac{1.1s}{n-1}\right)^{ds}.\]
From an union bound over all the set $T$ disjointed with $S$ and with $|T|=0.1s$, all the set $S$ with $s$ elements and all the possible sizes $s=1,\dots,n/2$ of $s$ we get
\begin{equation}
\label{eq:pr_not_expander}
    \Prc{\text{$G$ is not an expander}} \leq \sum_{s=1}^{n/2} \binom{n}{s} \binom{n-k}{0.1s} \left(\frac{1.1s}{n-1}\right)^{ds}\,.
\end{equation}
From standard calculus, it can be proved that, for $d \geq 3$, the equation above is upper bounded by $1/n^{d-2}$. This is obtained by bounding each binomial coefficient with the bound $\binom{n}{k}\leq \left(\frac{n \cdot e}{k}\right)^k$ and by computing the derivative of the function $f(s)$ (representing each term of the sum), obtaining that each of these terms attained its maximum in $s=1$ or in $s=n/2$.
\end{proof}

\section{Useful Tools for the Poisson Models}

\begin{definition}[Counting process]
The stochastic process $\{X(t),t \geq 0\}$ is said to be a \emph{counting process} if $X(t)$ represents the total number of events which have occurred up to time $t$.
\end{definition}

\begin{definition}
Let $\{X(t),t \geq 0\}$ be a counting process. It is a \emph{Poisson process} if
\begin{enumerate}
    \item $X(0)=0$;
    \item $X(t)$ has independent increments;
    \item The number of events in any interval of length $t$ has a Poisson distribution with mean $\lambda t$. That is, for all $s,t \geq 0$
    \[\textbf{Pr} \left(X(t+s)-X(s)=n\right)=e^{-\lambda t}\frac{(\lambda t)^n}{n!} \quad n \geq 0 \, .\]
\end{enumerate}
\end{definition}

\begin{theorem}
Let $\{X(t),t \geq 0\}$ be a Poisson process. Then, given $X(t)=n$, the $n$ arrival times $S_1,\dots,S_n$ have the same distribution as the order statistics corresponding to $n$ independent random variables uniformly distributed in the interval $(0,t)$.
\label{thm:arrival_prob}
\end{theorem}

\begin{theorem}[Tail bound for the Poisson distribution]
Let $X$ have a Poisson distribution with mean $\lambda$. Then, for each $\varepsilon>0$,
\begin{equation}
    \textbf{Pr} \left(\left|X-\lambda \right| \geq x\right) \leq 2e^{-\frac{x^2}{2(\lambda+x)}}
\end{equation}
\label{thm:tail_bound_poisson_distribution}
\end{theorem}

\begin{theorem}[\cite{norrismarkov98}]
\label{thm:minimum_exp}
Let $I$ be a countable set and let $T_k$, $k \in I$, be independent exponential random variables of parameter $q_k$. Let $0<q=\sum_{j \in I}q_k \leq \infty$. Set $T=\inf_k T_k$. Then this infimum is attained at a unique random value $K$ of $k$, with probability $1$. Moreover, $T$ and $K$ are independent with $T$ exponential of parameter $q$ and $\Prc{K=k}=q_k/q$.
\end{theorem}

\end{document}